\documentclass[a4paper,11pt]{article}

\usepackage[usenames,dvipsnames]{color}

\usepackage{fullpage}

\usepackage{amsmath}
\usepackage{amsfonts}
\usepackage{amssymb}

\usepackage{amsthm}
\usepackage{mathrsfs}

\usepackage{pictexwd,dcpic}
\usepackage{float}
\usepackage{color}
\usepackage{graphicx}
\usepackage{tikz-cd}
\usetikzlibrary{matrix,arrows,calc}

\tikzset{ampersand replacement=\&}

\usepackage{comment}

\usepackage[colorlinks=true,linkcolor=blue,citecolor=red,urlcolor=Violet,bookmarksdepth=subsection,final]{hyperref}

\theoremstyle{plain} 
\newtheorem{thm}{Theorem}[section]
\newtheorem{defn}{Definition}[section]
\newtheorem{lem}[thm]{Lemma}
\newtheorem{prop}[thm]{Proposition}

\theoremstyle{remark}
\newtheorem{rem}{Remark}

\numberwithin{equation}{section}

\newcommand{\de}{\mathrm{d}} 

\newcommand{\oo}{\infty}
\newcommand{\del}{\partial}
\renewcommand{\d}{\mathrm{d}}

\newcommand{\Lie}{\mathcal{L}}
\newcommand{\sse}{\subseteq}

\newcommand{\tr}{\operatorname{tr}}

\title{\textbf{On Wick polynomials of boson fields in locally covariant algebraic QFT}}
\author{Igor Khavkine$^{(a)}$\footnote{igor.khavkine@unimi.it}, Alberto Melati$^{(b)}$\footnote{alberto.melati@unitn.it}, Valter Moretti$^{(c)}$\footnote{valter.moretti@unitn.it}\\
	$\null$\\
	\small
	$^{(a)}$ Dipartimento di Matematica,
	Universit\`a di Milano  and INFN Milano,\\
	\small
	Via Cesare Saldini, 50,
	I-20133 Milano (MI), Italy \\
	\small
	$^{(b)}$ Dipartimento di Fisica,
	Universit\`a di Trento  and INFN-TIFPA Trento,\\
	\small
	Via Sommarive, 14, I-38123 Povo (Trento), Italy\\
	\small 
	$^{(c)}$	Dipartimento di Matematica,
	Universit\`a di Trento  and INFN-TIFPA Trento,\\
	\small
	Via Sommarive, 14, I-38123 Povo (Trento), Italy}

\begin{document}
	\maketitle
	\begin{abstract}
		This work presents some results about Wick polynomials of a vector field renormalization in locally covariant algebraic quantum field theory in curved spacetime. General vector fields are pictured as sections of natural vector bundles over globally hyperbolic spacetimes and quantized through the known functorial machinery in terms of local $*$-algebras. These quantized fields may be defined on spacetimes with given classical background fields, also sections of natural vector bundles, in addition to the Lorentzian metric. The mass and the coupling constants are in particular viewed  as background fields. Wick powers of the quantized vector field are axiomatically defined 
		imposing in particular local covariance, scaling properties and smooth dependence on smooth perturbation of the background fields.
		A general classification theorem is established for finite renormalization terms (or counterterms) arising when comparing different solutions satisfying the defining axioms of Wick powers. The result is specialized to the case of general tensor fields. In particular, the case of a vector Klein-Gordon field and the case of a scalar field renormalized together with its derivatives are discussed as examples. In each case, a more precise statement about the structure of the counterterms is proved. The finite renormalization terms turn out to be finite-order polynomials tensorially and locally constructed with the backgrounds fields and their covariant derivatives whose coefficients are locally smooth functions of polynomial scalar invariants constructed from the so-called marginal subset of the background fields. The notion of local smooth dependence on polynomial scalar invariants is made precise in the text.
		
		Our main technical tools are based on the Peetre-Slov\'ak theorem characterizing differential operators and on the classification of smooth invariants on representations of reductive Lie groups.
	\end{abstract}

\section{Introduction}
This work is a continuation of the work started in the previous article~\cite{KM16} by the first and last author. While we have aimed the current article to be self-contained, the reader may be referred to the previous article for the details of some proofs.

\subsection{Wick polynomials}
Wick polynomials and time-ordered products of Wick polynomial are the building blocks for perturbative renormalization  of quantum fields, both  in Minkowski spacetime and in curved spacetime, where the metric is considered as a given external classical field.
Although of utmost physical  relevance, \emph{e.g.},\ the stress-energy tensor is a Wick polynomial and it plays a most important part in semiclassical quantum gravity (see~\cite[Sec.6.3]{AAQFT15Ch6} for some cosmological applications),
these formal operators do not belong to the algebra of observables generated by the smoothly smeared field operators (operator-valued distributions). This is  because 
they correspond to {\em products of distributions at a given point}  and this notion  is not well-defined in general. As an elementary
example for the Klein-Gordon scalar field operator $\phi(x)$ on a spacetime $(M,g)$ viewed as formal integral kernel of a field operator represented on a Fock Hilbert space smeared with 
smooth compactly supported functions $f$, namely  $\phi(f):= \int_M \phi(x) f(x) \sqrt{\left|\det g\right|}\,dx$,
the simplest Wick power is a suitable interpretation of $\phi^2(x) = \phi(x)\phi(x)$. It stands for
the integral kernel $\phi^2(f) = \int_M \phi^2(x) f(x) \sqrt{\left|\det g\right|}\,dx$.
However this interpretation is very difficult to support. For example, any naive attempt to define it (with $\langle~|~\rangle$ the Fock inner product and $\Psi$ any state in the appropriate domain) as the limit for $n \to \infty$ of
\begin{equation*}
	\int_{M\times M}\langle \Psi| \phi(x)\phi (y)\Psi \rangle  f(x) \delta_n(x,y) \sqrt{\left|\det g(x)\right|}\sqrt{\left|\det g(y)\right|}\,dy\, dx \:,
\end{equation*}
where $\delta_n(x,y) \to \delta(x,y)$ as $n \to \infty$,  gives rise to divergences for any physically meaningful state $\Psi$, such as the Fock vacuum in Minkowski spacetime. 
The popular and perhaps most effective procedure to eliminate the short-distance divergences consists of simply keep a regulated smearing function $\delta_n(x,y)$ and simply subtracting a suitable divergent function of $n$ as $n\to \oo$, that is the regulator is removed. A much more elegant procedure (see~\cite{AAQFT15Ch5} for a recent introductory account) consists of first restricting ourselves to a suitable class of 
physical states (Hadamard states). For any Hadamard state $\Phi$, the singularity structure of (more precisely the wavefront set of) $\langle \Phi|\phi(x)\phi(y) \Phi \rangle$ is under sufficient control so that we can find a distribution $G(x,y)$ that is independent of $\Phi$ (as long as it remains Hadamard) such that the difference $\langle \Phi|\phi(x)\phi(y) \Phi \rangle - G(x,y)$ is regular enough to be smeared with some distributions and $f(x)\delta(x,y)$ in particular, with a test function $f(x)$. Thus, we could formally define the Wick square by $\phi^2(x) := \lim_{y\to x} \phi(x)\phi(y) - G(x,y) 1$, and so on for higher Wick powers (this is known as the Hadamard parametrix regularization method~\cite{HW1,HW2,BF00}). In any case, even such a procedure do not lead to a unique definition. The constructed Wick powers (or also more generally time ordered products of Wick powers) may be still affected by (finite) ambiguities, popularly called {\em finite-renormalization terms} (or \emph{counterterms}). Within the divergence subtraction paradigm their nature is obvious: depending on how the regularization is carried out, $\infty -\infty$ could be any number.

A given Wick product $\phi^n(x)$, interpreted as a distributional kernel evaluated at $x$, can  always be be redefined by adding similar counterterms of lower order multiplied 
with coefficients depending on $x$:  $\phi^n(x) = \sum_{k<n} C_k(x) \phi^k(x)$. The structure of these coefficients  $C_k$ can be fixed 
by imposing some further physical constraints.

A definite difference exists between flat and curved spacetime renormalization procedures defining 
Wick polynomials (we will not discuss time ordered products in this work). In Minkowski spacetime, Wick polynomials (though not time ordered products Wick polynomials) are completely 
fixed by the so called {\em normal-ordering prescription} which is feasible because there exist a unique Poincar\'e invariant reference state and all the 
manipulations are in fact performed in the Fock-Hilbert space, relying upon that vacuum state. The
normal-ordering prescription is able to  simultaneously get rid of ultraviolet divergences and fix all remaining finite renormalization ambiguities 
of Wick polynomials in Minkowski spacetime. 
From our viewpoint it consists of removing the ultraviolet singularities and imposing that the expectation values of the obtained operators vanish on the  unique 
Poincar\'e invariant state. 
Unfortunately no such preferred reference state exists in generic curved spacetime, though a viewpoint similar to the Minkowskian one may be adopted dealing with
maximally symmetric spacetimes like de Sitter spacetime, where a natural and unique notion of symmetry-invariant vacuum is available (at least for massive fields~\cite{AF87}).  In the absence of a 
sufficiently large group of Killing symmetries able to single out a physically privileged reference vacuum state, 
or for specific values of the parameters defining the quantum field (think of a Klein-Gordon field in Minkowski space time with $m^2<0$),
the Minkowskian procedure cannot be 
adopted  to completely fix the definition of Wick polynomial even if the normal ordering prescription  is mathematically meaningful in  the  Fock space of every Gaussian (or quasi-free) state.

Still, a relic of the Minkowskian short-distance
divergence remains in generic curved spacetime, encoded in the universal Hadamard short distance divergence of $n$-point functions, mathematically  corresponding to a definite structure 
of the wavefront set of these distributions, in the language of microlocal analysis (\emph{e.g.},\ see~\cite{AAQFT15Ch5}). However removing this divergence is by
no means sufficient to uniquely define Wick powers. Ambiguities remain and the best we can do is to reduce them to the smallest possible number of types and to classify them. 

Let us briefly and heuristically describe how these ambiguities have been studied in previous works.
In addition to the obvious requirement that any general procedure should give rise to the known result in Minkowski spacetime, the general strategy~\cite{HW1}
is to avoid any specific choice from scratch: first of all, no preferred Hilbert space representation of operators is chosen and all the discussion 
takes place at the level of abstract $*$-algebras of operators $\phi_{(M,\mathbf{b})}(x)$,
where $(M,\mathbf{b})$ denotes a spacetime endowed with a set of background fields, like the metric and the mass generally, both allowed to vary on $M$.
Quantum states are {\em algebraic states}, namely positive normalized linear functionals over the said unital $*$-algebra of observables. The standard Hilbert space representation arises via the {\em GNS theorem}~\cite[Sec.5.1.3]{AAQFT15Ch5}.
Next the notion of Wick polynomial is required  to be consistent with the requirements of locality and covariance~\cite{BFV}. In other words, Wick polynomials (and all other observables too)
must be equivariant with respect to causal embeddings of different spacetimes (locality) and with respect to causal diffeomorphisms of a given yet arbitrary 
spacetime (covariance). This is not enough however to fix the Wick polynomials completely  and, barring obvious technical requirements, further natural constraints must be added~\cite{HW1},
like the behavior of Wick polynomials under rescaling of the metric and the other given background fields, which together allow a satisfactory classification of the remaining freedom in finite renormalizations. Unfortunately, one of the technical requirements from~\cite{HW1}, so-called ``analytic dependence on analytic metrics'', was long considered unnatural, despite its crucial role in the classification result.

An important refinement of these results was
recently made (by some of the authors) for the real scalar field~\cite{KM16}, where ``analytic dependence'' was replaced by the more natural ``smooth dependence'' of the Wick polynomial on smooth compactly localized perturbations of external fields. 
As a matter of fact one requires that smooth deformations of the background fields  $\mathbf{b} \mapsto \mathbf{b}(s)$, where $s \in \mathbb R$, produces smooth fields $(x,s) \mapsto \phi^n_{(M,\mathbf{b}(s))}(x)$,
with a suitable precise interpretation of these mathematical objects.
This smooth dependence requirement together with locality allowed us to replace all appeals to analyticity in the proof of the classification result by an appeal to the \emph{Peetre-Slov\'ak} theorem~\cite{slovak}.
This fundamental result in differential geometry states that a (possibly non-linear) map between spaces of smooth functions is a differential operator precisely when it is local (depends only on the germ of its argument) and maps smooth families of smooth functions to smooth families of smooth functions. The latter regularity condition is what inspired our ``smooth dependence'' requirement.
This argument, applied in the case of the Klein-Gordon scalar field in~\cite{KM16}, was used to show that the renormalization coefficients $C_k(x)$ are locally given by differential operators applied to the background fields $\mathbf{b}(x) = (g,m^2,\xi)$, which consist of the metric, the mass squared and the curvature coupling (cf.~Section~\ref{proca_section}).  The further requirements of locality, covariance and scaling finally restricts the functions $C_k(x)$ to be scalar 
polynomials in the mass $m^2$ and any other scalars covariantly built out of the Riemann tensor $R_{abcd}$ and its covariant derivatives, with the coefficients of 
these polynomials given by arbitrary smooth functions of the curvature coupling $\xi$. The difference between polynomial and smooth dependence on some of the parameters is explained by their different properties under scaling transformations. This classification result had reproduced the earlier results of~\cite{HW1}, modulo their analytic dependence on $\xi$ \emph{vs.}\ our smooth dependence on $\xi$, though with more natural hypotheses.

\subsection{Structure of the work and main results}
This work deals with the classification of Wick polynomials of a rather general locally covariant bosonic vector-valued quantum field, in the presence of rather general classical background fields.
The constant parameters that usually defining a quantum field, 
like the mass and coupling constants, are included in among the classical background fields (and may be restricted to be constants). Fermionic fields (like the locally covariant Dirac field~\cite{zahn14}) are not handled by our analysis and will be discussed elsewhere.

Our main results are split in two. The first (Theorem~\ref{lemma_ipotesi}) holds when both the dynamical and background fields are sections natural vector bundle over spacetime (\emph{natural} here means that the bundle transforms in a well-defined way under diffeomorphisms, cf.~\cite[Ch.IV]{kms}). It gives the structure of the general form of the finite renormalizations of Wick powers, which are parametrized by differential operators locally depending on the background fields, including the non-trivial relations between terms that renormalize Wick powers of different degrees. The simplest examples of natural tensor bundles are trivial bundles, tangent and cotangent bundles, as well as any direct sums and products thereof. But natural bundles also include more general examples like bundles of connection coefficients or even jet bundles themselves.

The second main result (Theorem~\ref{thm_uniqueness}) holds when we restrict the dynamical and background fields to be only tensor fields. It completely classifies the differential operators parametrizing the finite renormalizations of Wick powers to be covariantly constructed from the metric, the curvature, the background tensor fields and all of their covariant derivatives, and furthermore to have a certain polynomial structure of bounded degree.

The above results are the first complete and rigorous ones for non-scalar dynamical and background fields.
To supplement our great generality, we illustrate our results with two physically relevant particular cases: the real vector Klein-Gordon field $A_a$, possibly with tensorial coupling to the curvature, and the $(\varphi, \nabla_a\varphi)$ pair consisting of the real scalar Klein-Gordon field and its spacetime derivative. 

Similarly to~\cite{KM16}, the  more complex final classification theorem of Wick polynomials arises by assuming that the
Wick products satisfy a certain list of axioms (Definition~\ref{def_lc}) including local covariance, scaling, 
smooth dependence on perturbations of background fields and commutation relations (kinematic completeness). 
The relevant commutation relations are quite general. Since no specific equations of motion are assumed, we only require that the commutator of two linear quantum fields is a $c$-number distribution satisfying a certain regularity condition (which is known to hold for usual causal propagator of a well-posed hyperbolic equation).

Except for some general remarks, we shall not discuss the existence of Wick products satisfying our set of axioms, same as in~\cite{KM16}.
We expect that, at least for the Klein-Gordon field and related fields (\emph{e.g.},\ its derivatives), a proof of existence may be obtained by a 
straightforward re-adaptation to the vector case field of the reasoning appearing in~\cite{HW1} and~\cite{HW2} for the scalar field by taking 
advantage of the characterization of Hadamard states and parametrices for vector field presented in~\cite{sahlmann01}.

The paper is organized as follows, where we also list the main results of the work.

Section~\ref{section_geom} introduces the general notations and the general geometric setup concerning vector bundles and their symmetric tensor powers, also jet bundles, as well as a brief discussion of the
{\em Peetre-Slov\'ak theorem}  stated in a form useful for our purposes.  This section also includes the definition and useful identities for the symmetrized product of sections of a vector bundle. The symmetrized product will be useful for giving index-free versions of various formulas in our results.

Section~\ref{section_algebra_bkg} is devoted to the introduction the concepts of {\em natural vector bundles}, \emph{background geometry}, \emph{scaling} and {\em locally covariant net of
	algebras} of quantum observables on a globally hyperbolic spacetime. The definitions will take advantage of elementary language of category theory and basic notions of operator algebras.

Section~\ref{section_quantum_fields}  deals with the precise notion of a general bosonic {\em locally covariant quantum field}, again with the help of notions from category theory. 
A precise notion of {\em scaling degree} is defined to be used later. The concept of {\em $c$-number field} is also presented.

Section~\ref{section_wick} presents
the general notion of a {\em Wick power} of a general bosonic local quantum field. The precise definition (Definition~\ref{def_wick}) lists all necessary physical requirements: 
behavior of {\em Low Powers}, {\em Scaling}, {\em Kinematic Completeness}, {\em Commutator Expansion}, and {\em Smoothness}. Based on these requirements, we prove our first main result in Theorem~\ref{lemma_ipotesi}.
It establishes that the difference of two different prescriptions for Wick powers with equal order can be expanded as a sum of lower order Wick powers whose coefficients are, by invoking the Peetre-Slov\'ak theorem, tensor fields of suitable scaling degrees given by differential operators (of locally bounded order) locally depending on the background fields. It is difficult to say more about these differential operators without further assumptions.

Next, in Section~\ref{section_tensor_field}, the last result is specialised to the case when both the dynamical and background fields are tensors. After some 
preparatory technical results, the central theorem of this paper (Theorem~\ref{thm_uniqueness}) is stated and proved. It precisely characterizes the form of the differential operator coefficients in the general finite renormalization formula of Theorem~\ref{lemma_ipotesi}. As mentioned earlier, these coefficients must be linear combinations of tensor valued polynomials, covariantly constructed out of the curvature tensors, the background field tensors and all of their covariant derivatives. The number of independent terms and the degrees of these polynomials are \emph{a priori} bounded, with the bound determined by the scaling dimension of the Wick power and the ranks of the tensors involved. The coefficients of these polynomials are \emph{locally} (in a precise sense) smooth functions (no longer just polynomial) of finitely many polynomial scalars covariantly constructed out of the subset of the background fields. Crucially these finiteness results hold only when all background fields are \emph{admissible}. Here a background tensor field is \emph{admissible} if its physical scaling weight and its tensor rank satisfy an inequality (Definition~\ref{def_admissible}). Those background fields that saturate the inequality are called \emph{marginal}%
	\footnote{Similar terminology, \emph{marginal}, \emph{relevant} and
	\emph{irrelevant} fields, appears also in the Wilsonian approach to the
	``renormalization group.'' We emphasize that the similarity is only
	superficial, since both refer to some kind of scaling dimension.
	Note that the Wilsonian terminology refers to dynamical fields
	and only to \emph{physical} scaling, while ours refers mostly to background
	fields and to a combination of \emph{physical} and \emph{coordinate}
	scalings.} %
and only they are allowed to appear non-polynomially in the finite renormalization terms. In comparison with the treatment of the scalar field in~\cite{KM16}, this classification demands much stronger results from the classical invariant theory of the general linear and (Lorentzian) orthogonal groups. The notion of \emph{local} smooth dependence on a set of polynomial invariants (Definition~\ref{def_loc_poly}) was actually born out of the necessity of dealing with the complicated orbit structure for the action of the orthogonal group on background tensors.

The final part of Section~\ref{section_tensor_field} is devoted to two examples of physical relevance that illustrate the various aspects of our classification theorem: the vector Klein-Gordon
field $A_a$ (Section~\ref{proca_section}) and the $(\varphi,\nabla_a\varphi)$ pair consisting of the scalar Klein-Gordon field  
and its spacetime derivative (Section~\ref{sec_grad_scalar})).  \\
The choice to deal with the KG field rather than the Proca one is due to a basic requirement we and previous works have imposed on Wick powers: they must be smooth functions of the 
mass around the zero value. The zero mass $m^2\to 0$ limit is a very delicate issue for Proca field and gauge invariance should be taken into account. A recent paper on the subject is~\cite{2017arXiv170901911S}.
A short discussion on this point appears in  Remark~\ref{rem_KG}. Even the regularity for the KG (vector or scalar) field at zero mass is a delicate matter when Wick powers are constructed by the Hadamard parametrix regularization method, since the corresponding parametrix includes a term of the form $\ln m^2$. Remark~\ref{rem_KG} also contains a brief discussion on this technical problem.

Finally, three appendices collect technical results needed at various stages of the proof of Theorem~\ref{thm_uniqueness}. Appendix~\ref{section_scaling} contains results on (almost) homogeneous functions under scaling. Appendix~\ref{section_coord_scal} contains a convenient version of the Thomas Replacement Theorem, which roughly restricts any tensorial differential operator that is equivariant under diffeomorphisms to depend on the derivatives of the metric only through the Riemann curvature and its covariant derivatives. Appendix~\ref{section_equivar} collects fundamental results on smooth invariants of the (Lorentzian) orthogonal group acting on tensor representations. It should be noted that, even though Appendices~\ref{section_coord_scal} and~\ref{section_equivar} mostly collect results that are known, these results are rather scattered in the expert literature, and their proofs may be difficult to track down. Thus, for the convenience of the reader, we have aimed to provide complete proofs when they could be made reasonably elementary and concise.

\section{Geometric setup}\label{section_geom}

\paragraph{Notations.}
{}
In the following, $VM \to M$ denotes a smooth real vector bundle over a manifold $M$ whose  fibres $V_p$ are isomorphic to a given $\mathbb R^a$. We shall make use of the auxiliary tensor bundles $V^{\otimes k} M \to M$ and $V^{*\otimes l} M \to M$,
which are bundles of tensor products  of $k$ copies of  the bundle $VM$ and $l$ copies of the dual bundle $V^*M$ respectively. 
 In the following we also consider two special sub-bundles, namely those of the fully symmetrized contravariant and covariant tensor products, defined by
\begin{equation*}
V^{\odot k}M = S^kVM \subset V^{\otimes k} M \:,
\quad 
{V^*}^{\odot l} = MS^lV^*M \subset V^{*\otimes l} M
\end{equation*}
where we denoted with $\odot$ the symmetric tensor product.
\begin{rem}
	\label{rem_decomp}
	In the following we will consider bundles which are constructed as direct sum, \emph{i.e.},
	\begin{equation*}
	VM=\bigoplus_{i=1}^N W_i M
	\end{equation*}
	for some vector bundles $W_iM\rightarrow M$. We stress that, in this case, using the distributivity of tensor product with respect to direct sum, we have
	\begin{equation*}
	V^{\otimes k} M=\bigoplus_{|P|=k}\bigotimes_{i=1}^NW_i^{\otimes p_i}M,\quad S^kVM=\bigoplus_{|P|=k}\bigotimes_{i=1}^NS^{p_i}W_iM
	\end{equation*}
	where $P=(p_1,\ldots, p_N)$ is a multi-index and $|P|=p_1+\cdots +p_N$. It is straight forward to write the analogous decomposition for $V^{*\otimes l} M$ and $S^l V^*M$.
\end{rem}
\noindent We will take advantage of the following spaces of smooth sections.
\begin{itemize}
	\item $\mathscr{E}(X):=\Gamma(X)$ space of smooth sections of the bundle $X$;
	\item $\mathscr{D}(X)$ space of smooth and compactly supported sections of the bundle $X$;
\end{itemize}
where $X$ can be anyone of the introduced bundles. Obviously $\mathscr{E}(X)\supset \mathscr{D}(X)$.

\begin{rem}
	In the sequel, we sometimes write tensors with indices, adopting the well-known \emph{abstract index notation}~\cite{wald84}. We use two type of indices: we use the notation with Greek indices for sections of a generic tensor bundle (for example, $v^{\mu_1,\ldots,\mu_k}$ denotes a section of $V^{\otimes k} M$) and Latin indices for spacetime tensors, \emph{i.e.},\ section of tensor products of $TM$ and $T^*M$ (for example $t^{a_1,\ldots a_k}$ denotes a section of $T^{\otimes k} M$).
\end{rem}

\begin{rem}
	\label{rem_identif}
	In the following, if the bundle $VM$ has the introduced direct sum structure, we will often take advantage of the identification: if $\mathscr{E}(VM)\ni f=\bigoplus_i f_i$ we identify $f_i\simeq \bigoplus_k\delta_i^kf_k$. With this identification we can substitute the direct sum with a standard sum:
	\begin{equation*}
	f=\sum_if_i.
	\end{equation*}
\end{rem}
It is also convenient, for notational reason, to introduce the following contraction product between tensors. We recall~\cite[Lem.9.1.1]{procesi} that fully symmetric tensors are spanned by decomposable tensors of the form $f^{\odot l}$.
\begin{defn}
	The  $l$-\textbf{contraction product}  of symmetric sections 
	\begin{equation*}
	\cdot_l\colon \mathscr{E}(S^lV^*M)\times \mathscr{E}(S^kVM)\longrightarrow \mathscr{E}(S^{k-l}VM) \quad \mbox{with $k\geq l$, }
	\end{equation*}
	is defined pointwise on decomposable tensors $g^{\odot l}\in \mathscr{E}(S^lV^*M)$,  $f^{\odot k}\in\mathscr{E}(S^kVM)$ by
	\begin{equation*}
	\left(g^{\odot l}\cdot_l f^{\odot k}\right):=\binom{k}{l}\left\langle g,f\right\rangle^l f^{\odot k-l}.
	\end{equation*}
	and extended by linearity.
\end{defn}

\begin{prop}
	\label{contrac_associativity}
	Let $k,l,s>0$ be such that $l\leq k$ and $s\leq k-l$. For $h\in \mathscr{E}(S^sV^*M), g\in\mathscr{E}(S^lV^*M)$ and $f\in\mathscr{E}(S^kVM)$ it holds
	\begin{equation}
	h\cdot_s\left(g\cdot_l f\right)=g\cdot_l\left(h\cdot_sf\right)
	\end{equation}
\end{prop}
\begin{proof}
	It is immediate using the definition. It is sufficient to prove the result for decomposable tensors and then use linearity to extend the proof to generic tensors. We consider $h^{\odot s}\in \mathscr{E}(S^sV^*M)$,  $g^{\odot l}\in\mathscr{E}(S^lV^*M)$, $f^{\odot k}\in\mathscr{E}(S^kVM)$. Thus
	\begin{flalign*}
	h^{\odot s}\cdot_s\left(g^{\odot l}\cdot_l f^{\odot k}\right)&=h^{\odot s}\cdot_s\left(\binom{k}{l}\left\langle g,f\right\rangle^l f^{\odot k-l}\right)\\
	&=\binom{k-l}{s}\binom{k}{l}\left\langle h,f\right\rangle^s\left\langle g,f\right\rangle^l f^{\odot k-l-s}\\
	&=\binom{k}{s}\binom{k-s}{l}\left\langle h,f\right\rangle^s\left\langle g,f\right\rangle^l f^{\odot k-l-s}\\
	&=g^{\odot l}\cdot_l\left(h^{\odot s}\cdot_s f^{\odot k}\right).
	\end{flalign*}
\end{proof}
We now prove some technical results that will be useful in the subsequent part. In the following we often use the shorthand notation $f^k:=f^{\odot k}$.
\begin{prop}
	\label{prop_contrac}
	Let $g\in\mathscr{E}(V^*M)$, $f_i\in\mathscr{E}(VM)$ and $p_i\geq 1$ for $i=1,\ldots N$.
	The following relations hold
	\begin{itemize}
		\item[(a)] $\displaystyle g\cdot_1 \left(g^{l-1}\cdot_{l-1}f_i^k\right)=lg^l\cdot_l f^k$,
		\item[(b)] $g\cdot_1\left(f_1^{p_1}\odot \cdots\odot f_N^{p_N}\right)=\displaystyle\sum_{l=1}^Nf_1^{p_1}\odot \cdots\odot\left(g\cdot_1 f_l^{p_l}\right)\odot\cdots\odot f_N^{p_N}$,
		\item[(c)] $g^l\cdot_l\left(f_1^{p_1}\odot \cdots\odot f_N^{p_N}\right)=\displaystyle\sum_{\substack{|Q|=l\\ q_i\leq p_i}}\left(\prod_{i=1}^N\binom{p_i}{q_i}\right)\left(g^l\cdot_l\left(f_1^{q_1}\odot\ldots\odot f_N^{q_N}\right)\right)f_1^{p_1-q_1}\odot\ldots\odot f_N^{p_N-q_N}$.
	\end{itemize}
	Moreover, if $Q=(q_1,\ldots, q_N)$ and $P=(p_1,\ldots, p_N)$ are multi-indices such that $|Q|=|P|=l$, $h^Q\in\mathscr{E}\left(\odot_{i=1}^NS^{q_i}W_iM\right)$ and $f_i\in\mathscr{E}(W_iM)$, then
	\begin{itemize}
		\item[(d)]  $h^Q\cdot_{l}\left(f_1^{p_1}\odot \cdots\odot f_N^{p_N}\right)=0$ if $P\neq Q$.
	\end{itemize}
\end{prop}

\begin{proof}
	Relation \emph{(a)} follows immediately from the definition of contraction product. We prove relation \emph{(b)}.
	It is sufficient to prove the relation for the product $f_1\odot f_2$ and then the relation \emph{(b)} follows immediately using the obtained result recursively. Defining  $f=f_1+f_2$ we obtain
	\begin{flalign*}
	g\cdot_1 f^2&=2\langle g,f\rangle\odot f=2\left(\langle g,f_1\rangle+\langle g,f_2\rangle\right)\odot(f_1+f_2)\\
	&=g\cdot_1 f_1^2+2\langle g,f_1\rangle\odot f_2+2\langle g,f_2\rangle\odot f_1+g\cdot_1 f_2^2\\
	&=g\cdot_1 f_1^2+2 (g\cdot_1f_1)\odot f_2+2( g\cdot_1f_2)\odot f_1+g\cdot_1 f_2^2
	\end{flalign*} 
	but we also have
	\begin{equation*}
	g\cdot_1 f^2=g\cdot_1 \left(f_1^2+2f_1\odot f_2+f_2^2\right)=g\cdot_1 f_1^2+2g\cdot_1( f_1\odot f_2)+g\cdot_1 f_2^2
	\end{equation*}
	and then
	\begin{equation*}
	g\cdot_1( f_1\odot f_2)=(g\cdot_1f_1)\odot f_2+f_1\odot( g\cdot_1f_2).
	\end{equation*}
	We now prove relation \emph{(c)}. Applying recursively relation \emph{(a)} and recalling that $\cdot_1$ acts as a derivation (relation \emph{(b)}), we have, for $q_i\leq p_i$,
	\begin{flalign*}
	g^l\cdot_l\left(f_1^{p_1}\odot \cdots\odot f_N^{p_N}\right)&=\frac{1}{l!}\underbrace{g\cdot_1\left(\cdots g\cdot_1\right.}_{l-\textrm{times}}\left.\left(f_1^{p_1}\odot \cdots\odot f_N^{p_N}\right)\right)\\
	&=\frac{1}{l!}\sum_{|Q|=l}\binom{l}{Q}\bigodot_{i=1}^N \underbrace{g\cdot_1\left(\cdots g\cdot_1\right.}_{q_i-\textrm{times}}\left. f_i^{p_1}\right)\\
	&=\sum_{|Q|=l} \bigodot_{i=1}^N g^{q_i}\cdot_{q_i}f_i^{p_i}\\
	&=\sum_{|Q|=l}\left(\prod_{i=1}^N\binom{p_i}{q_i}g^{q_i}\cdot_{q_i}f_i^{q_i}\right) f_1^{p_1-q_1}\odot\ldots\odot f_N^{p_N-q_N}\\
	&=\sum_{|Q|=l}\left(\prod_{i=1}^N\binom{p_i}{q_i}\right)\left(g^l\cdot_l\left(f_1^{q_1}\odot\ldots\odot f_N^{q_N}\right)\right)f_1^{p_1-q_1}\odot\ldots\odot f_N^{p_N-q_N}
	\end{flalign*} 
	where the last equality holds because, if $f=\sum_i f_i$,
	\begin{equation*}
	g^l\cdot_l f^l=\left(\sum_i\langle g,f_i\rangle\right)^l=\sum_{|P|=l}\binom{l}{P}\prod_i\langle g,f_i\rangle^{p_i}=\sum_{|P|=l}\binom{l}{P}\prod_i g^{p_i}\cdot_{p_i}f_i^{p_i}
	\end{equation*}
	but also
	\begin{equation*}
	g^l\cdot_l f^l=\sum_{|P|=l}\binom{l}{P}g^l\cdot_l\left(f_1^{p_1}\odot\ldots\odot f_N^{p_N}\right).
	\end{equation*}
	Finally, we have to prove relation $(d)$. Define $h=\sum_i h_i$ and $f=\sum_i f_i$ where $h_i,f_i\in\mathscr{E}(W_iM)$. Then
	\begin{equation*}
	h^l\cdot_l f^l=\left(\sum_i\langle h_i,f_i\rangle\right)^l=\sum_{|P|=l}\binom{l}{P}\prod_i\langle h_i,f_i\rangle^{p_i}=\sum_{|P|=l}\sum_{|Q|=l}\delta_{PQ}\binom{l}{P}\prod_i\langle h_i,f_i\rangle^{p_i},
	\end{equation*}
	but also
	\begin{equation*}
	h^l\cdot_l f^l=\sum_{|P|=l}\sum_{|Q|=l}\binom{l}{P}\binom{l}{Q}\left(h_1^{q_1}\odot\cdots\odot h_N^{q_N}\right)\cdot_l\left(f_1^{p_1}\odot\ldots\odot f_N^{p_N}\right).
	\end{equation*}
	Thus
	\begin{equation*}
	\left(h_1^{q_1}\odot\cdots\odot h_N^{q_N}\right)\cdot_l\left(f_1^{p_1}\odot\ldots\odot f_N^{p_N}\right)=0,\quad \textrm{if } P\neq Q
	\end{equation*}
	and, since $h^Q$ is a linear combination of $h_1^{q_1}\odot\cdots\odot h_N^{q_N}$, we have concluded the proof.
\end{proof}

After the first part about notations, in this section we discuss some preliminary results: we briefly recall the notions of jet bundles, just to fix notations, and we present the Peetre-Slov\'ak theorem, which is the most important result that we will use in the following.\\
If $E\to M$ is a smooth bundle, unless otherwise specified, we henceforth denote  the canonical projection by $\pi_E\colon E \to M$, the standard  fiber by $F^{(E)}$ and   $\Gamma(E)$ indicates the  set of the  smooth sections of $E$ (the smooth maps $\psi\colon M \to E$ such that $\pi_E(\psi(x))=x$ for every $x \in M$).

\subsection{Jet bundles}
In the following we use extensively the notion of \emph{jets} and \emph{jet bundles}. Naively, given a bundle $E \to M$ and a section $f\colon M\to E$, the jet of $f$ at a point $p\in M$ collects the information about the coordinate derivatives of $f$ at $p$ up to some order. The collection of all jets then forms the jet bundle associated to $E$. In this part we briefly recall some standard notions about jets and jet bundles~\cite{kms}.
\begin{defn}
	Consider a pair of smooth manifolds  $M$, $E$ and  the class of smooth functions  $f\colon M\rightarrow E$, in particular $E$ may be a bundle with base $M$ and in this case the relevant set of functions $f$ is that of smooth  sections.\\
	The \textbf{germ} of $f$ at $p\in M$ is the equivalence class $[f]_p$ of smooth functions (sections) $M\rightarrow E$ that are equal to $f$ on some neighbourhood of $p$. The $r$-\textbf{jet} of $f$ at $p\in M$, denoted by $j^r_pf$, is the equivalence class $[f]_p^r$ of smooth functions (sections) $M\rightarrow E$ that have the same Taylor expansion at $p$ as $f$ to order $r$ with respect to fixed local coordinate systems  in $M$ and $E$ (this property being  independent from the choice of the coordinate patch).
	When $E\to M$ is a smooth bundle,
$J^rE \to M$ denotes the set of $r$-jets varying the point in the base, itself a smooth bundle.
	Finally, if $\psi \in \Gamma(E)$ is a smooth section, the \textbf{$r$-jet extension} of $\psi$, denoted with $j^r\psi \in  \Gamma(J^r E)$, is the  section of $J^r E$ which collects the $r-$jets of $\psi$ over each point $p\in M$.
\end{defn}
\noindent A fiber $(J^rE)_p$ at $p\in M$ is diffeomorphic to $E_p\times \mathbb{R}^{s_r}$ where $E_p$ is the fiber of $E$ at $p$ and $s_r$ is the number of all (symmetrized) partial derivatives up to order $r$ with respect to any local chart on the base around $p$.\\
The notion of jet extension gives rise to the definition of  {\em local adapted coordinates on jet bundles}. 

\begin{defn}\label{defadptedchartJ}
Let $(x^a,v^i)$ be a local adapted coordinate chart on a bundle $E\rightarrow M$, where $x^a$ are local coordinates on an open domain $U\subseteq M$ and $(x^a,v^i)$ are trivializing coordinates on the fibers over the open domain  $Z\subseteq E$ projecting onto $U$.  This charts extends to an {\bf adapted coordinates chart}
 $(x^a,v_A^i)$ on the jet bundle $J^rE$ defined as follows.  Its  domain is $Z^r\subseteq J^rE$ is diffeomorphic to $Z\times\mathbb{R}^{s_r}$.  Moreover
\begin{equation*}
v^i_A\left(j^r\psi(p)\right)=\partial_Av^i\left(j^r\psi(p)\right)=\frac{\partial}{\partial x^{a_1}}\cdots \frac{\partial}{\partial x^{a_l}}v^i\left(j^r\psi(p)\right)
\end{equation*}
for any section $\psi$ of the bundle $E$ and where  $A=a_1\cdots a_l$ is a multi-index of size $|A|=l$ with $l=0,1,\ldots, r$.
\end{defn}

\subsection{The Peetre-Slov\'ak theorem}
\label{section_peetre}
Let $E\rightarrow M$ be a smooth bundle. We recall that the afore-mentioned $r$-jet extension of sections acts as a map $j^r\colon\Gamma(E) \ni \psi \mapsto j^r\psi \in \Gamma(J^rE)$. 
\begin{defn}
	Let $E\rightarrow M$ and $F\rightarrow M$ be smooth bundles over the same base $M$. Consider a map $D\colon\Gamma(E)\rightarrow\Gamma(F)$.
	\begin{enumerate}
		\item $D$ is a \textbf{differential operator of globally bounded order} if there exists an integer $r\geq 0$, the order, and a smooth map
		$$
			d \colon J^rE \rightarrow F\:,
		$$
		which leaves fixed the base of the transformed point 
		($\pi_F \circ d = \pi_{J^rE}$)
		 such that for any section $\psi\in\Gamma(E)$ we have an associated section of the form $$D[\psi]=d\circ j^r\psi\:.$$
		\item $D$ is a \textbf{differential operator of locally bounded order} if it satisfies a similar condition locally. Namely, if for every $y\in M$ and every $\psi_0\in\Gamma(E)$, there exists
		\begin{itemize}
			\item a neighborhood $U\subseteq M$ of $y$ with compact closure;
			\item an integer $r\geq 0$;
			\item an open neighborhood $Z^r\subseteq J^r(E)$ of $j^r\psi_0(U)$ projecting onto $U$;
			\item a smooth  function  $d\colon Z^r\rightarrow F$ which leaves fixed the base of the transformed point  
		\end{itemize}
		such that $$D[\psi](x)=d\circ j^r\psi(x)$$ for all $x\in U$ and all $\psi\in\Gamma(E)$ with $j^r\psi(U)\subseteq Z^r$.
	\end{enumerate}
\end{defn}
\noindent By elementary reasoning, the function $d$ is actually uniquely
determined by the operator $D$, once the domain $Z^r$ is fixed, since
every point in $Z^r$ lies on the graph of some $j^r\psi$.

A differential operator $D$  transforms sections $\psi$ to sections  $D[\psi]$ with the constraint that the value $D[\psi](x)$  of the transformed section attained at a point $x\in M$ depends only on the value 
of the initial section $\psi$ at the same point $x$ together with  the values of its $M$-derivatives at $x$ up to a certain order $k$, the jet $j^k_x\psi$ evaluated at the said $x$.     A natural question is how to characterize these type of local transformations of sections
among the whole class of maps $\Gamma(E) \to \Gamma(F)$.  An answer is provided by some results known as the Peetre-Slov\'ak theorem we state into two versions (there is a third more complete version we do not consider here~\cite{kms,slovak}). 
\begin{thm}[Linear Peetre's Theorem]
	\label{thm_peetre}
	Let $E\rightarrow M$ and $F\rightarrow M$ be vector bundles over the same base $M$ and $\Psi \colon\Gamma(E)\rightarrow\Gamma(F)$ a map such that $\Psi[\psi](x)\in F$ depends only on the germ of $\psi$ at $x$ for very $\psi \in \Gamma(E)$ and $x\in M$. If $\Psi$ is linear with respect to the natural vector space structures of $\Gamma(E)$ and $\Gamma(F)$, then $\Psi$ is a  differential operator of locally bounded order.
\end{thm}
\noindent In other words, if $\Psi$ is {\em linear},  even if the values $\Psi[\psi](x)$  potentially depends on the germ of $\psi$ around every considered $x\in M$, actually they only depend on  the \emph{jet} of $\psi$ at $x$ as it is proper of differential operators. 
This noticeable result for a function $\Psi\colon E \to F$,
can be made stronger keeping the requirement of dependence on the germ
but relaxing the {\em linearity} hypothesis (thus also dropping the vector space structure of the fibers of $E$ and $F$) and replacing linearity for a suitable {\em regularity} condition.  This   alternate condition  demands regularity  of $\Psi$ when it acts  on certain smooth  families of sections $\psi_s$  parametrized by $s \in {\mathbb R^k}$ we go to  introduce with the help of an auxiliary bundle used to specify the  joint-smoothness  of these families.
Given a smooth bundle $E \to M$ and  the standard projection $\pi \colon \mathbb R^n \times M \ni (s,x) \mapsto x \in M$, we define an associated smooth bundle, 
called the
{\bf pullback bundle},  $p \colon \pi^*E \to \mathbb{R}^n\times M$
whose canonical fiber is isomorphic to that  of  $E$ and  the base is $\mathbb R^n \times M$.
As a set, $ \pi^*E =  {\mathbb R}^n \times E$  
with  canonical projection onto its base given by $p\colon  \pi^*E \ni   (s, e) \to (s,\pi_E(e)) \in \mathbb R^n \times M$.  The  smooth differentiable 
structure of $\pi^*E$ is defined  accordingly.  The  smooth projection  $q \colon \pi^*E \ni (s,e) \mapsto e\in E$ restricts to  fiber diffeomorphisms $q|_{p^{-1}(s,x)} \colon p^{-1}(s,x) \to \pi^{-1}_E(x)$.  (This way  the following diagram is commutative,
 \begin{equation*}
		\begindc{\commdiag}[50]
		\obj(0,7)[1]{$\pi^*E$}
		\obj(9,7)[2]{$E$}
		\obj(0,0)[3]{${\mathbb R}^n\times M$}
		\obj(9,0)[4]{$M$}
		\mor{1}{2}{$q$}
		\mor{1}{3}{$p$}
		\mor{2}{4}{$\pi_E$}
		\mor{3}{4}{$\pi$}
		\enddc  
		\end{equation*}
and  can be used to abstractly define $\pi^*E$ taking advantage of a certain  universal property of the triple $(\pi^*E, p, q)$.)
It is now clear that a smooth section $\sigma\in \Gamma(\pi^*E)$
uniquely defines a 
$\mathbb R^n$-parametrized jointly-smooth family of  sections $\{\psi_s\}_{s \in \mathbb R^n} \subseteq \Gamma(E)$, where $\psi_s(x) := q\circ \sigma(s, x)$ for $(s,x) \in {\mathbb R}^n \times M$. This observation  justifies the following definition.

\begin{defn}\label{def-variation}  Given smooth bundle $E\rightarrow M$ and the associated pullback bundle $\pi^* E \to \mathbb R^k\times M$,
  $\sigma\in \Gamma(\pi^*E)$  is called 
{\bf smooth $n$-dimensional family} of sections of $E$. 
 If furthermore  there exists a compact subset $K\subseteq  M$ such that $\sigma(s,x)=\sigma(s',x)$  if  $x \not \in K$ and  $s,s' \in \mathbb R^n$, then $\sigma$ is said to be a   smooth {\bf compactly supported $n$-dimensional  variation}. 
\end{defn}
\noindent We are in a position to state the relevant definition about the necessary regularity required in the Peetre-Slov\'ak Theorem~\cite{kms,slovak}.
\begin{defn} \label{def_regular}
	Given smooth bundles $E\rightarrow M$ and $F\rightarrow M$,  a map $\Psi \colon\Gamma(E)\rightarrow\Gamma(F)$ is \textbf{regular} if it maps smooth $n$-dimensional families of sections to  smooth $n$-dimensional families of sections for every natural $n$.  $\Psi$ is \textbf{weakly regular} (cf.~\cite[Apx.A]{KM16}) if it maps smooth compactly supported $n$-dimensional variations to smooth compactly  
		supported $n$-dimensional variations for every natural $n$.
\end{defn}
\begin{thm}[Peetre-Slov\'ak Theorem]
	\label{thm_peetre_nonlin}
	Let $E\rightarrow M$ and $F\rightarrow M$ be smooth bundles over the same base $M$ and $\Psi\colon\Gamma(E)\rightarrow\Gamma(F)$ a map such that,  $\Psi[\psi](x)\in F$ depends only on the germ of  $\psi$ at $x$ for every  $\psi \in \Gamma(E)$ and $x\in M$. If  $\Psi$ is weakly regular, then it is a  differential operator of locally bounded order.
\end{thm}

It is worth pointing out that there is an alternative version of the
Peetre-Slov\'ak theorem, which uses \emph{differentiability in the
sense of Bastiani} in place of our \emph{regularity} hypothesis (see
Thmeorem~VII.3 of~\cite{BDLR17}, whose proof is a special case of the
proofs given for the slightly more general Theorems~VII.5,7). It is also
argued in~\cite{BDLR17} that these two hypotheses should in any case be
equivalent, because $\Gamma(E)$ is a Fr\'echet space in the usual
Whitney topology. This alternative version cannot immediately replace
the version that we will need below, because \emph{weak regularity}
coincides with \emph{regularity} only in the physically uninteresting
case of compact $M$. However, there is no reason why Theorem~VII.3
of~\cite{BDLR17} could not be strengthened along the lines of
Appendix~A of~\cite{KM16}, where we discussed strengthening the original
Peetre-Slov\'ak theorem from \emph{regularity} to \emph{weak
regularity}. It remains an open question whether our version of
\emph{weak regularity} or such an alternative version based on Bastiani
differentiability would be easier to verify in practice.

\section{Algebras of quantum observables over spacetimes with background classical fields}\label{section_algebra_bkg}
Within this section we describe the general settings in which we study the renormalization of bosonic fields in curved spacetime. We consider a quantum field over a time-oriented globally-hyperbolic spacetime $(M, \mathbf{g})$ of dimension $n$.
It is convenient to start from $M$ as a simple smooth manifold. In general, in addition to the quantum field,  there are some classical assigned \emph{background fields} on $M$. They  influence the evolution of the  quantum field, for example because they  may be  present in the equation of motion of the quantum field. The first necessary  background field is the metric itself  $\mathbf{g}$, however further   tensor or spinor  fields may enter the theory. Background fields are described as sections $\mathbf{b}$ of a suitable  bundle $BM\rightarrow M$.  Here, however, we do not assume a precise form  for the  fibers of $BM$ which will be completely fixed later.
 Since we are working  with  a locally covariant framework~\cite{BFV},  we have to deal with all bundles of (definite types of) background fields simultaneously and coherently for every globally-hyperbolic spacetime. To this end, we introduce all mathematical structures we need to appropriately describe background fields in a locally covariant framework equipped with other technical structures which will be useful later. We will take advantage of some elementary notions of the theory of categories. 

\subsection{Background geometries} 
\label{sectionbkggeom}
A notion which will play a crucial technical role in our result is the action of the multiplicative group $\mathbb R^+$ as group of physical dilations. We state a general and abstract definition.

\begin{defn}
	A bundle $E\rightarrow M$  is said to be {\bf dimensionful} if it is equipped with a  smooth action 
	of the multiplicative group $\mathbb{R}^+:= (0,+\infty)$
	$$\mathbb{R}^+\times E \ni (\lambda, e) \mapsto e_\lambda \in  E,$$
	 called {\bf scaling}.
	  It is assumed that  every  bundle  diffeomorphism $E \ni e \mapsto e_\lambda \in E$ leaves fixed each fiber of $E$ (so that 
	  the $\lambda$-parametrized family of these restrictions  to a given fiber defines a group representation of $\mathbb R^+$ in terms of fiber diffeomorphisms).
	 \\ A dimensionful  bundle  is said to be {\bf dimensionless} if the action  is chosen to be everywhere  trivial.
\end{defn}
\begin{rem} Every vector bundle or a cone sub-bundle of a vector bundle (a cone is a subset of a vector space that is invariant under multiplication by positive real numbers, \emph{e.g.},\ the cone of metrics of Lorentzian signature in the vector space of symmetric 2-tensors) can be viewed as  dimensionful, since it can be equipped with a well-defined multiplication by scalars with some  fixed power $p \in \mathbb R$  on its fibers: $t \mapsto \lambda^p t$.
\end{rem}
\noindent Next we pass to  introduce some relevant categories which will be specialized later.

\begin{itemize}
	\item $\mathfrak{Man}$ is a category of smooth manifolds. Here  objects are connected smooth manifolds $M$ of fixed dimension $n$ and morphisms are smooth embeddings $\chi \colon M \to M'$.  
	\item $\mathfrak{Bndl}$ is a category of dimensionful smooth bundles.  Here objects $\pi_E \colon E \to M$ are smooth bundles over a smooth base
	of fixed dimension $n$. Since a smooth bundle is locally trivializable, its typical fiber is diffeomorphic to a fixed  manifold  $F$ with possibly some additional structures (\emph{e.g.},\ a vector space structure) compatible with the smooth structure.
	Morphisms are smooth maps $\xi \colon E \to E'$ that are both

	 (i) {\bf fiber preserving}:  $\pi_{E'}\circ \xi = \chi_\xi \circ \pi_E$
	for uniquely associated smooth maps $\chi_\xi \colon M\to M'$ and preserving the additional structure of the fiber if any,

	  (ii) {\bf equivariant with respect to scaling}:  $\xi(e)_\lambda = \xi(e_\lambda)$  for  $\lambda \in \mathbb R^+$ and $e\in E$.
\end{itemize}

\begin{rem}
The definition applies also when a scaling action is not defined. In this case the standard scaling action is assumed to be the trivial one, \emph{i.e.},\ the bundles are supposed to be dimensionless.
\end{rem}
\noindent The interplay of a class of manifolds, which will be later interpreted as spacetimes,  with corresponding bundles, which represent background classical fields (including the metric),
 and implementing the ideas of local covariance is encapsulated in a certain type of functor~\cite{BFV} that we define below into a very general fashion. Later we specialize it to the case relevant to our work.

\begin{defn}\label{defnatbund}
	A natural (dimensionful) bundle is a functor $H\colon\mathfrak{Man}\rightarrow\mathfrak{Bndl}$ 
	such that, using the notation  $HM := H(M)$ for every $M \in \mathfrak{Man}$,
	 a morphism $\chi\colon M\rightarrow M'$ has an associated morphism $H_\chi\colon HM\rightarrow HM'$
	with $\pi_{HM'}\circ H_\chi = \chi \circ \pi_{HM}$
	and $H_\chi$ is  a local diffeomorphism (a diffeomorphism 
	onto its image).
\end{defn}
\noindent  Given a morphism $\chi \colon M \to M'$ and exploiting the fact that  $H_\chi$ is a local diffeomorphism,  it is  possible to construct a 
{\bf pullback action on sections} of the associated 
bundles $$\chi^*\colon\mathscr{E}(HM')\rightarrow\mathscr{E}(HM)$$ which is completely 
defined by requiring that   
\begin{equation}\label{defpullbacksections}\mathbf{h}'\circ\chi=H_\chi\circ(\chi^*\mathbf{h}') \quad  \mbox{for $\mathbf{h}'\in\mathscr{E}(HM')$.}\end{equation}
 Since the morphism $H_\chi$ is equivariant, the scaling commutes with the pull-back, \emph{i.e.},
 $$\chi^*(\mathbf{h}'_\lambda)=(\chi^*\mathbf{h}')_\lambda\quad \mbox{with $\lambda\in\mathbb{R}^+$.}$$
  Furthermore, exploiting the compactness of the support of the elements of $ \mathscr{D}(HM)$, also 
 a natural push-forward map $\chi_{*} \colon \mathscr{D}(HM)\to \mathscr{D}(HM')$ arises immediately. It is defined as follows 
 \begin{equation}
 \label{defpullforwardsections}
 \left(\chi_{*}f\right)(p')= H_{\chi}|_{\chi^{-1}(p')} f(\chi^{-1}(p'))\:,
 \end{equation}
 for $f\in\mathscr{D}(HM)$ and $p'\in \chi(M)$ and the right-hand side is extended to the zero function for $p' \not \in \chi(M)$.\\
Finally, if $H,H'\colon \mathfrak{Man} \to \mathfrak{Bndl}$ are natural bundles, the duals $H^*$, $H'^*$, the direct sum $H\oplus H'$ and tensor product $H\otimes H'$ also define natural bundles.

Dealing with a general framework of relativistic quantum field theory, a relevant
natural bundle, denoted by
 $B\colon\mathfrak{Man}\rightarrow\mathfrak{Bndl}$,
is
\begin{equation} 
\label{background_general}
BM= \mathring{S}^2T^*M\bigoplus_{i=1}^N  \left(T^{\otimes k_i}M \otimes T^{*\otimes l_i}M\right)
\end{equation}
where  $\mathring{S}^2T^*M \subset S^2T^*M$ is the bundle of Lorentzian metrics over $M$ and some choice of tensor powers $k_i$ and $l_i$. We will later use $BM$ as the bundle of background fields for a model of quantum fields.
Scalar fields in particular are admitted when
$k_i=l_i=0$.
As previously observed, the bundles  \eqref{background_general} are naturally dimensionful.  The sections of these type of bundles represent the {\em non-quantized fields} of definite type assigned in every spacetime simultaneously and coherently. The metric is one of these given fields. Let us state a pair of precise definitions adding some further relevant details concerning the effective action of $\mathbb R^+$.

\begin{defn} \label{def_bkg}
	$B\colon \mathfrak{Man} \to \mathfrak{Bndl}$ is the natural bundle
	of the form \eqref{background_general}, with fixed $k_i,l_i$.
	\\
	A \textbf{background field} is a section $\mathbf{b}\colon
	M \to BM$. A pair $(M,\mathbf{b})$ is a \textbf{background geometry},
	provided the section $\mathbf{b} = (\mathbf{g}, \mathbf{t}_1, \ldots,  \mathbf{t}_N)$ is such that $(M,\mathbf{g})$ is a
	time-oriented globally hyperbolic spacetime. \\
 The action of $\mathbb R^+$ on the bundles of the form~\eqref{background_general} is
 such that, for every background field,
\begin{equation}
 (\mathbf{g}, \mathbf{t}_1, \ldots,  \mathbf{t}_N)\longmapsto \left(\lambda^{-2}\mathbf{g},\lambda^{s_1}\mathbf{t}_1, \ldots, \lambda^{s_N}\mathbf{t}_N\right)\:,\quad (\mathbf{g},\mathbf{t}_1, \ldots,  \mathbf{t}_N)\in\mathscr{E}(BM)\quad \lambda \in \mathbb R^+  \label{ps}\:,
\end{equation}
for given reals $s_i$ independent from the section and $M$.
Each such transformation is called \textbf{physical scaling} transformation.
\end{defn}

\begin{defn} \label{def_lc}
 Referring to the natural bundle $B\colon \mathfrak{Man} \to \mathfrak{Bndl}$ of the form~\eqref{background_general}, we define 
	the following associated categories.

	\textbf{(a)} $\mathfrak{BkgG}$ is the \textbf{category of background geometries},
	having time-oriented background geometries as objects and morphisms
	given by smooth embeddings $\chi\colon M\to M'$ that preserve the
	background fields, $\chi^* \mathbf{b}' = \mathbf{b}$ on $M$, the time orientation, and
	causality (every causal curve between $\chi(p)$ and
	$\chi(q)$ in $M'$ is the $\chi$-image of a causal curve between $p$ and
	$q$ in $M$). 
	
	\textbf{(b)}  $\mathfrak{BkgG}^+$ is  the \textbf{category of oriented background
		geometries} having  oriented and time-oriented background geometries as
	objects and morphisms as in $\mathfrak{BkgG}$, but also required to preserve the
	spacetime orientation.
\end{defn}

\begin{rem}
The group of physical scaling transformations acts on the above categories as $(M,\mathbf{b}) \mapsto (M,\mathbf{b}_\lambda)$, for any
	$\lambda \in \mathbb{R}^+$. By equivariance of the pullback of background
	fields, physical scalings actually act as functors, $\mathfrak{BkgG} \to \mathfrak{BkgG}$ and $\mathfrak{BkgG}^+ \to \mathfrak{BkgG}^+$
	respectively.
\end{rem}

\subsection{The net of local quantum observables} The introduced formalism permits us to 
describe the {\em net  algebras of local quantum observables} on our background geometries. 
We explicitly only deal with $\mathfrak{BkgG}$, but everything we say can be trivially re-adapted to $\mathfrak{BkgG}^+$.  
\begin{defn}
	\label{def_netalgebra}
	A \textbf{net of algebras (of local quantum observables)} $\mathcal{W}$ is an assignment of a
	complex unital $*$-algebra $\mathcal{W}(M,\mathbf{b})$ to every background geometry
	$(M,\mathbf{b})$ in $\mathfrak{BkgG}$ together with an assignment of an injective unital
	$*$-algebra homomorphism $\iota_\chi \colon \mathcal{W}(M,\mathbf{b}) \to \mathcal{W}(M',\mathbf{b}')$ to
	every morphism in $\mathfrak{BkgG}$, respecting compositions and associating identities  to identities. In other words $\mathcal{W}\colon\mathfrak{BkgG} \to \mathfrak{Alg}$ is a functor
	from the category of background geometries into the category of {(complex)
		unital $*$-algebras} whose morphisms are injective unital $*$-algebra
	homomorphisms. Further, we require
	that $\mathcal{W}$ respects (i) {\bf scaling} and (ii) the {\bf time slice axiom}, as described below.	
	\begin{itemize}
		\item[\textbf{(i)}]
			Physical scaling transformations $(M,\mathbf{b}) \mapsto (M,\mathbf{b}_\lambda)$ are represented  in terms of 
		$*$-algebra isomorphisms $\sigma_\lambda\colon \mathcal{W}(M,\mathbf{b}) \to
		\mathcal{W}(M,\mathbf{b}_\lambda)$ such that $\sigma_1=id$ and $\sigma_\lambda \circ \sigma_{\lambda'}= \sigma_{\lambda\lambda'}$.  Varying $(M, \mathbf{b})$, scaling transformations must commute with embeddings, \emph{i.e.},\ they  act as natural isomorphisms
		$\sigma_\lambda \colon \mathcal{W} \to \mathcal{W}_\lambda$ between the $*$-algebra
		valued functors $\mathcal{W}$ and $\mathcal{W}_\lambda$, the latter defined by
		$\mathcal{W}_\lambda(M,\mathbf{b}) = \mathcal{W}(M,\mathbf{b}_\lambda)$.
		\item[\textbf{(ii)}]
		Given a morphism $\chi \colon M \to M' $ between the background
		geometries  $(M,\mathbf{b})$ and  $(M',\mathbf{b}')$, if the image $\chi(M) \sse M'$ contains a Cauchy surface
		for $(M',\mathbf{g}')$, then the induced $*$-homomorphism $\iota_\chi \colon
		\mathcal{W}(M,\mathbf{b}) \to \mathcal{W}(M',\mathbf{b}')$ is a $*$-isomorphism.
	\end{itemize}
	We refer to a similar functor $\mathcal{W} \colon \mathfrak{BkgG}^+ \to \mathfrak{Alg}$ with analogous 
	properties as a net of algebras as well.
\end{defn}
\begin{rem}$\null$\\
{\bf (1)}  The unit of every  algebra $\mathcal{W}(M,\mathbf{b})$ will be simply denoted by $1$ in place of a cumbersome notation $1_{(M,\mathbf{b})}$.\\ 
{\bf (2)} The scaling axiom is necessary because we will be required to  compare  local algebras defined on a given manifold  which are identified by scaling transformations. These algebras must be viewed as distinct since  their background fields are different. Therefore to compare them we need to assume that there is an isomorphism $\sigma_\lambda$ identifying them. In more physically minded presentations this structure is not discussed  and the said identification   is 
hidden in the formalism.
\end{rem}
\noindent The time slice axiom has a consequence which will play a fundamental role in the sequel, in particular for the application of the Peetre-Slov\'ak Theorem.

\begin{prop}
	\label{tau} Referring to  Definition~\ref{def_netalgebra} consider  $(M, \mathbf{b})\:,  (M, \mathbf{b}')\in \mathfrak{BkgG}$ 
	(resp. $\mathfrak{BkgG}^+$)
	such that $\mathbf{b}= \mathbf{b}'$ with identical temporal orientation outside a compact region $K \subseteq M$.	 There exists a unital $*$-algebra isomorphism 
	$$
		\tau \colon  \mathcal{W}(M, \mathbf{b}) \to \mathcal{W}(M, \mathbf{b}')
	$$
	such that $\tau|_{ \mathcal{W}(N, \mathbf{b}|_N)} \colon
	\mathcal{W}(N, \mathbf{b}|_N) \to \mathcal{W}(N, \mathbf{b}'|_N)$ is the identity for every  $(N, \mathbf{b}|_N)
	\in  \mathfrak{BkgG}$ (resp. $\mathfrak{BkgG}^+$) satisfying  $N \cap J_{(M,\mathbf{b})}^+(K) = \emptyset$. 
\end{prop}

\begin{proof}
	The proof is based mainly on the time slice axiom (Definition~\ref{def_netalgebra}). See~\cite[Sec.3]{KM16} for more details.
\end{proof}

\section{Quantum fields}\label{section_quantum_fields}

{}
 We have so far discussed all the mathematical structures  we need to describe background fields  and the abstract notion of a net of quantum observables.  At this abstract level we may introduce the definition of quantum fields as special elements of the algebras of observables~\cite{BFV}.

\begin{defn} \label{def_quant_field}
Fix a net of local quantum observables $ \mathcal{W}$ as in Definition~\ref{def_netalgebra}
and a natural vector bundle $V$.
A {\bf quantum $V$-field} is an assignment $\Phi_{(M,\mathbf{b})}$ of an algebra-valued distribution
\begin{equation*}
\Phi_{(M,\mathbf{b})}\colon\mathscr{D}(V^*M)\rightarrow \mathcal{W}(M,\mathbf{b})
\end{equation*}
to each background geometry $(M,\mathbf{b}) \in \mathfrak{BkgG}$.
\end{defn}
The given definition does not yet assume any particular relation between $\Phi_{(M,\mathbf{b})}$ and $\Phi_{(M',\mathbf{b}')}$
when $(M,\mathbf{b})$ and $(M',\mathbf{b}')$ are connected by a morphism $\chi$ of  $\mathfrak{BkgG}$. 
A quantum field $\Phi$ is said to be {\em locally covariant} when each pair of $\Phi_{(M,\mathbf{b})}$ and $\Phi_{(M',\mathbf{b}')}$ is in fact 
connected according to 
a natural rule arising from the  definition of a natural bundle, translating into the mathematical language the ideas of locality and general covariance.

\begin{defn}
	\label{local_field}
	A quantum $V$-field $\Phi$ (Definition~\ref{def_quant_field}) with  respect to the net of local quantum observables $ \mathcal{W}$ as in Definition~\ref{def_netalgebra} is said to be {\bf locally covariant}
if it satisfies the following identity for each morphism $\chi\colon(M,\mathbf{b})\rightarrow (M',\mathbf{b}')$ where $\mathbf{b}= \chi^*\mathbf{b}'$,
	\begin{equation}\label{Akloccov}
	\iota_\chi\left(\Phi_{(M,\mathbf{b})}(f)\right)=\Phi_{(M',\mathbf{b}')}\left(\chi_*f\right),\qquad \forall f\in\mathscr{D}(V^*M).
	\end{equation}
	where $\chi_*$ is the push-forward with respect to a natural bundle, as in \eqref{defpullforwardsections}.
\end{defn}

\begin{rem}$\null$

\noindent{\bf (1)} It is usual to require that the field $\Phi$ has a definite \textbf{scaling degree} $d_{\Phi} \in\mathbb R$ with respect to the action of physical scaling. However, when $\Phi$ has multiple components, different components of a $V$-field can be grouped together by their scaling degree, giving rise to the decomposition of the field bundle as in Remark~\ref{rem_decomp}. Then the role of the scaling degree is played by a globally diagonalizable endomorphism $\mathbf{d}_\phi \colon VM \to VM$, whose eigen-subspaces constitute the bundle decomposition $VM = \bigoplus_{i=1}^N W_i M$ and whose eigenvalues correspond to the weights of these field sub-bundles. Alternatively, once this field bundle decomposition is known, the endomorphism $\mathbf{d}_{\Phi}$ can be identified with its eigenvalues $(d_{\Phi_1},\ldots,d_{\Phi_N})$. Informally written, the relation between the scaling of background fields and $\Phi$ means that
\begin{equation}\label{sloppyscaling}
(\mathbf{g}, \mathbf{t}_1, \ldots,  \mathbf{t}_N)\longmapsto \left(\lambda^{-2}\mathbf{g},\lambda^{s_1}\mathbf{t}_1, \ldots, \lambda^{s_N}\mathbf{t}_N\right)\quad\Longrightarrow\quad \Phi\longmapsto \lambda^{\mathbf{d}_{\Phi}}\Phi,\qquad (\mathbf{g},\mathbf{t})\in\Gamma(BM)\:.
\end{equation}
To formulate a precise statement  valid also for  $V$-fields exploiting our formalism we need a precise scaling procedure based on the isomorphism
$\sigma_\lambda$ introduced in Definition~\ref{def_netalgebra}. If $\Phi$ is a quantum $V$-field, we can define the rescaled quantum $V$-field $S_\lambda \Phi$ as
\begin{equation}
\label{S_lambda}
(S_\lambda \Phi)_{(M,\mathbf{b})}(f)=\sigma^{-1}_\lambda(\Phi_{(M,\mathbf{b}_\lambda)}(\lambda^nf))\:, \quad \lambda \in \mathbb R^+\:,
\end{equation}
where $n$ is the dimension of the spacetime $M$ and $\sigma_\lambda$ is the algebra isomorphism introduced in Definition~\ref{def_netalgebra}. It should be clear here  that both $\Phi_{(M,\mathbf{b})}$ and  $(S_\lambda \Phi)_{(M,\mathbf{b})}$ are element of the same algebra  $\mathcal{W}(M,\mathbf{b})$ due to the presence of $\sigma_\lambda^{-1}$ in the second case. The factor $\lambda^n$ just compensates the scaling of the volume form $dg\mapsto \lambda^{-n}dg$ when $\mathbf{g}\mapsto\lambda^{-2}\mathbf{g}$.  A mathematically rigorous version of \eqref{sloppyscaling} is now
$$
	(S_\lambda \Phi)_{(M,\mathbf{b})}(f) = \lambda^{d_\Phi}\Phi_{(M,\mathbf{b})}(f)\:, \quad (M, \mathbf{b})\in \mathfrak{BkgG}\:, \quad  f \in\mathscr{D}(VM)\:, \quad \lambda \in \mathbb R^+ \:.
$$

\noindent{\bf (2)} As usual, it is convenient  to think of the  algebra-valued distribution $\Phi$  as a formal point-like field $\Phi^\mu(x)$ smeared with a test section $f\in \mathscr{D}(V^*M)$. In the sequel, we will make extensive use of the case when $V$ is replaced by $S^k V$. Then we may write the formal point-like field
$\Phi^{\mu_1\cdots \mu_k}(x)$ smeared with a test section $f\in\mathscr{D}(S^kV^*M)$ as
\begin{equation*}
\Phi_{(M, \mathbf{b})}(f)=\int_M \Phi^{\mu_1\cdots \mu_k}(x)f_{\mu_1\cdots \mu_k}(x)\; dg(x) \:,
\end{equation*}
where $dg(x)$ is the volume form induced by the metric $\mathbf{g}$ on $M$.\\
Similarly, if $\mathbf{C}$ is any function that maps a background geometry $\mathbf{b}$ to a distribution on a certain space of test functions $\mathbf{C}_{(M,\mathbf{b})}\colon \mathscr{D}(S^kVM)\to \mathbb R$, 
it is heuristically convenient to use the distributional notation  
\begin{equation}\label{formalC}
\mathbf{C}_{(M,\mathbf{b})}(f):= \int_M C[M,\mathbf{b}]_{\mu_1\cdots \mu_k}(x)f^{\mu_1\cdots \mu_k} (x)\; dg(x)\:,  \quad 
f\in\mathscr{D}(S^kV^*M), 
\end{equation}
where the formal  $c$-\emph{number field} $C[M,\mathbf{b}]_{a_1\cdots a_k}(x)$ may not be smooth.  Such a function $\mathbf{C}$ will be named a {\bf $c$-number $S^kV$-field}. If the sections $C[M,\mathbf{b}]$ are smooth for every $(M,\mathbf{b})$ (\emph{i.e.},\ $C[M,\mathbf{b}]\in\mathscr{E}(S^kV^{*} M)$), $\mathbf{C}$ is said to be a {\bf smooth $c$-number $S^kV$-field}. From now on we systematically identify $\mathbf{C}_{(M,\mathbf{b})}(f)$ with  the corresponding  trivial element, a so called  {\em $c$-number}, $\mathbf{C}_{(M,\mathbf{b})}(f)1$ of  $\mathcal{W}(M,\mathbf{b})$. In this sense a $c$-number $S^kV$-field is a sub-case 
of a quantum $S^kV$-field and,  for instance,  the scaling action \eqref{S_lambda} applies to these particular quantum fields as well.\\

\noindent{\bf (3)} If in addition $\mathbf{C}$ satisfies the identity 
\begin{equation}\label{Cloccov}
\mathbf{C}_{(M,\chi^*\mathbf{b}')}(f)=\mathbf{C}_{(M',\mathbf{b}')}(\chi_* f)\:, \end{equation}
for every  background morphism $\chi\colon(M,\mathbf{b})\rightarrow (M',\mathbf{b}')$ (so $\mathbf{b} = \chi^* \mathbf{b}'$) and every test section $f \in \mathscr{D}(S^kV^*M)$, then $\mathbf{C}$ defines a \textbf{locally covariant} $c$-number  $S^kV$-field.\\
Using the definitions of pull-back $\chi^{*} \colon \mathscr{E}(S^kV^*M')\to \mathscr{E}(S^kV^*M)$ and push forward 
$\chi_{*} \colon \mathscr{D}(S^kVM)\to \mathscr{D}(S^kVM')$
, it is easy to prove that if $\mathbf{C}$ is described by $ C[M,\mathbf{b}] \in \mathscr{E}(S^kV^*M)$
for every choice of $(M,\mathbf{b})$ by means of \eqref{formalC}, then \eqref{Cloccov}
is equivalent to 
\begin{equation}\label{Cloccov2}
C[M,\chi^*\mathbf{b}'](x)= \left(\chi^* C[M',\mathbf{b}']\right)(x)\:, \end{equation} for every  background morphism $\chi\colon(M,\mathbf{b})\rightarrow (M',\mathbf{b}')$ (with $\mathbf{b} = \chi^* \mathbf{b}'$) and $x\in M$.
\end{rem}

\section{Wick powers of quantum boson fields}\label{section_wick}
In this section we want to study the renormalization of Wick powers of a generic quantum boson field. Before discussing this problem, we have to introduce a useful definition concerning the notion of physical scaling introduced in Section~\ref{sectionbkggeom}.

\paragraph{Physical scaling degree.}
Physical scaling will be used together with the notion of {\em homogeneous} and {\em almost homogeneous} scaling degree. Since these notions are quite abstract we can present them into a general inductive  definition~\cite[Def.2.3]{KM16}.
\begin{defn}
	\label{def_homogeneous}
	Consider a linear representation $\rho\colon\mathbb{R}^+\rightarrow GL(W)$ of the multiplicative group $\mathbb{R}^+$ on a vector space $W$ whose action is indicated by   $W\ni F\mapsto F_\lambda := \rho(\lambda)F \in W$, for every $\lambda\in\mathbb{R}^+$.
	\begin{enumerate}
		\item[\textbf{(a)}]An element $F\in W$ is said to have \textbf{homogeneous degree} $k\in\mathbb{R}$ if 
		\begin{equation*}
		F_\lambda=\lambda^kF,\quad \mbox{for all } \lambda\in\mathbb{R}^+.
		\end{equation*}
		\item[\textbf{(b)}]An element $F\in W$ is said to have \textbf{almost homogeneous degree} $k\in\mathbb{R}$ \textbf{and order} $l\in\mathbb{N}$ if $l\geq 0$ is an integer such that (with the sum over $j$ is omitted when $l=0$)
		\begin{equation*}
		F_\lambda=\lambda^kF+\lambda^k\sum_{j=1}^l\left(\log^j\lambda\right)G_j,\quad \mbox{for all } \lambda\in\mathbb{R}^+,
		\end{equation*}
		and for some $G_j\in W$ depending on $F$, which have respectively almost homogeneous degree $k$ and order $l-j$. An element that is almost homogeneous of order $l=0$ is   homogeneous by definition.
	\end{enumerate}
\end{defn}
\noindent In the rest of the paper we will exploit  several technical results about  physical scaling in concrete application, all reported in Appendix~\ref{section_scaling}, with the exception of the next general lemma,  proved in~\cite[Lem.2.5]{KM16}.
\begin{lem}
	\label{lemma_product}
	Referring to Definition~\ref{def_homogeneous}, consider a pair of vector spaces $W, W'$ endowed with corresponding representations of $\mathbb{R}^+$. Concerning (b) below, assume also that there is a product  $W\times W' \rightarrow V$ such that  (i) $V$ admits a representation of $\mathbb{R}^+$ and (ii) the map  $W\times W' \to V$ is equivariant: $F_\lambda F'_\lambda = (FF')_\lambda$ for $F\in W$, $F'\in W'$ and $\lambda \in \mathbb{R}^+$. The following facts hold.
	\begin{itemize}
		
		\item[\textbf{(a)}] A linear combination of two elements $F, F' \in W$ of almost homogeneous degree $k$ and order $l$ is of almost homogeneous degree $k$ and order $l$.
		
		\item[\textbf{(b)}]  A product of an element $F\in W$, of almost homogeneous degree $k$ and order $l$, and an element $F'\in W'$, of
		almost homogeneous degree $k'$ and order $l'$, has almost homogeneous degree $k+k'$ and order $l+l'$.
	\end{itemize}
\end{lem}

\paragraph{General settings.} Our general setting is the following: 
\begin{enumerate}
	\item We start with a bundle $VM$ which is constructed as a direct sum of vector bundles
		\begin{equation}
			\label{bundle_decomposition}
			VM=\bigoplus_{i=1}^N W_iM.
		\end{equation}

	\item We consider a locally covariant quantum $V$-field $A_{(M,\mathbf{b})}\colon\mathscr{D}(VM)\rightarrow \mathcal{W}(M,\mathbf{b})$ and we characterize it as a quantum \emph{boson field} in the following way. We assume that the  commutator of two $V$-fields $[A_{(M,\mathbf{b})}(f),A_{(M,\mathbf{b})}(g)]$  is a $c$-number, \emph{i.e.},
	\begin{equation}
	\label{rem_commutator}
	[A_{(M,\mathbf{b})}(f),A_{(M,\mathbf{b})}(g)]=C_{(M,\mathbf{b})}(f\otimes g)1,
	\end{equation}
	where $C_{(M,\mathbf{b})}\in\mathscr{D}'(VM\times VM)$ is a distribution with some suitable properties (\emph{e.g.},\ for {\em boson fields} it vanishes for spacelike separated arguments).
	Thus,  Schwartz' kernel theorem implies that  a unique continuous linear map $\Delta_{(M,\mathbf{b})}\colon \mathscr{D}(VM)\rightarrow\mathscr{D}'(VM)$ exists such that $[\Delta_{(M,\mathbf{b})}(g)](f)=C_{(M,\mathbf{b})}(f\otimes g)$. We require moreover that $\Delta_{(M,\mathbf{b})}(g)$ is \textbf{regular} in the sense that
	\begin{equation}\label{kernel}
	\Delta_{(M,\mathbf{b})}\colon  \mathscr{D}(VM)\rightarrow\mathscr{E}(V^*M) ,
	\end{equation}
	where we used the fact that $\mathscr{E}(V^*M)\subset\mathscr{D}'(VM)$.
	There are many ways to implement this requirement in practical cases, for example our assumption holds when the dynamics of the field $A$ is ruled by any hyperbolic field equation in view of the theorem of propagation of singularities (in this case $\Delta_{(M,\mathbf{b})}$ is simply the causal propagator, see Section \ref{proca_section}  for an example). More generally it holds when some microlocal spectrum (cf.~\cite{AAQFT15Ch5}, \cite[Ch.4]{BF09}) hypothesis on the wavefront set  of $n$-point functions is  assumed with respect to relevant classes of states even in the absence of a field equation\footnote{\eqref{kernel} is valid when $WF(C_{(M,\mathbf{b})}) \not \ni (x,y,p_x,p_y)$ with either $p_x=0$ or $p_y=0$ and this is guaranteed as soon as some standard  microlocal spectrum condition on $C_{(M,\mathbf{b})}$ is valid, in particular if $C_{(M,\mathbf{b})}$ is a bisolution of a hyperbolic field equation.}.
	
	If we use explicitly the decomposition \eqref{bundle_decomposition} the map $\Delta_{(M,\mathbf{b})}$ can be seen as a direct sum of maps
	\begin{equation*}
	\Delta_{(M,\mathbf{b})}=\bigoplus_{l=1}^N\sum_{j=1}^N\Delta_{(M,\mathbf{b})}^{lj},
	\end{equation*}
	where $\Delta_{(M,\mathbf{b})}^{lj}\colon  \mathscr{D}(W_jM) \rightarrow\mathscr{E}(W^*_lM)$.
	
	\item Since we assumed a bundle constructed as in \eqref{bundle_decomposition}, the $V$-field $A_{(M,\mathbf{b})}$ can be written as a $N$-tuple of $W_i$-fields
	\begin{equation*}
		A_{(M,\mathbf{b})}=\left((A_1)_{(M,\mathbf{b})},\ldots,(A_N)_{(M,\mathbf{b})}\right).
	\end{equation*}
	We assume that each $W_i$-field $(A_i)_{(M,\mathbf{b})}$ scales homogeneously under physical scaling with degree $d_{A_i}\in\mathbb{R}$, \emph{i.e.},
	$$
		(S_\lambda A_i)_{(M,\mathbf{b})}(f) = \lambda^{d_{A_i}}(A_i)_{(M,\mathbf{b})}(f)\:, \quad (M, \mathbf{b})\in \mathfrak{BkgG}\:, \quad  f \in\mathscr{D}(W_iM)\:, \quad \lambda \in \mathbb R^+ \:.
	$$
	We will say that the $V$-field $A_{(M,\mathbf{b})}$ scales homogeneously with degree $\mathbb{R}^N\ni \mathbf{d}_A=(d_{A_1},\ldots d_{A_N})$ under physical scaling.
	
	\item We then consider the Wick powers $A^k$ of $A$. These quantum fields $A^k$ have the physical interpretation of products of $k$ factors $A$ evaluated at the same point $x$. Formally, assuming a geometric background $(M, \mathbf{b})$ has been fixed,  
	$$A^k_{\mu_1\ldots \mu_k}(x) = (A_{\mu_1}\cdots A_{\mu_k})(x) \:.$$
	It is worth stressing that these quantum fields are not elements of the sub unital $*$-algebra generated by $1$ and elements $A(f)$ since these elements  are associated with kernels formally evaluated at  {\em different} points of spacetime, \emph{i.e.},\ they are linear combinations of objects $A_{\mu_1}(x_1)\cdots A_{\mu_k}(x_k)$. Thus Wick powers need a specific definition which, as is well-known, involves some \emph{renormalization} procedure.\\
	Finally, we stress that, using the decomposition introduced in Remark~\ref{rem_decomp}, the Wick powers $A^k$ can be written as a sum
	\begin{equation}
	\label{field_decomp}
	A^k(f)=\sum_{|P|=k}\binom{k}{P}A^k\left(f_1^{p_1}\odot\cdots\odot f_N^{p_N}\right)=:\sum_{|P|=k}\binom{k}{P}A^P\left(f_1^{p_1},\cdots, f_N^{p_N}\right),
	\end{equation} 
	where $P=(p_1,\ldots,p_N)$ is a multi-index and $\binom{k}{P}=\frac{k!}{\prod_{i=1}^N p_i!}$. The last equality is intended as a definition.
\end{enumerate}

  \begin{figure}
	\begin{center}
		\def\svgwidth{10cm}
		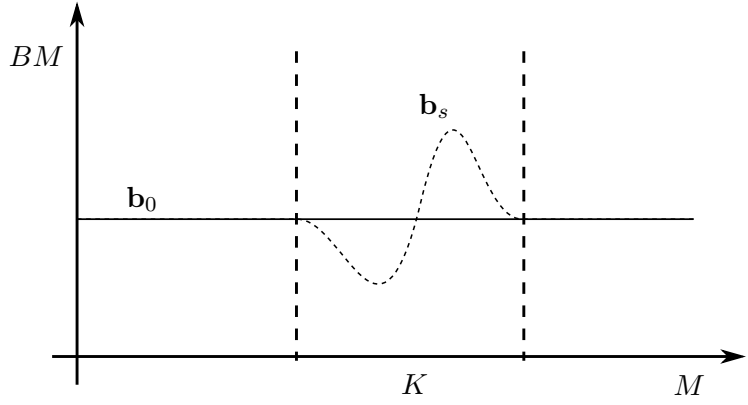
		\caption{We consider compactly supported variations of background fields.}
	\end{center}
\end{figure}
We assume an axiomatic viewpoint stating five axioms regarding Wick powers. These axioms do not determine them, but determine the degrees of freedom due to the different possible choice of renormalization procedures and classify the finite renormalization counterterms.
 Regarding the 5th requirement in the definition below, for clarity we recall the notion of a compactly supported variation from Definition~\ref{def-variation}. If $(M,\mathbf{b})$ is a background geometry, the jointly smooth function $\mathbf{B}= \mathbf{b}_s(x)$ with $s \in \mathbb R^m$ and $x\in M$ defines  a smooth $m$-dimensional ($m\geq 0$ integer) family of smooth compactly supported variations of $\mathbf{b}$ if $\mathbf{b}_s(x) =\mathbf{b}(x)$ for $x \in M$ and  $\mathbf{b}_s(x)=\mathbf{b}_{s'}(x)$ for $s, s' \in \mathbb R^n$ and $x \not \in K$ for a fixed compact $K \subseteq M$ depending on the family.
According to Proposition~\ref{tau}, we can identify each algebra $\mathcal{W}(M,\mathbf{b}_s)$ with $\mathcal{W}(M,\mathbf{b})$ by means of 
the unital $*$-algebra isomorphism 
\begin{equation}\label{taus}
\tau_s \colon \mathcal{W}(M,\mathbf{b}) \to \mathcal{W}(M,\mathbf{b}_s)  \:,
\end{equation}
 which reduces to the identity 
on every spacetime $(N, \mathbf{g}_s|_N)$ if $N \cap J_{(M,\mathbf{g})}^+(K)= \emptyset$.

\begin{defn}[Wick powers for general boson fields]
	\label{def_wick}
	Consider  a net of algebras $\mathcal{W}$ on the category of background
	geometries $\mathfrak{BkgG}$ (resp. $\mathfrak{BkgG}^+$) and a locally-covariant quantum  $V$-field $A$ (Definition~\ref{local_field})  with $A_{(M,\mathbf{b})} \colon \mathscr{D}(VM) \to \mathcal{W}(M,\mathbf{b})$ for every $(M, \mathbf{b}) \in \mathfrak{BkgG}$ (resp. $\mathfrak{BkgG}^+$).\\
	 A class of \textbf{Wick powers} $\{A^k\}$ of $A$, for $k=0,1,2,\dots$ is a
	 family of	 \emph{symmetric} locally-covariant quantum $S^kV$-fields (so that each $k$ defines an assignment  of algebra-valued distribution $A_{(M,\mathbf{b})}^k \colon \mathscr{D}(S^kVM) \to \mathcal{W}(M,\mathbf{b})$ to every 
$(M, \mathbf{b}) \in \mathfrak{BkgG}$ (resp. $\mathfrak{BkgG}^+$) respecting \eqref{Akloccov}) satisfying the following requirements.

	\begin{enumerate}
		\item \textbf{Low powers.} $A^0=\mathbf{1}$, the unit $c$-number field,  $A^1=A$, the $V$-field.
		
		\item \textbf{Scaling.} With respect to the decomposition \eqref{field_decomp}, each component $A^k_P$ of the Wick power $A^k$ is almost homogeneous of degree $\langle P,\mathbf{d}_A\rangle=p_1d_{A_1}+\cdots p_Nd_{A_N}$, with respect to the action of physical scalings $S_\lambda$ in \eqref{S_lambda}; that is, there exists an integer $l\geq 0$ and quantum $k$-tensor fields $B_j$ such that
		\begin{equation*}
		S_\lambda A^P=\lambda^{\langle P,\mathbf{d}_A\rangle} A^P+ \lambda^{\langle P,\mathbf{d}_A\rangle}\sum_{j=1}^{l}\left(\log^j\lambda\right)B_j,
		\end{equation*}
		where each $B_j$ is itself almost homogeneous of degree $\langle P,\mathbf{d}_A\rangle$ and order $l-j$.
		(Every degree is supposed to be independent from the choice of the background geometry).

		\item\textbf{Kinematic completeness.} For any $(M,\mathbf{b})$, an element $a\in\mathcal{W}(M,\mathbf{b})$ satisfies $$[a,A_{(M,\mathbf{b})}(f)]=0\quad  \mbox{for every $f\in\mathscr{D}(VM)$}$$ iff $a=\alpha 1$, with $\alpha\in\mathbb{C}$ and $1$ the unit element of the algebra.

		\item \textbf{Commutator expansion.} Each Wick power $A^k$ also satisfies the following properties\footnote{We recall that there is a factor $k$ hidden in the contraction product.}:			
		\begin{equation}
		\label{comm_exp}
		\left[A^k_{(M,\mathbf{b})}(f),A_{(M,\mathbf{b})}(g)\right]=iA_{(M,\mathbf{b})}^{k-1}(\Delta_{(M,\mathbf{b})}(g)\cdot_1f),\qquad f\in\mathscr{D}(S^kVM),\quad g\in\mathscr{D}(VM)
		\end{equation}
		where $\Delta_{(M,\mathbf{b})}\colon\mathscr{D}(VM)\rightarrow\mathscr{E}(V^*M)$ is a given map.
		\item \textbf{Smoothness.}		If  $(M, \mathbf{b}) \in \mathfrak{BkgG}$  (resp. $\mathfrak{BkgG}^+$), we require that there exist a class of states $\mathcal{S}_{(M,\mathbf{b})}$ on  $\mathcal{W}(M,\mathbf{b})$ such that if $\omega\in\mathcal{S}_{(M,\mathbf{b})}$,  the expectation values $\omega\; \circ\; \tau^{-1}_s\left(A_{(M, \mathbf{b}_s)}^k(f)\right
		)$ (with $f\in\mathscr{D}(S^kVM)$)
		can be written as 
		$$\omega \circ \tau^{-1}_s\left(A^k_{(M, \mathbf{b}_s)}(f)\right) = \int_M \omega_{\mu_1\cdots \mu_k}(s,x) f^{\mu_1\cdots \mu_k}(x) dg(x),$$
		for some jointly smooth kernels
		$$\mathbb R^m \times M \ni (s,x) \mapsto \omega_{\mu_1\cdots \mu_k}(s,x) \in \mathbb R, $$  for every  smooth $m$-parameter family of compactly supported variations $\mathbf{b}_s$ of $\mathbf{b}$ on $M$ and $\tau_s \colon \mathcal{W}(M,\mathbf{b}) \to \mathcal{W}(M,\mathbf{b}_s)$ defined as in \eqref{taus} and every integer $m\geq 0$.
		\item \textbf{Hermiticity.} For all background geometry $(M,\mathbf{b})$ and $f \in \mathscr{D}(V^*M)$, we require that
		$$
		A_{(M,\mathbf{b})}(f) = A_{(M,\mathbf{b})}(f)^*\:
		$$
		where * denotes the corresponding operation in the $*$-algebra $\mathcal{W}$.
		\end{enumerate}
\end{defn}

\begin{rem}
	While the first four axioms, and also the last one, are standard requirements, we would like to comment briefly the smoothness axiom. We require that  any Wick powers has \emph{smooth} expectation value both with respect to $x$, the coordinate on spacetime manifold, and $s$, the parameter that labels the variations of $\mathbf{b}$. The smoothness with respect to $x$ reflect the physical idea that a renormalized observables is smooth since we have removed all singularity in the renormalization procedure. The joint smoothness in $(x,s)$ is a version of the \emph{parametrized microlocal spectrum condition} that was introduced in~\cite[Def.3.5(iv)]{KM16}, as a substitute for the old analyticity condition of Hollands and Wald~\cite{HW1}.
\end{rem}

\begin{rem}
	In components, \emph{i.e.},\ with respect to equation \eqref{field_decomp}, the commutator expansion axiom, for $j\in{1,\ldots, N}$ and a multi-index $P$, becomes 
	\begin{equation}
	\label{commutator_components}
	\begin{gathered}
	\left[A_1^{p_1}\cdots A_N^{p_N} \left(f_1^{p_1},\cdots, f_N^{p_N}\right),A_j(g_j)\right]=\\
	i\sum_{l=1}^{N}A_1^{p_1}\cdots A_l^{p_l-1}\cdots A_N^{p_N}\left(f_1^{p_1},\cdots , (\Delta^{lj}(g_j)\cdot_1f_l^{p_l}), \cdots ,f_N^{p_N}\right).
	\end{gathered}
	\end{equation}
	To show this, fixing a background geometry $(M,\mathbf{b})$, consider  $g_j\in (W^*M)_j$ and $f=\sum_j f_l$ with $f_l\in W_lM$, where we have used the identification introduced in Remark~\ref{rem_identif}, \emph{i.e.},
	\begin{equation*}
	g_j=(0,\ldots,g_j,\ldots,0),\quad 	f_j=(0,\ldots,f_j,\ldots,0).
	\end{equation*}
	 We recall that, by definition,
	\begin{equation*}
		A_1^{p_1}\cdots A_N^{p_N} \left(f_1^{p_1},\cdots, f_N^{p_N}\right):=A^k\left(f_1^{p_1}\odot\cdots\odot f_N^{p_N}\right),
	\end{equation*}
	then, using Proposition~\ref{prop_contrac},
	\begin{flalign*}
	\left[A_1^{p_1}\cdots A_N^{p_N} \left(f_1^{p_1},\cdots, f_N^{p_N}\right),A_j(g_j)\right]&=\left[A^k\left(f_1^{p_1}\odot\cdots\odot f_N^{p_N}\right),A(g_j)\right]\\
		&=iA^{k-1}\left(\sum_{l=1}^N\Delta^{lj}(g_j)\cdot_1\left(f_1^{p_1}\odot\cdots\odot f_N^{p_N}\right)\right)\\
		&=i\sum_{l=1}^NA^{k-1}\left(f_1^{p_1}\odot\cdots\odot \left(\Delta^{lj}(g_j)\cdot_1f_l^{p_l}\right)\odot\cdots \odot f_N^{p_N}\right),
	\end{flalign*}
	which by definition is equal to \eqref{commutator_components}.

\end{rem}

\subsection{General renormalization formula for Wick products of a quantum boson field} If  $\{\tilde{A}^k\}_{k=1,2,\ldots}$ and $\{A^k\}_{k=1,2,\ldots}$ are two families  of Wick powers of the same quantum $V$-field $A$, our task is now to find a formula relating these two pairs of Wick powers relying on the fact that both classes satisfy the above set of general axioms. The following theorem is a generalization of~\cite[Lem.3.3]{KM16}
\begin{thm}
	\label{lemma_ipotesi}
	 Let $\{\tilde{A}^k\}_{k=1,2,\ldots}$ and $\{A^k\}_{k=1,2,\ldots}$ be two families  of Wick powers (Definition~\ref{def_wick}) of the same locally-covariant quantum $V$-field $A$ (Definition~\ref{local_field}) of homogeneous scaling degree $\mathbf{d}_A \in \mathbb R^N$. Then   there exists a family of smooth locally-covariant $c$-number $S^kV$-fields $\{C_k\}_{k=1,2,\ldots}$, where $C_1=0$, such that, for every $k=1,2, \ldots$,
	\begin{equation}
	\label{formula_generic} 
	 \widetilde{A}_{(M, \mathbf{b})}^k(f) = A_{(M, \mathbf{b})}^k(f) + \sum_{l=0}^{k-1} A_{(M, \mathbf{b})}^l\left(C_{k-l}[M, \mathbf{b}]\cdot_{k-l}f\right),
	 \end{equation}
	  where $(M, \mathbf{b}) \in  \mathfrak{BkgG}$ 
	(resp. $\mathfrak{BkgG}^+$) and $f \in \mathscr{D}(S^kV^*M)$. In components equation \eqref{formula_generic} turns out to be
	\begin{equation}\label{formula_generic_comp}
	\begin{split}
	 \widetilde{A}^P_{(M, \mathbf{b})}\left(f_1^{p_1},\ldots f_N^{p_N}\right)= &A^P_{(M, \mathbf{b})}\left(f_1^{p_1},\ldots f_N^{p_N}\right)\\
	 & + \sum_{l=0}^{k-1}\sum_{\substack{|Q|=l\\ q_i\leq p_i}}\left(\prod_{i=1}^N\binom{p_i}{q_i}\right)A^Q_{(M, \mathbf{b})}\left(\left(C_{k-l}^{P-Q}[M, \mathbf{b}]\cdot_{k-l}f^{P-Q}\right)f_1^{q_1},\ldots ,f_N^{q_N}\right)
	\end{split}
	\end{equation}
	where $Q=(q_1,\ldots, q_N)$, $P=(p_1,\ldots,p_N)$ are multi-indices and $C_{k}=\sum_{|Q|=k}C_{k}^{Q}$.\\
	 Finally, for every fixed $M \in \mathfrak{Man}$,
	
	 (i)  the map 
	  $$\Gamma(BM) \ni  \mathbf{b} \mapsto C_k[M,\mathbf{b}] \in   \mathscr{E}(S^kV^*M)$$
	  is a differential operator of locally bounded order. Regarding components of the coefficients $C_{k}^{Q}[M,\mathbf{b}]\in\mathscr{E}(\bigodot_{i=1}^NS^{q_i}W^*_iM)$;

	   (ii) each $C_k^Q[M,\mathbf{b}]$ scales almost homogeneously of degree $\langle Q,\mathbf{d}_A\rangle$ under the physical scaling transformation on $\mathbf{b}$. 
\end{thm}
\begin{proof}
In the first part of the proof we write $A(f)$ in place of $A_{(M,\mathbf{b})}(f)$ and we adopt similar notations for the other involved fields, for the sake of notational simplicity.
For all $k$, the difference 
\begin{equation*}
\widetilde{A}^k(f)-A^k(f)=D_{k}(f),\qquad f\in\mathscr{D}(S^kVM)
\end{equation*}
defines, by construction, a symmetric  locally-covariant quantum $V^{\otimes k}$-field of order $k$, in particular is self-adjoint. Using Axiom 1 and the commutator expansion \eqref{comm_exp} in Axiom 4, it is easy to show that
\begin{equation} \label{relationD}
\left[D_k(f),A(g)\right]=iD_{k-1}\left(\Delta(g)\cdot_1 f\right).
\end{equation}
	$D_k(f)$ is an element of the algebra $\mathcal{W}(M,\mathbf{b})$ and we go to prove that it can be expanded as a linear combination of elements of the form $A^l$.\\
	We proceed by induction in $k$. The thesis holds for $k=1$ and $C_{1}=0$ since, using again Axiom 1, $D_1(f)=0$ for all $f\in\mathscr{D}(V^*M)$. Suppose now that \eqref{formula_generic}  holds for $k-1$ with respect to  some functions $C_i\colon\Gamma(BM)\rightarrow\mathscr{E}(S^iV^*M)$, $i=1,2,\ldots,k-1$, that satisfy all the desired properties. We intend to establish the validity of the thesis also for $i=k$.
	Consider the Wick polynomial, for $f\in\mathscr{D}(S^kVM)$, 
	\begin{equation*}
	W_k(f) :=\sum_{l=1}^{k-1}A^l(C_{k-l}[M, \mathbf{b}]\cdot_{k-l}f).
	\end{equation*}
	We stress that the sections 	$C_{k-1}[M, \mathbf{b}],C_{k-2}[M, \mathbf{b}],\ldots, C_1[M, \mathbf{b}]$ appearing  in the sum,  by hypotheses are smooth and have all the desired properties stated in the theorem. Writing $C_{k-1}$ in place of $C_{k-1}[M, \mathbf{b}]$, we have:
	\begin{flalign*}
	&\left[D_k(f)-W_k(f), A(g)\right]=iD_{k-1}(\Delta(g)\cdot_1 f)-\sum_{l=1}^{k-1}\left[A^l(C_{k-l}\cdot_{k-l}f),A(g)\right]\\
	&=i\sum_{l=0}^{k-2} A^l\left(C_{k-1-l}\cdot_{k-1-l}\Delta(g)\cdot_1f\right)-\sum_{l=1}^{k-1}iA^{l-1}(\Delta(g)\cdot_1C_{k-l}\cdot_{k-l}f)\\
	&=i\sum_{l=1}^{k-1}A^{l-1}\left(C_{k-l}\cdot_{k-l}\Delta(g)\cdot_1f\right)-A^{l-1}\left(\Delta(g)\cdot_1C_{k-l}\cdot_{k-l}f\right)=0,
	\end{flalign*}
	where we have used Proposition~\ref{contrac_associativity}. Thus, we can conclude that $[D_k(f) - W_k(f), A(g)]= 0$ for any test function $g$.  Due to
	Axiom 3 we must therefore have 
	\begin{equation}
	\label{C_k_costruction}
	D_k(f)-W_k(f)=C_k(f)1,
	\end{equation}
	where $C_k(f)$ is real and must define a locally-covariant $c$-number $S^kV$-field since it is a difference of that type of fields. 

	Next, we will appeal to the Peetre-Slov\'ak theorem to characterize
	the dependence of $C_k(f)$ on the background field $\mathbf{b}$. This
	theorem has two main hypotheses: \emph{locality} and \emph{weak
	regularity}, which we verify by the covariance and smoothness axioms,
	respectively. Let us consider any smooth variation of the background
	field $\mathbf{b}_s$, together with a corresponding family of
	distinguished states $\omega^{(s)} = \omega \circ \tau_s^{-1}$. Then,
	the smoothness axiom implies that the left-hand side of
	\begin{equation*}
	\omega^{(s)}\left(	D_{k,\mathbf{b}_s}(f)-W_{k,\mathbf{b}_s}(f)\right)
		= C_{k,\mathbf{b}_s}(f)\:,
	\end{equation*}
	is a distribution with smooth real kernel (it is real due to the \textit{hermiticity} axiom), meaning that so is
	$C_{k,\mathbf{b}_s}(f)$, with $C_k[M, \mathbf{b}_s] \in
	\mathscr{E}(S^kV^*M)$ denoting its integral kernel for fixed $s$.
	Moreover, the axiom also implies that $C_k[M,\mathbf{b}_s](x)$ is
	jointly smooth in $(x,s)$ and hence weakly regular. On the other hand,
	the covariance requirement implies that the dependence of
	$C_k[M, \mathbf{b}](x)$ at any point $x\in M$ is local, for any fixed
	manifold $M$. Namely, using smaller and smaller neighborhoods $M' \ni
	x$ with $M' \subseteq M$ viewed as background geometries  on their own
	right (when equipped with the restriction of $\mathbf{b}$ to $M'$),
	covariance with respect to the inclusion embeddings $\chi \colon M'
	\ni x' \mapsto x' \in M$ implies that $C_k[M,\mathbf{b}](x)$ depends
	only on the germ of $\mathbf{b}$ at $x$.

	Thus the map $C_k[M,\cdot]$ is local and weakly regular. The Peetre-Slov\'ak theorem implies that   $\Gamma(BM) \ni  \mathbf{b} \mapsto C_k[M,\mathbf{b}] \in\mathscr{E}(S^kV^*M)$
	  is a differential operator of locally bounded order.
	Summing up,  we have proved that
	\begin{equation*}
	 \widetilde{A}^k(f) - A^k(f) = D_k(f)=W_k(f)+C_k(f)1=\sum_{l=0}^{k-1} A^l\left(C_{k-l}\cdot_{k-l}f\right) ,
	\end{equation*}
		where {\em all} coefficients $C_l[M, \mathbf{b}]$ from $l=0$ to $l=k$ have all properties stated in the thesis, but the scaling property which must be  still  established  for $C_k$ only. Choosing as test function $f=f_1^{p_1}\odot\cdots\odot f_N^{p_N}$ and using relation \emph{(c)} and \emph{(d)} from Proposition~\ref{prop_contrac}, we obtain immediately the formula \eqref{formula_generic_comp}.
		
	Proceeding again by induction, thanks to the \emph{scaling} property of $A^P$, $\widetilde{A}^P$ and $(C_{k-l}^Q)_{l=1}^{l=k-1}$, $C_k^Q$ is a linear combination of products of terms with almost homogeneous degree that add up to $\langle Q,\mathbf{d}_A\rangle$. Thus, by Lemma~\ref{lemma_product}, $C_k^Q$ itself has almost homogeneous degree $\langle Q,\mathbf{d}_A\rangle$, and thus
	\begin{equation*}
	S_\lambda C_k^Q 1 = \lambda^{\langle Q,\mathbf{d}_A\rangle} C_k^Q 1 + \lambda^{\langle Q,\mathbf{d}_A\rangle}\sum_{j=1}^{l}\left(\log^j\lambda\right)B_j^Q,
	\end{equation*}
	where $S_\lambda$ is the action of physical scalings on fields here applied to a $c$-number field, with $B_j^Q$ some other quantum fields of almost homogeneous degree $\langle Q,\mathbf{d}_A\rangle$. Using again the kinematic completeness of $A$, we find that $B_{j}^Q=F_j^Q1$ are also $c$-number fields. Now, exploiting the definition of $S_\lambda$ as in \eqref{S_lambda}, we find that $S_\lambda C_k^Q 1 = C_k^Q 1$, and similarly for the $F_j^Q$. Therefore, we find that for every $x\in M$,
	\begin{equation*}
	C_k^Q[M,\mathbf{b}_\lambda](x)=\lambda^{\langle Q,\mathbf{d}_A\rangle} C_k[M,\mathbf{b}](x)+ \lambda^{\langle Q,\mathbf{d}_A\rangle}\sum_{j=1}^{l}\left(\log^j\lambda\right)F_j^Q[M,\mathbf{b}](x),
	\end{equation*}
	is an almost homogeneous element of degree $\langle Q,\mathbf{d}_A\rangle$ of the vector space of maps $\Gamma(BM)\rightarrow\mathscr{E}(\odot_{i=1}^NS^{q_i}W^*_iM)$ under the action $F[M,\mathbf{b}] \mapsto F[M,\mathbf{b}_\lambda]$. The proof is concluded.
\end{proof}

\noindent We have finally obtained a general formula, \eqref{formula_generic},  that classifies all finite renormalizations counterterms  of Wick powers of  a generic locally-covariant   boson vector field $A$, 
where the coefficients $C_k[M,\mathbf{b}]$ depends on the type of vector bundle $VM$ and the nature of background fields $\mathbf{b}$
 of the field $A$. For this reason, in order to study in detail these coefficients, we have to consider physically relevant models.

 \section{Tensor fields and renormalizations of their Wick powers}
 \label{section_tensor_field}
 In this section we consider a class of physically relevant models and we study in detail the renormalization counterterms $C_k$ introduced in the last section. We choose as bundles
 \begin{equation*}
 VM=\bigoplus_{i=1}^N T^{*\otimes k_i}M,\qquad BM= \mathring{S}^2T^*M\oplus\left(\bigoplus_{j=1}^{K}  T^{*\otimes l_j}M\right)
 \end{equation*}
 which means that we are considering as fields an $N$-tuple of tensor fields with different tensor ranks
 \begin{equation}
 \label{ntuple_fields_general}
 A=\left(A_1,\ldots, A_N\right)\qquad A_i\colon\mathscr{D}(T^{\otimes k_i}M)\longrightarrow \mathcal{W}(M,\mathbf{b})
 \end{equation}
 and we will say that $A$ has tensor rank $\mathbf{k}=(k_1,\ldots,k_N)$.
 As background fields we consider the metric $\mathbf{g}$ together with other (covariant) tensor fields $\mathbf{t}_j$ of rank $l_j$
 \begin{equation*}
 \left(\mathbf{g},\mathbf{t}_1,\ldots, \mathbf{t}_K\right),\qquad \mathbf{g}\in\mathscr{E}(\mathring{S}^2T^*M),\; \mathbf{t}_j\in\mathscr{E}(T^{*\otimes l_j}M).
 \end{equation*}
 Regarding \emph{physical scaling}, we assume the most general situation, \emph{i.e.},
\begin{equation*}
A_i\longmapsto \lambda^{d_{A_i}}A_i\qquad (\mathbf{g}, \mathbf{t}_1, \ldots,  \mathbf{t}_K)\longmapsto \left(\lambda^{-2}\mathbf{g},\lambda^{s_1}\mathbf{t}_1, \ldots, \lambda^{s_K}\mathbf{t}_K\right)\qquad \lambda\in\mathbb{R}^+
\end{equation*}
under physical scaling transformation, where $s_j\in\mathbb{R}$ for $j=1,\ldots, K$. We require also another property of the background fields, encoded in the following
\begin{defn}
	\label{def_admissible}
	A background field $\mathbf{t}_j$ is called \textbf{admissible} if its rank $l_j$ and its degree under physical scaling $s_j$ fulfill the following condition
	\begin{equation*}
	l_j+s_j\geq 0.
	\end{equation*}
	If the above relation is an equality, then we also call $\mathbf{t}_j$
	\textbf{marginal}. By convention, let us order the background fields
	such that each $\mathbf{t}_j$ for $j=1,\ldots,K_0 \le K$ is marginal
	and collectively denote them by $\mathbf{z} = (\mathbf{t}_1,\ldots,
	\mathbf{t}_{K_0})$. To emphasize their distinction from other
	background fields, we will also use the notation $\mathbf{z}_j =
	\mathbf{t}_j$.
\end{defn}

\begin{rem} \label{rem_contravar}
	We have chosen all dynamical and background tensor fields to be
	purely covariant, \emph{i.e.},\ to be sections of powers of the
	cotangent bundle $T^*M$. This choice is motivated purely by
	convenience and the desire not to complicate our notation even
	further. Our main results, Theorems~\ref{lemma_ipotesi}
	and~\ref{thm_uniqueness}, hold in easily adapted forms also for
	contravariant or mixed tensors, as well as for tensors of symmetric,
	antisymmetric, or any other symmetry type. One does have to note that,
	in the definition of \emph{admissible} and \emph{marginal} background
	fields (Definition~\ref{def_admissible}), the tensor rank $l_j$ must
	be taken to be the number of covariant tensor indices minus the number
	of contravariant indices of $\mathbf{t}_j$.
\end{rem}

 Before studying the exact form for the renormalization counterterms in this case, we need to recall some results.

\paragraph{Preparatory definitions and results.}\label{section_coordinates_bundle}
In this paragraph, we introduce various local coordinate systems on $BM$ and $J^rBM$ together with the description of a particular class of diffeomorphisms called {\em coordinate scalings}.\\
Let $(x^a)$ be a local coordinate chart on the open set $U\subseteq M$ and let $(x^d,g_{ab},\ldots, (t_j)_{a_1\ldots a_{l_j}}, \ldots)$ be the corresponding adapted local coordinates on $Z\subseteq BM$ where by definition  $Z$ projects onto $U$ and, more strongly, each fiber $BM_x$ is completely included in $Z$ if $x\in U$.

\begin{itemize}
	\item \textbf{Covariant coordinates.} 
	According to Definition~\ref{defadptedchartJ},  the chart $(x^a)$ on $U$  induces corresponding adapted local coordinates on $J^rBM$ called  \emph{covariant coordinates}
	\begin{equation*}
	\left(x^a,g_{ab,A},  (t_j)_{a_1\ldots a_{l_j},A}\right)\quad \text{on the afore-mentioned domain}\; Z^r\subseteq J^rBM,
	\end{equation*}
	where only $n(n+1)/2$ metric components are considered because $g_{ab}$ is a symmetric tensor.
	
	\item \textbf{Contravariant coordinates.} 
	Since Lorentzian metrics are non-degenerate, they admit an inverse denoted, using a standard notation, with $g^{ab}$.  We correspondingly obtain induced coordinates $g^{ab}_{A}$ on jets of the inverse-metric bundle. Using the notation $g^{AB}=g^{a_1b_1}\cdots g^{a_lb_l}$, for $|A|=|B|=l$, we define the following functions:
	\begin{equation*}
	g=\left|\det g_{ab}\right|, \qquad g^{ab,A}=g^{AB}g^{ab}_{B},\qquad (t_j)^{a_1\ldots a_{l_j},A}=g^{AB}(t_j)^{a_1\ldots a_{l_j}}_B,
	\end{equation*}
	where we have chosen fully contravariant coordinates for tensor bundles. We have then obtained the set of local \emph{contravariant coordinates}	
	\begin{equation*}
	\left(x^a,g^{ab,A},(t_j)^{a_1\ldots a_{l_j},A}\right)\quad \text{on}\; Z^r\subseteq J^rBM.
	\end{equation*}
	
	\item \textbf{Rescaled contravariant coordinates.} 
	We can obtain another coordinate set by a suitable rescaling of the previous one: we introduce a factor of the form $g^\alpha$, with $\alpha\in\mathbb{R}$, to rescale the coordinates ($n$ is the dimension of $M$):
	\begin{equation*}
	\left(x^a,g,g^{-\frac{1}{n}}g_{ab},g^{\frac{1}{n}+\frac{1}{n}|A|}g^{ab,A},g^{\frac{l_j}{n}+\frac{s_j}{2n}+\frac{1}{n}|A|}(t_j)^{a_1\ldots a_{l_j},A}\right)\quad \text{on}\; Z^r\subseteq J^rBM.
	\end{equation*}
	It should be noticed that the functions $g$ and $g^{-\frac{1}{n}}g_{ab}$ are not functionally independent due to the identity $g^{-1}=\left|\det g^{ab}\right|$. So, to make an honest the coordinate system, we (implicitly) omit one of the components of $g^{-\frac{1}{n}}g_{ab}$ and replaced by $g$.
	The relevance of the rescaled contravariant coordinates consists of the fact since $s_j$ is the (physical) scaling degree of $\mathbf{t}_j$ these coordinates without the coordinate $g$, 
	$$	\left(x^a,g^{-\frac{1}{n}}g_{ab},g^{\frac{1}{n}+\frac{1}{n}|A|}g^{ab,A},g^{\frac{l_j}{n}+\frac{s_j}{2n}+\frac{1}{n}|A|}(t_j)^{a_1\ldots a_{l_j},A}\right)$$
	are {\em invariant} under physical scaling. 
	
	\item \textbf{Curvature coordinates.} 
	Since we have a Lorentzian metric $\mathbf{g}$, we can always define the Levi-Civita connection $\nabla$ and the Riemann tensor $\mathbf{R}$. By well-known formulas, we can also regroup the second order jet coordinates of the metric into the components of the Christoffel symbols $\Gamma^a_{bc}$ and the components of the fully covariant Riemann tensor $\bar{R}_{abcd}$. An alternative way to regroup the components of the Riemann tensor is into the following fully contravariant tensor $\mathbf{S}$, with components
	\begin{equation*}
	\bar{S}^{abcd}:=g^{aa'}g^{bb'}\bar{R}_{a'\quad b'}^{\;\;\;\;(c\quad d)}.
	\end{equation*}
	We denote by $\Gamma^a_{bc,A}$ the components of the $\partial_A$ coordinate derivatives of $\Gamma^a_{bc}$, by $\bar{S}^{abcd,A}$ the components of the symmetrized contravariant $\nabla^A=\nabla^{(a_1}\cdots\nabla^{a_l)}$ derivatives of $\mathbf{S}$, with $(\bar{t}_j)^{a_1\ldots a_{l_j},A}$ the components of the symmetrized contravariant derivatives of $(t_j)^{a_1\ldots a_{l_j}}$. It is possible to prove~\cite{at-report,at} that
	\begin{equation*}
	\left(x^a, g_{ab},\Gamma^a_{(bc,A)},\bar{S}^{ab(cd,A)},(\bar{t}_j)^{a_1\ldots a_{l_j},A}\right)
	\end{equation*}
	defines a complete coordinate system on $Z^r\subseteq J^rBM$, which we call \emph{curvature coordinates}.
	
	\item \textbf{Rescaled curvature coordinates.}
	Analogously to rescaled contravariant coordinates, we can rescale the curvature coordinates obtaining a new coordinate system 
	\begin{equation*}
	\left(x^a,g,g^{-\frac{1}{n}}g_{ab},\Gamma^a_{(bc,A)},g^{\frac{3}{n}+\frac{1}{n}|A|}\bar{S}^{ab(cd,A)},g^{\frac{l_j}{n}+\frac{s_j}{2n}+\frac{1}{n}|A|}(\bar{t}_j)^{a_1\ldots a_{l_j},A}\right).
	\end{equation*}
	As before, removing $g$ form the set of rescaled curvature coordinates we find a set of coordinates which is fixed under physical scaling (since $s_j$ is the scaling degrees of $\mathbf{t}_j$).

	\item We call a diffeomorphism $M \to M$ a \textbf{coordinate scaling} around of $p\in M$ if, in a neighborhood of $p$
	whose closure is included in the domain  $U\subseteq M$ of local coordinates $(x^d)$ centered at $p$ itself, it acts as
	\begin{equation*}
	x^a\longmapsto \mu^{-1} x^a \quad (a=1,\ldots, n)
	\end{equation*}
	for some $\mu>0$,
	and smoothly extends  to the identity before reaching the boundary of $U$.
	More precisely, defining $t:=-\ln \mu$, the \textbf{class of coordinate scaling} around $p$ is represented by the one-parameter group of diffeomorphisms $\{\phi_t\}_{t\in \mathbb R}$  of the whole  $M$ leaving $p$ fixed generated by the globally defined vector field $V^a = -h x^a \frac{\partial}{\partial x^a}$,
	where $h \in C_0^\infty(M)$ vanishes before reaching the boundary of $U$ and attains the constant value $1$ in a neighborhood of $p$.
	We stress that, unlike physical scaling, these transformation are induced by diffeomorphisms of $M$.
\end{itemize}

\begin{lem} \label{lem:no-neg-weight}
	Consider an admissible background field $\mathbf{t}_j$. Then all its rescaled coordinates have positive or null scaling weight under coordinate scaling. In particular the rescaled coordinates scale as
	\begin{equation*}
	g^{\frac{l_j}{n}+\frac{s_j}{2n}+\frac{1}{n}|A|}(\bar{t}_j)^{a_1\ldots a_{l_j},A}\longmapsto \mu^{l_j+s_j+|A|}g^{\frac{l_j}{n}+\frac{s_j}{2n}+\frac{1}{n}|A|}(\bar{t}_j)^{a_1\ldots a_{l_j},A}
	\end{equation*}
\end{lem}
\begin{proof}
	Under coordinate scaling we have the following rescaling
	\begin{equation*}
		(\bar{t}_j)^{a_1\ldots a_{l_j},A}\mapsto \mu^{-l_j-|A|}(\bar{t}_j)^{a_1\ldots a_{l_j},A}\qquad g\mapsto \mu^{2n}g.
	\end{equation*}
	Then the result follows immediately.
\end{proof}

We are finally ready to state and prove our main result, which generalizes Theorem~3.1 of~\cite{KM16}.

\begin{thm}
	\label{thm_uniqueness}
	Let $\{ A^k \}_{k=1,2,\ldots}$ and $\{ \tilde{A}^k \}_{k=1,2,\ldots}$ be two families of Wick powers (Definition~\ref{def_wick}) of the same locally-covariant quantum $V$-field $A$ of homogeneous scaling degree $\mathbf{d}_A \in \mathbb{R}^N$ and tensor rank $\mathbf{k}$ (same as in Theorem~\ref{lemma_ipotesi}),
	where the natural vector bundle $V = \bigoplus_{i=1}^N W_i$ is the $N$-tuple introduced in~\eqref{ntuple_fields_general}. Assume also that all background fields $\mathbf{b}$, sections of the bundle $BM$ (Definition~\ref{def_bkg}), are admissible (Definition~\ref{def_admissible}).
	Recall also from Theorem~\ref{lemma_ipotesi} the renormalization coefficients $C_k$, $k=1,2,\ldots$ (with $C_1 =0$) appearing in \eqref{formula_generic} when comparing two families of Wick powers of $A$. Finally, recall the notation $R_{abcd}$ for the Riemann tensor, $\nabla_a$ for the Levi-Civita connection of $g_{ab}$, and $\epsilon^{a_1\cdots a_n}$ for the associated Levi-Civita tensor.

	Then the following facts hold:
	\begin{enumerate}
		\item[\textbf{(a)}] If $Q=(q_1,\ldots, q_N)$ is a multi-index with $|Q|=k$ such that $\langle Q, \mathbf{d}_A + \mathbf{k} \rangle = \sum_{i=1}^N q_i (d_{A_i} + k_i) < 0$, then the corresponding component $C_k^Q$ of the renormalization coefficient $C_k$ vanishes.
		\item[\textbf{(b)}] If $A$ is locally covariant with respect to the category $\mathfrak{BkgG}$ (Definition~\ref{def_lc}), then for every background geometry $(M,\mathbf{b})$, every $x \in M$ and each $k=1,2, \ldots$, the renormalization coefficients $C_k$ are given by differential operators of globally bounded order
		\begin{multline*}
		C_k[M,\mathbf{b}](x) =C_k\Big(g^{ab}(x),R_{abcd}(x),\ldots,\nabla_{e_1}\cdots\nabla_{e_h}R_{abcd}(x),\ldots\\
		\ldots(t_j)_{a_1\cdots a_{l_j}}(x),\ldots,\nabla_{e_1}\cdots\nabla_{e_r}(t_j)_{a_1\cdots a_{l_j}}(x),\ldots\Big) ,
		\end{multline*}
		where $C_k({\cdots})$ is a tensor field covariantly constructed from its arguments, whose structure is described in more detail below.
		\item[\textbf{(c)}] If $A$ is locally covariant with respect to the category $\mathfrak{BkgG}^+$ (Definition~\ref{def_lc}), then for every background geometry $(M,\mathbf{b})$, every $x \in M$ and each $k=1,2, \ldots$, the renormalization coefficients $C_k$ are given by differential operators of globally bounded order
		\begin{multline*}
		C_k[M,\mathbf{b}](x) =C_k\Big(g^{ab}(x),\varepsilon^{a_1\cdots a_n},R_{abcd}(x),\ldots,\nabla_{e_1}\cdots\nabla_{e_h}R_{abcd}(x),\ldots\\
		\ldots(t_j)_{a_1\cdots a_{l_j}}(x),\ldots,\nabla_{e_1}\cdots\nabla_{e_r}(t_j)_{a_1\cdots a_{l_j}}(x),\ldots\Big) ,
		\end{multline*}
		where $C_k({\cdots})$ is a tensor field covariantly constructed from its arguments, whose structure is described in more detail below.
	\end{enumerate} 

	In both (b) and (c), by \emph{covariantly constructed} we mean that the $C_k$ are \emph{equivariant} functions of their tensorial arguments, at each $x\in M$, in the sense of Lemma~\ref{lem_gen_equiv}. That is, each $C_k({\ldots})$ is a linear combination of finitely many covariantly constructed tensors that are polynomial in $g_{ab}$, $g^{ab}$ and the rest of the tensorial arguments, with scalar coefficients that are smooth functions depending locally (Definition~\ref{def_loc_poly}) on finitely many polynomial scalars covariantly constructed from the tensor fields $\mathbf{z}$, which consist of those background tensors $\mathbf{t}_j$ that are \emph{marginal} according to Definition~\ref{def_admissible}.
 Moreover, the functional form of the $C_k({\cdots})$ does not depend on $(M,\mathbf{b})$.

 Finally, each $C_k$ can be written as $C_k=\sum_{|Q|=k}C_k^Q$ with respect to the multiplet decomposition $V = \bigoplus_{i=1}^N W_i$, where $Q$ is a multi-index and where each $C_k^Q$ is homogeneous degree under physical scalings. More precisely, it scales as
$$
	C_k^Q\mapsto \lambda^{\langle Q, \mathbf{d}_A\rangle}C_k^Q,
$$
 when its arguments are rescaled according to
$$
	(t_j)_{a_1\cdots a_{l_j}}\mapsto \lambda^{s_j}(t_j)_{a_1\cdots a_{l_j}}\:, \quad g^{ab}\mapsto \lambda^2 g^{ab}\:, \quad \varepsilon^{a_1\cdots a_n}\mapsto \lambda^n\varepsilon^{a_1\cdots a_n}\:, \quad R_{abcd}\mapsto\lambda^{-2}R_{abcd}
$$
 (the covariant derivatives are fixed under rescaling). These scaling properties fix the upper bound on the differential and polynomial order of $C_k^Q$.
\end{thm}

\begin{rem} \label{rem_proof_structure}
Before going on to the proof, how it resembles and differs from the
proof of Theorem~3.1 of~\cite{KM16}, which proved a similar result but
only for scalar dynamical and background fields. Generally speaking, the
structure of the two proofs are similar, which are broken down into
roughly the same number of steps, roughly in the same sequence. In both
cases, we start out by knowing that the renormalization coefficients
$C_k[M,\mathbf{b}]$ are differential operators of locally bounded order.
Hence, each $C_k$ is given by a smooth function defined on the jet
bundle $J^rBM$ of the background fields, at least when applied to
sections $\mathbf{b}$ whose jets fall into some open neighborhood in the
jet bundle. The remaining steps gradually fix the structure of the $C_k$
more and more rigidly, while also expanding its domain of definition on
$J^rBM$, ultimately extending it to the entire jet bundle and thus
showing that it is of \emph{globally bounded} differential order. The
structure of $C_k$ is first restricted by appealing to its properties
under physical scaling, using results from
Appendix~\ref{section_scaling}. One immediate difference in the new
proof is the need to keep track of different (both physical and
coordinate) scaling weights for the different components of the $C_k$.
Next, the structure of the $C_k$ is further restricted by its local
covariance, meaning that it commutes with diffeomorphisms. The results
from Appendix~\ref{section_coord_scal}, provide the necessary tools for
that, which essentially consist of a strengthened version of the Thomas
Replacement Theorem reported in~\cite[Prop.2.6]{KM16}. Finally, local
covariance is once again used to fix the final form of the $C_k$, by
using the results of Appendix~\ref{section_equivar}, which essentially
strengthen the classification of equivariant and isotropic tensors
reported in~\cite[Prop.2.7,Lem.2.8]{KM16}. These supporting results
needed to be strengthened, compared to the ones used in~\cite{KM16},
because of the transition from only scalar dynamical and background
fields to tensorial ones. While, the results reported in
Appendices~\ref{section_coord_scal} and~\ref{section_equivar} are not
original, they are somewhat difficult to locate in the existing
literature. Thus, when possible to do so in a reasonably concise and
elementary manner, they are reported with proofs and references to more
specialized literature.
\end{rem}
 
\begin{proof}
	 We already know that, from the Peetre-Slov\'ak theorem (see Section~\ref{section_peetre}), the coefficients $C_k^Q$
	define a differential operator
	\begin{equation*}
	\Gamma(BM) \ni  \mathbf{b} \mapsto C_k^Q[M,\mathbf{b}] \in   \mathscr{E}\left(\bigodot_{i=1}^N S^{q_i}T^{*\otimes k_i}M	\right) \subseteq  \Gamma\left(\bigodot_{i=1}^N S^{q_i}T^{*\otimes k_i}M\right)
	\end{equation*}
	of locally bounded order  as established in Theorem~\ref{lemma_ipotesi}. 
	The rest of the proof is broken down into five steps, which are described in more detail below. 
	
	\paragraph{1. Physical scaling.} We now take advantage from almost homogeneity under physical scaling of the components of the coefficients $C_k^Q$ to find their functional form. Consider a Lorentzian manifold $M$ endowed with a metric $\mathbf{g}_0$, as well as a point $y\in M$ and an open neighborhood of $U$ of $y$ with compact closure. 
	We also assume that $\mathbf{g}_0$ restricted to $U$ is flat.
	Consider also a coordinate system $(x^d)$ on $U$ centered at $y$. These coordinates induce  adapted local coordinates on $Z^r \subseteq J^rBM$, which we write as
	\begin{equation*}
	\left(x^a,g,g_{ab},g^{ab,A},(t_1)^{a_1\ldots a_{l_1},A},\ldots, (t_K)^{a_1\ldots a_{l_K},A}\right).
	\end{equation*}
	Recall that the coordinates $(g,g^{ab})$ are functionally independent up to the identity $\left|\det g_{ab}\right|=g$. 
	We already know that $\mathbf{b} \mapsto C_k^Q[M, \mathbf{b}](x)$ is a differential operator of locally bounded order, thus for 
	$\mathbf{g}_0$, $y$ and $U$ (defined as above), there exists an integer $r\geq 0$ such that $C_k^Q$ is a differential operator on $U$ of local order $r$ when acting on sections of $BM$ close to $\mathbf{b}_0 :=(\mathbf{g}_0,\mathbf{t}_j=0)$\footnote{We stress that the flatness assumption on $\mathbf{g}_0$ is not a strong requirement because the flat metric is only the section with respect to which we consider variations. At the moment we can consider only metrics in a neighbourhood of $\mathbf{g}_0$ but we will gradually enlarge it to the whole set of Lorentzian metrics. A similar argument is also valid for all $\mathbf{t}_j$, which at the moment have to be close to the sections $\mathbf{t}_j=0$.}. In other words, there exists a neighbourhood $Z^r_1 \subseteq Z^r  \subseteq J^rBM$ of $j_y^r \mathbf{b}_0$, projecting onto $U$, and a function 
	\begin{flalign*}
	F_k^Q\colon Z^r_1&\longrightarrow \left(\odot_{i=1}^N S^{q_i}T^{*\otimes k_i}M\right) \\
	\left(x^a, g, g^{ab,A},(t_j)^{a_1\ldots a_{l_j},A}\right)&\longmapsto F_k^Q\left(x^a, g, g^{ab,A},(t_j)^{a_1\ldots a_{l_j},A}\right),
	\end{flalign*}
	such that
	\begin{equation}
	\label{C_k_diff}
	C_k^Q[M,\mathbf{b}](x)=F_k^Q\left(j^r\mathbf{b}(x)\right),
	\end{equation}
	for any section $\mathbf{b}\in\Gamma(BM|_U)$ such that $j^r\mathbf{b}(U)\subseteq Z^r_1$. Without loss of generality, but possibly shrinking the domain of $F_k$, we choose it in such a way that
	\begin{equation*}
	\begindc{\commdiag}[50]
	\obj(0,7)[1]{$Z^r_1\simeq$}
	\obj(4,7)[2]{$U$}
	\obj(6,7)[3]{$\times$}
	\obj(9,7)[4]{$W^r_1$}
	\obj(16,0)[5]{\small$\left(g, g^{ab,A},(t_j)^{a_1\ldots a_{l_j},A}\right)$}
	\obj(4,0)[6]{\small$\left(x^a\right)$}
	\mor{4}{5}{}
	\mor{2}{6}{}
	\enddc.
	\end{equation*}
	At the moment, we are very far from arguing  that  $Z^r_1 = J^rBM$ especially  because, using the Peetre-Slov\'ak theorem we only know  that the order of the differential operator $C_k$ is locally bounded and a finite global bound may not exist. During the proof we will gradually enlarge the domain $Z^r_1$ to eventually cover all of $J^rBM$ while maintaining the identity \eqref{C_k_diff}. The differential order $r$ of $C_k^Q$ may increase in the process, but will remain finite. These extensions  will be labeled  by an increasing index $j$ in  $Z^r_j$. Presently   $j=1$.\\
	Theorem~\ref{lemma_ipotesi} implies that  $C_k^Q$ and hence the function $F_k^Q$ scales almost homogeneously with degree $\langle Q,\mathbf{d}_A\rangle$ under physical scaling of the background fields. Thus, thanks to 
	Lemma~\ref{lem_inf_scaling} and Lemma~\ref{lem_functional_form_1}, there exists an integer $l>0$ and function $B_j$ on $Z^r_1$, for $h=0,\ldots,l$, such that 
	\begin{equation}
	\label{F_k_scaling}
	F_k^Q=g^{-\frac{\langle Q,\mathbf{d}_A\rangle}{2n}}\sum_{h=0}^l\log^h\left(g^{-\frac{1}{2n}}\right)B_h,
	\end{equation}
	where each $B_h$ is invariant under the action of physical scaling. Therefore,  adopting  rescaled coordinates (which are invariant under physical scaling),  $B_h$ cannot depend on $g$ and can  be written as
	\begin{equation*}
	B_h=B_h\left(x^d,g^{-\frac{1}{n}}g_{ab},g^{\frac{1}{n}+\frac{1}{n}|A|}g^{ab,A},\ldots , g^{\frac{l_j}{n}+\frac{s_j}{2n}+\frac{1}{n}|A|}(t_j)^{a_1\ldots a_{l_j},A},\ldots\right).
	\end{equation*}
	We now extend the domain $Z^r_1$ to a larger domain $Z^r_2\subseteq J^rBM$. We define $Z^r_2$ to be the smallest domain invariant under physical scaling and containing $Z^r_1$. That is, we can write it as 
	\begin{equation*}
	\begindc{\commdiag}[50]
	\obj(0,7)[1]{$Z^r_2\simeq$}
	\obj(5,7)[2]{$\mathbb{R}^+$}
	\obj(8,7)[3a]{$\times$}
	\obj(11,7)[4a]{$U$}
	\obj(14,7)[3]{$\times$}
	\obj(17,7)[4]{$W^r_2$}
	\obj(24,1)[5]{}
	\obj(35,0)[7]{\small$\left(g^{-\frac{1}{n}}g_{ab},g^{\frac{1}{n}+\frac{1}{n}|A|}g^{ab,A},g^{\frac{l_j}{n}+\frac{s_j}{2n}+\frac{1}{n}|A|}(t_j)^{a_1\ldots a_{l_j},A}\right)$}
	\obj(5,0)[6]{\small$\left(g\right)$}
	\obj(11,0)[6a]{\small$\left(x^d\right)$}
	\mor{4}{5}{}
	\mor{2}{6}{}
	\mor{4a}{6a}{}
	\enddc.
	\end{equation*}
	Up to now, we know that the identity \eqref{C_k_diff} holds only when the germ of $\mathbf{b}$ at $x\in M$ projects onto one of the jets in the domain $Z^r_1\in J^rBM$, but the function $F_k^Q$, via formula \eqref{F_k_scaling}, has a \emph{unique} extension to $Z^r_2$ that scales almost homogeneously and agrees with $F_k^Q$ on $Z^r_1$. The identity \eqref{C_k_diff} must remain valid also for germs  at $x$ that projects onto $Z^r_2$ since any element of $\mathbf{b}' \in Z^r_2$, using the action of physical scaling, can be brought back to $\mathbf{b} \in Z^r_1$, \emph{i.e.},\ $\mathbf{b}'  = \mathbf{b}_\lambda$ for some $\lambda>0$.
	Since $C_k^Q[M,\mathbf{b}]$ scales almost homogeneously and is already defined on $Z^r_2$, we conclude that it must coincide there with the unique extension of $F_k^Q$.

		\paragraph{2. Diffeomorphism covariance.} We consider now the covariance properties of the coefficient $C_k^Q$ under diffeomorphisms. In the previous paragraph, we fixed a point $y\in M$ and a fixed background geometry $\mathbf{b}_0$. But, since that choice was arbitrary, all the same results are also valid for any other choice of $y'\in M$, open neighborhood $U'\subseteq M$ of $y'$ and background geometry $\mathbf{b}'_0$, so long as $\mathbf{b}'_0=\chi^*\mathbf{b}_0$ on $U'$, where $\chi\colon M\rightarrow M$ is some diffeomorphism such that $\chi(y')=y$. The Peetre-Slov\'ak theorem gives us a differential operator of order $r'$ on a domain $Z'^{r'}_2\subseteq J^{r'}BM$. The diffeomorphism covariance of $C_k^Q$ then implies that the order may be chosen the same, $r=r'$.\\
	We now extend the domain $Z^r_2$ to a larger domain $Z^r_3 \subseteq J^rBM$. We define $Z^r_3$ to be the smallest domain invariant under $\mathrm{Diff}(M)$ and containing $Z^r_2$.
	Since the coefficient $C_k^Q$ is $\mathrm{Diff}(M)$-covariant, the function $F_k^Q$ is itself $\mathrm{Diff}(M)$-covariant on $Z^r_3$. The case of orientation preserving diffeomorphisms $\mathrm{Diff}^+(M)$ is strictly analogous.\\  Now we can use  
	Thomas replacement theorem (Theorem~\ref{thomas_rep_thm})  in order to eliminate the dependence of $F_k^Q$ on some of the coordinates on $Z^r_3$. We  apply Theorem~\ref{thomas_rep_thm} separately to the various functions $B_h$  appearing   in \eqref{F_k_scaling}, obtaining
	\begin{multline*}
	g^{-\frac{\langle Q,\mathbf{d}_A\rangle}{2n}}B_h\left(x^a,g^{-\frac{1}{n}}g_{ab},g^{\frac{1}{n}+\frac{1}{n}|A|}g^{ab,A}, g^{\frac{l_j}{n}+\frac{s_j}{2n}+\frac{1}{n}|A|}(t_j)^{a_1\ldots a_{l_j},A}\right)\\
	=g^{-\frac{\langle Q,\mathbf{d}_A\rangle}{2n}}G_h\left(g^{-\frac{1}{n}}g_{ab},g^{\frac{3}{n}+\frac{1}{n}|A|}\bar{S}^{ab(cd,A)}, g^{\frac{l_j}{n}+\frac{s_j}{2n}+\frac{1}{n}|A|}(\bar{t}_j)^{a_1\ldots a_{l_j},A}\right),
	\end{multline*}
	where each $g^{-\frac{\langle Q,\mathbf{d}_A\rangle}{2n}}G_j$ is equivariant under the action of $GL(n)$ (respectively $GL^+(n)$). In particular, $G_j$ does not depend on the coordinates $(x^a)$ and $(\Gamma^a_{(bc,A)})$.
	Since $Z^r_3$ is $\mathrm{Diff}(M)_U$-invariant  it has the structure:
	\begin{equation*}
	\begindc{\commdiag}[50]
	\obj(0,8)[1]{$Z^r_3\simeq$}
	\obj(5,8)[2]{$U$}
	\obj(8,8)[3]{$\times$}
	\obj(11,8)[4]{$L_n$}
	\obj(14,8)[5]{$\times$}
	\obj(17,8)[6]{$\mathbb{R}^\gamma$}
	\obj(20,8)[7]{$\times$}
	\obj(23,8)[8]{$W_3^r$}
	\obj(5,0)[9]{\small$\left(x^a\right)$}
	\obj(11,0)[10]{\small$\left(g_{ab}\right)$}
	\obj(17,0)[11]{\small$\qquad\left(\Gamma^a_{(bc,A)}\right)$}
	\obj(30,1)[12]{}
	\obj(43,0)[13]{\small$\left(g^{\frac{3}{n}+\frac{1}{n}|B|}\bar{S}^{ab(cd,B)}, g^{\frac{l_j}{n}+\frac{s_j}{2n}+\frac{1}{n}|A|}(\bar{t}_j)^{a_1\ldots a_{l_j},A}\right)$}
	\mor{2}{9}{}
	\mor{4}{10}{}
	\mor{6}{11}{}
	\mor{8}{12}{}
	\enddc.
	\end{equation*}
	The second factor describes the metrics at a fixed point $p \in U$ and coincides with the full set $L_n$  of non-degenerate bilinear forms on $\mathbb{R}^n$ with Lorentzian signature. This is 
	because 	the fiber action of the subgroup of  $\mathrm{Diff}(M)$ which leaves $p$ fixed is the action  of the whole $GL(n)$ which, in turn, acts transitively on $L_n$. 
	$W^r_3$ contains at least the point with all components $g^{\frac{3}{n}+\frac{1}{n}|B|}\bar{S}^{ab(cd,B)}$ and $ g^{\frac{l_j}{n}+\frac{s_j}{2n}+\frac{1}{n}|A|}(\bar{t}_j)^{a_1\ldots a_{l_j},A}$ vanishing (in particular because $\mathbf{g}_0$ is flat on $U$).	
	$W^r_3$  is invariant under the said natural action of the whole $GL(n)$. The same argument 
	applies  to $\mathrm{Diff}^+(M)$ and $GL(n)^+$.

	\paragraph{3. Covariance under coordinate scaling.}
	Now we use the equivariance of the function $F_k^Q$ under the action of a subgroup of $GL(n)$ (respectively $GL^+(n)$), the subgroup of coordinate scaling.%
		\footnote{This part of the proof is analogous to part~\emph{4.}\ of
		the proof of Theorem~3.1 in~\cite{KM16}. Unfortunately, that earlier
		argument contained an error when eliminating logarithmic terms from
		$F_k$. This error has been corrected in the current argument, which
		should also be considered retroactively inserted into the proof
		given in~\cite{KM16}.} %
	We can rewrite the set of coordinates over $L_n \times W_3^r$ (the remaining coordinates  $x^a$
	and  $\Gamma^a_{(bc,A)}$ of $Z_3^r$
	do not appear in the explicit form of $F_k^Q$ as already established)
	as
	\begin{equation*}
	\left(g, g^{-\frac{1}{n}}g_{ab},z^j,q^i\right)
	\end{equation*}
	Here the coordinates are grouped together along with the following idea: $(g^{-\frac{1}{n}}g_{ab},z^j)$ have weight $0$ under coordinate scaling (\emph{i.e.},\ $z^j$, $j=1,\ldots,m_z$, are precisely the rescaled components of those undifferentiated coordinates of the background fields $\mathbf{t}_i$ satisfying $l_i+s_i=0$, or precisely the components of the marginal background fields $\mathbf{z}$, Definition~\ref{def_admissible}), 
	$g$ transforms as $g\to \mu^{2n}g$
	and all remaining coordinates, here denoted by  $q^i$, $i=1,\ldots,m_q$,  have positive weight  ($d_i>0$) under coordinate scalings.  
	There are no coordinates with negative weight (Lemma~\ref{lem:no-neg-weight}).	\\
	Lets recall that $F_k^Q$ is a $(\bigodot_{i=1}^N S^{q_i}T^{*\otimes k_i}M)$-valued function and that the vector $\mathbf{k}=(k_1,\ldots, k_N)$ is constructed with the tensor ranks $k_i$. Then the general diffeomorphism equivariance of the function $F_k^Q$ specialized  to coordinate scalings (centered at some base point $(x^a)$, which could be arbitrary within the domain of definition of $F^Q_k$), implies the identity
	\begin{equation}
	\begin{split}
	F_k^Q\left(g,g^{-\frac{1}{n}}g_{ab},z^j,q^i\right)
	&= \mu^{-\langle Q,\mathbf{k}\rangle}
		F_k^Q\left(\mu^{2n}g,g^{-\frac{1}{n}}g_{ab},z^j,\mu^{d_i}q^i\right)\\
	&= g^{-\frac{\langle Q,\mathbf{d}_A\rangle}{2n}}
		\sum_{h=0}^{l-1} \mu^{-\langle Q,\mathbf{d}_A+\mathbf{k}\rangle}
		\log^h\left(\mu^{-1}g^{-\frac{1}{2n}}\right)G_h\left(g^{-\frac{1}{n}}g_{ab},z^j,\mu^{d_i}q^i\right)
	\end{split}
	\end{equation}
	for any point in $Z^r_3$
	and any $\mu>0$. As we mentioned in the previous part, the limit $(g^{-\frac{1}{n}}g_{ab},z^j,0)$ of the argument of the functions $G_h$ as $\mu\rightarrow 0$, belongs to the domain of the function $G_h$, which is smooth there. Therefore we have the Taylor expansions
	\begin{equation*}
	G_h\left(g^{-\frac{1}{n}}g_{ab},z^j,q^i\right)
	= \sum_{|I|<N_q} G_{h,I}\left(g^{-\frac{1}{n}}g_{ab},z^j\right)q^I+O\left(q^{N_q}\right),
	\end{equation*}
	around $(g, g^{-\frac{1}{n}}g_{ab},z^j,0)$,
	where $I=i_1\cdots i_{m_q}$ is a multi-index with respect to the coordinates $(q^i)$, the coefficients $G_{h,I}$ are smooth, and $N_q>0$ is an integer such that
	\begin{equation*}
	\langle d,I\rangle=\sum_{j=1}^{m_q} d_ji_j>\langle Q,\mathbf{d}_A+\mathbf{k}\rangle \quad \text{for all $I$ such that $|I|\geq N_q$.}
	\end{equation*}
	This choice guarantees that each error term $O(q^{N_q})$ is mapped to $O(\mu^{\langle Q,\mathbf{d}_A+\mathbf{k}\rangle+1})$ under the substitution $q^i \mapsto \mu^{d_i} q^i$ as $\mu\rightarrow 0$. Thus we obtain
	\begin{multline}
	\label{coord_scaling_1}
	F_k^Q\left(g,g^{-\frac{1}{n}} g_{ab}, z^j, q^i\right)
	= g^{-\frac{\langle Q,\mathbf{d}_A\rangle}{2n}} \sum_{h=0}^{l-1}
		\mu^{-\langle Q, \mathbf{d}_A+\mathbf{k}\rangle} \log^h\left(\mu^{-1} g^{-\frac{1}{2n}}\right)
		G_h\left(g^{-\frac{1}{n}}g_{ab},z^j, \mu^{d_i} q^i\right) \\
	= g^{-\frac{\langle Q,\mathbf{d}_A\rangle}{2n}} \sum_{h=0}^{l-1} \sum_{|I|<N_q}
		\mu^{\langle d, I\rangle -\langle Q, \mathbf{d}_A+\mathbf{k}\rangle} \log^h\left(\mu^{-1} g^{-\frac{1}{2n}}\right)
		G_{h,I}\left(g^{-\frac{1}{n}}g_{ab},z^j\right)q^I \\
		+\mu^{-\langle Q,\mathbf{d}_A+\mathbf{k}\rangle} O\left(\mu^{\langle Q,\mathbf{d}_A+\mathbf{k}\rangle+1}\right) .
	\end{multline}
	Now, if we take the limit $\mu\rightarrow 0$, the left-hand side of~\eqref{coord_scaling_1} does not change, being independent of $\mu$, and in particular remains bounded. Hence, for equality to hold, any term on the right-hand side of~\eqref{coord_scaling_1} that independently goes to $\oo$ as $\mu\to 0$ must vanish. That is, the coefficient of each $\mu^p \log^h \mu$ term with $p<0$ or $p=0,h>0$ must be zero. Actually taking the $\mu \to 0$ limit on the right-hand side of~\eqref{coord_scaling_1} we obtain the identity
	\begin{equation}
	\label{F_k_poly}
	F_k^Q\left(g,g^{-\frac{1}{n}} g_{ab}, z^j, q^i\right)
	=\sum_{\langle d,I\rangle=\langle Q,\mathbf{d}_A+\mathbf{k}\rangle}g^{-\frac{\langle Q,\mathbf{d}_A\rangle}{2n}} G_{0,I}\left(g^{-\frac{1}{n}}g_{ab},z^j\right)q^I.
	\end{equation}
	All terms $\langle d, I\rangle > \langle Q, \mathbf{d}_A+\mathbf{k}\rangle$ were set to zero by the limit, which by consistency means that they had zero coefficients to begin with.
	Notice that this identity implies that the function $F_k^Q$ scales \emph{homogeneously} with degree $\langle Q,\mathbf{d}_A\rangle$ (that is, it has almost homogeneous order zero).
	This sum could conceivably be empty, if it happens that $\langle Q, \mathbf{d}_A+\mathbf{k}\rangle < 0$ (recall that $d_i > 0$), which can only happen if some of the combinations $d_{A_i}+k_i < 0$. In that case, $F_k^Q = 0$ and the corresponding component $C_k^Q$ of the renormalization coefficient $C_k$ vanishes, which proves part (a) of the theorem.

	We can now enlarge again the domain of the function $F_k^Q$ along the fibers, where the identity \eqref{C_k_diff} holds, from $Z^r_3$ to $Z^r_4\subseteq J^rBM$. The new domain is isomorphic to
	\begin{equation*}
	\begindc{\commdiag}[50]
	\obj(0,8)[1]{$Z^r_4\simeq$}
	\obj(5,8)[2]{$U$}
	\obj(8,8)[3]{$\times$}
	\obj(11,8)[4]{$L_n$}
	\obj(14,8)[5]{$\times$}
	\obj(17,8)[8]{$W_4$}
	\obj(20,8)[9]{$\times$}
	\obj(23,8)[10]{$\mathbb{R}^\gamma$}
	\obj(26,8)[11]{$\times$}
	\obj(29,8)[19]{$\mathbb{R}^{m_q}$}
	\obj(5,0)[12]{\small$\left(x^a\right)$}
	\obj(11,0)[13]{\small$\;\;\left(g^{-\frac{1}{n}}g_{ab}\right)$}
	\obj(17,0)[15]{\small$\;\;\left(z^j\right)$}
	\obj(23,0)[16]{\small$\quad\left(\Gamma^a_{(bc,A)}\right)$}
	\obj(36,1)[17]{}
	\obj(57,0)[18]{\small$\left(g^{\frac{3}{n}+\frac{1}{n}|A|}\bar{S}^{ab(cd,A)},g^{\frac{l_j}{n}+\frac{s_j}{2n}+\frac{1}{n}|A|}(\bar{t}_j)^{a_1\ldots a_{l_j},A}\right)_{l_j+s_j+2|A|>0}$}
	\mor{2}{12}{}
	\mor{4}{13}{}
	\mor{8}{15}{}
	\mor{10}{16}{}
	\mor{19}{17}{}
	\enddc.
	\end{equation*}
	The function $F_k^Q$ extends uniquely to $Z^r_4$ as a covariant function under coordinate scaling. Essentially we have enlarged the factor $W^r_3$ to $W_4 \times \mathbb{R}^{m_q}$. We can do that because all the $(q^i)$ coordinates have positive weight under coordinate scaling, so that their domain can be extended to all of $\mathbb{R}^{m_q}$. The range of the $(z^j)$ coordinates is limited to $W_4 \subset \mathbb{R}^{m_z}$ because these coordinates are invariant under coordinate scaling. 
	Note that the dependence of $F_k^Q$ on the $\mathbb{R}^\delta$ factor in $Z^r_4$ is polynomial and remember that $F_k^Q$ does not depend on the factor $U\times \mathbb{R}^\gamma$ (see previous part).
	
	\paragraph{4. Global definition.}
	
	It is now the moment to expand the domain $Z_4^r$ to all $J^rBM$, for an appropriate choice of $r$. In \eqref{F_k_poly}, a generic $q^I$ is of the form
	\begin{equation*}
		\prod_{|A|,|B|}\left(\bar S^{ab(cd,A)}\right)^{p_{S,|A|}}\left((\bar{t}_j)^{a_1\cdots a_{l_j},B}\right)^{p_{j,|B|}},
	\end{equation*}
	where all the $p$-exponents are non-negative integer numbers and $p_{j,0}=0$ if $l_j+s_j=0$. The constraint $\langle d,I\rangle=\langle Q,\mathbf{d}_A+\mathbf{k}\rangle$ in \eqref{F_k_poly} can be written explicitly as
	\begin{equation*}
	\langle Q,\mathbf{d}_A+\mathbf{k}\rangle=\sum_{|A|,|B|}(2+|A|)p_{S,|A|}+(s_j+l_j+|B|)p_{j,|B|}.
	\end{equation*}
	By the admissibility of the background fields (Definition~\ref{def_admissible}), we have $s_j+l_j\ge 0$. Hence, the coefficients of the $p$-exponents are non-negative and grow linearly with $|A|$ and $|B|$. Thus, there exists a bound on the maximum values of $|A|,|B|$ with non-zero $p$-exponents. Let $r_{k}$ be the maximum number of derivatives of the curvature or background tensors for which the $p$-exponents are non-zero. Note that $r_k$ depends only on the structure of the bundle $BM$ and $k$, and not on the chosen domain $Z^r_4$. Then we can set $r=r_{k}$ in all the previous parts of the proof, \emph{i.e.},\ we end up with a domain
	\begin{equation*}
	Z_4^{r_{k}}\subseteq J^{r_{k}}BM,\qquad Z_4^{r_{k}}=U\times L_n \times W_4\times \mathbb{R}^\delta.
	\end{equation*}
	We can now extend one last time the domain $Z_4^{r_{k}}$ keeping the order of  $F_k$ globally bounded. The factor $L_n$ is already maximal since it contains all Lorentzian metrics. At the beginning of the proof we chose as the initial domain $Z_1^r$ a neighbourhood of the point $j^r_y(\mathbf{g}_0,\mathbf{t}_j=0) \in J^rBM$. Recall that we later split the coordinates on $J^rBM$ into two groups, the $q$-coordinates, identified by positive scaling weights $(s_j+l_j>0)$, and the $z$-coordinates, identified by zero scaling weights $(s_j+l_j=0)$, the components of the marginal tensor fields $\mathbf{z}$ (Definition~\ref{def_admissible}). What was essential for the subsequent arguments was that, for each allowed value of the $z$-coordinates, $(z,q=0)$ was also contained in $Z^r_1$, because $q\circ j^r_y(\mathbf{g}_0, \mathbf{t}_j=0) = 0$. However, the condition $\mathbf{z}(y) = z\circ j^r_y(\mathbf{g}_0, \mathbf{t}_j=0) = 0$ did not play a significant role. Thus, the entire proof would work without any changes had we chosen different background fields $\mathbf{t}_j$ such that still $q\circ j^r_y(\mathbf{g}_0, \mathbf{t}_j) = 0$, but $\mathbf{z}(y) = z\circ j^r_y(\mathbf{g}_0, \mathbf{t}_j)$ assuming an arbitrary value. Then, having a priori fixed $r=r_k$, the functions $F_k$ on different $Z^r_1$ domains would necessarily agree on overlaps (since they are merely local expressions of the globally defined differential operator $C_k$) and the union of all the $Z^r_1$ domains would cover arbitrary values of the $z$-coordinates. Thus, having already performed the extension of the domain into the $q$-coordinates, we can set $W^4=\mathbb{R}^{m_q}$ in $Z^r_4$. In other words we can set	
	\begin{equation}
	Z^r_4 = Z_4^{r_{k}}=\pi^{-1}\left(U\right)\label{projU2}
	\end{equation}
	for some open neighborhood $U \subset M$ of $y\in M$, where $\pi\colon J^{r_{k}}BM \rightarrow M$.

	The union of all those open sets $U$, when $y$ varies in $M$, completely covers $M$. Thus, the corresponding domains $Z^r_4$ completely cover $J^{r_k}BM$. Thus, the globally defined differential operator $C_k$ is of globally bounded order at most $r_k$ and its components $C^Q_k \colon J^{r_k}BM \to \bigotimes_{i=1}^N S^{q_i} T^{*\otimes k_i}M$ have the form~\eqref{C_k_diff} when restricted to a domain of the form $Z^{r_k}_4$ with the functions $F^Q_k$ satisfying~\eqref{F_k_poly}.

		\paragraph{5. $GL(n)$-equivariance.}
	In this last point, we intend to give a precise form of the function $F_k^Q$ exploiting their $GL(n)$-equivariance. From the previous discussion we know that the function $F_k^Q$ satisfying \eqref{C_k_diff}, is defined on the domain $Z_4^{r_k}\simeq U\times \mathbb{R}^\gamma\times Z_4$, but it \emph{depends} only on the coordinates corresponding to the factor $Z_4=L_n\times \mathbb{R}^{m_z}\times \mathbb{R}^{m_q}$. We also know the following:
	\begin{enumerate}
		\item the dependence is \emph{polynomial} with respect to the standard coordinates on the $\mathbb{R}^{m_q}$ factor;
		\item the coefficients $g^{-\frac{\langle Q,\mathbf{d}_A\rangle}{2n}} G_{0,I}(g^{-\frac{1}{n}}g_{ab},z^j)$ of these polynomials depend only on $L_n\times \mathbb{R}^{m_z}$.
	\end{enumerate}
	Each factor in $Z_4$ carries a tensor density representation of $GL(n)$ (resp.\ $GL^+(n)$) arising from the action of the subgroup of $\mathrm{Diff}(M)$ (resp.\ $\mathrm{Diff}^+(M)$) which leaves fixed a given point of $U$. More precisely, if $u\in GL(n)$:
	\begin{enumerate}
		\item on $L_n$ the action is given by $(u, \mathbf{g})\mapsto \left|\det u\right|^{-\frac{2}{n}}u^{\otimes 2}\mathbf{g}$;
		\item on $\mathbb{R}^{m_z}$, which which corresponds to the rescaled components $g^{\frac{l_j}{n}+\frac{s_j}{2n}}(\bar{t}_j)^{a_1\cdots a_{l_j}}$, $j=1,\ldots, K_0$, of the marginal background tensor fields $\mathbf{z} = (\mathbf{t}_1,\ldots \mathbf{t}_{K_0})$ (Definition~\ref{def_admissible}), the action is given by $(u,\mathbf{t}_j)\mapsto \left|\det u\right|^{\frac{l_j}{n}+\frac{s_j}{2n}}u^{\otimes l_j}\mathbf{t}_j$;
		\item on $\mathbb{R}^{m_q}$, which corresponds to the rescaled components $g^{\frac{l_j}{n}+\frac{s_j}{2n}+\frac{1}{n}|A|} (\bar{t}_j)^{a_1\cdots a_{l_j},A}$, for $l_j+s_j+2|A|>0$, and $g^{\frac{3}{n}+\frac{1}{n}|A|} \bar{S}^{ab(cd,A)}$, for $|A|\ge 0$, also decomposes into a direct sum of corresponding tensor density representations
		\begin{equation*}
		\mathbb{R}^{m_q}=\bigoplus_\alpha \mathcal{R}_\alpha,
		\end{equation*}
		where $\mathcal{R}_\alpha$ carries a tensor density representation of rank $n_\alpha$;
		\item the fibers of the bundle where the functions $F^Q_k\colon Z_4 \to \bigotimes_{i=1}^N S^{q_i} T^{*\otimes k_i} M$ take their values also carry a representation of $GL(n)$ (resp.~$GL^+(n)$), which obviously decomposes into a direct sum of tensor density representations, which we will denote by
		\begin{equation*}
			\mathcal{T} = \bigoplus_\beta \mathcal{T}_\beta ,
		\end{equation*}
		where $\mathcal{T}_\beta$ has rank $n_\beta$.
	\end{enumerate}

	Note also that, the homogeneous polynomials $\mathcal{P}^\delta$ of degree $\delta$ on $\mathbb{R}^{m_q}$ carry the representation
	\begin{equation*}
	\left(uP\right)(\rho):=P\left(u^{-1}\rho\right),\quad \mbox{for any } u\in GL(n), \;P\in\mathcal{P}^{\delta},\; \rho\in\mathbb{R}^{m_q} .
	\end{equation*}
	This representation on polynomials is made up of direct sums of symmetric tensor powers of $\mathbb{R}^{m_q}$ and hence itself also decomposes into a direct sum of tensor density representations
	\begin{equation*}
		\mathcal{P}^\delta = \bigoplus_\gamma \mathcal{S}^\delta_\gamma ,
	\end{equation*}
	where $\mathcal{S}^\delta_\gamma$ has rank $n^\delta_\gamma$.

	From the above remarks, it is easy to see that the equivariance of the functions $F^Q_k$ (see Proposition~\ref{prop_equiv_inv} for the relation between invariant and equivariant functions) and the linear independence of the monomials $q^I$ on $\mathbb{R}^{m_q}$ implies that the polynomial coefficients in~\eqref{F_k_poly} are themselves smooth equivariant maps
	\begin{equation}
		G_0 \colon L_n \times \mathbb{R}^{m_z} \to \mathcal{T} \otimes \mathcal{P}^{\langle Q,\mathbf{k}\rangle} = \bigoplus_{\beta,\gamma} \mathcal{T}_\beta \otimes \mathcal{S}^{\langle Q, \mathbf{k}\rangle}_\gamma .
	\end{equation}
	See~\cite{KM16}, point \emph{5.} of the proof of the main Theorem~3.1, for a more detailed elaboration of this argument.

Now, since the components $(G_0)_{\beta,\gamma}$ are \emph{equivariant
tensor densities} (Definition~\ref{def_equivar_tens}), we can invoke the
classification Lemma~\ref{lem_gen_equiv} to conclude that each
$(G_0)_{\beta,\gamma}$ is, up to an overall power of $g = \left|\det
g_{ab}\right|$, a tensor of appropriate rank built covariantly out of
$g_{ab}$, $g^{ab}$, $\varepsilon^{a_1\cdots
a_n}(g)$ and the tensor components of $\mathbf{z}$,
$t_j^{a_1\cdots a_{l_j}}$, for $j=1,\ldots,K_0$. To be more
precise, each $(G_0)_{\beta,\gamma}$ is a finite linear combination of
$\mathcal{T}_\beta \otimes \mathcal{S}^{\langle Q,
\mathbf{k}\rangle}_{\gamma}$ terms, each built from a tensor product of
finitely many aforementioned ingredients (possibly repeating) followed
by any number of index contractions or permutations, with coefficients
being \emph{smooth} functions of all possible polynomial scalar
invariants covariantly constructed from the same ingredients,
\begin{equation*}
	(G_0)_{\beta,\gamma}
	= g^{\alpha_{\beta,\gamma}} \sum_m
		c^m_{\beta,\gamma} (g_{ab},
			\varepsilon^{a_1\cdots a_n}(g), \ldots ,
			t_j^{a_1\cdots a_{l_j}}, \ldots)
		P^m_{\beta,\gamma} (g_{ab},
			\varepsilon^{a_1\cdots a_n}(g), \ldots ,
			t_j^{a_1\cdots a_{l_j}}, \ldots) \:.
\end{equation*}
Lemma~\ref{lem_gen_equiv} also tells us that, in each case, there are
only finitely many algebraically independent polynomial scalar
invariants that the coefficients $c^m_{\beta,\gamma}$ can depend on and
there are only finitely many tensor valued polynomials
$P^m_{\beta,\gamma}$ that are linearly independent up to a redefinition
of the $c^m_{\beta,\gamma}$ coefficients. The dependence on
$\varepsilon(g)$ is allowed only in the $GL(n)^+$ case. Also,
note that for the contractions $G_{0,I} q^I$ to remain equivariant, all
the explicit appearances of powers of $g = \left|\det g_{ab}\right|$
must cancel.

Finally, combining the above conclusions with~\eqref{F_k_poly}, we can
say that
\begin{equation}
	F^Q_k = F^Q_k \left(
		g^{ab}, g_{ab}, \epsilon^{a_1\cdots a_n}(g),
		\ldots, (z_j)_{a_1\cdots a_{l_j}}, \ldots ;
		g^{ab}, \epsilon^{a_1\cdots a_n}(g),
		\ldots, S_{ab(cd,A)}, \ldots, (t_j)_{a_1\cdots a_{l_j},A}, \ldots
	\right) \:,
\end{equation}
where the dependence on the second group of arguments is purely
polynomial, while the dependence on first group of arguments is smooth
with respect to finite set of algebraically independent scalar
polynomial invariants that can be formed from them by tensor products
and contractions. Recall that we have used the notation $\mathbf{z}_j =
\mathbf{t}_j$, for $j=1, \ldots, K_0$, that is for those background
tensor fields such that are marginal, satisfying $s_j+l_j = 0$. This
completes the proof.
\end{proof}

After the proof of this very general model, we can move on to some more physically relevant models.

\subsection{Vector Klein-Gordon field}
\label{proca_section}
We now consider a specific  quantum vector field  in order to investigate in detail the form of the coefficients $C_k$ in \eqref{formula_generic}: we focus on the  {\bf vector Klein-Gordon field}. 
The classical configurations of the vector KG field  over an oriented globally hyperbolic spacetime $(M, \mathbf{g})$ are smooth  $1$-forms, \emph{i.e.},\ sections of the cotangent bundle $T^*M$, namely
$A \in \mathscr{E}(T^*M)$. The vector KG equation, where we include also a coupling term with the curvature $R$, reads
\begin{equation}
\label{proca}
-\square_{\mathbf{g}} A+m^2A+\xi RA=0 
\end{equation}
where $m^2$ and $\xi$ are here smooth real-valued functions on $M$ (they can be constant functions, but in general we admit that $m^2$ and $\xi$ can vary on the spacetime). 
When passing to the quantum formulation, the {\bf locally-covariant quantum vector KG field}, indicated by the same symbol $A$,
is defined as in Definition~\ref{local_field} with $k=1$ and $VM=T^*M$. Moreover we have the following requirements.
\begin{itemize}
	\item[(a)] The net of local quantum observables $\mathcal{W}$ including the vector KG field is as in Definition~\ref{def_netalgebra} is fixed
	according to  equation \eqref{proca}, which suggests that the natural bundle of background fields is the one completely defined by
	\begin{equation}\label{BMProca}
	BM= \mathring{S}^2T^*M\oplus {\mathbb R} \oplus \mathbb R,
	\end{equation}
	so that the 
	sections $M\rightarrow BM$ are triples $\mathbf{b}=(\mathbf{g},m^2,\xi)$.
	(The metric $\mathbf{g}$ affects the theory because it enters $\square_\mathbf{g}$, $R$ (also derived) and even the Levi-Civita tensor $\epsilon$ in case one deals with the category of background geometries $\mathfrak{BkgG}^+$ instead of $\mathfrak{BkgG}$).

	\item[(b)] The natural vector bundle is completely fixed by requiring 
	\begin{equation*}
	VM=T^*M 
	\end{equation*} 
	and the morphism $V_{\chi}$, whose associated  pushforward on test sections is exploited to define the notion of local covariance as in Definition~\ref{local_field}, is nothing but the natural lift of the embeddings $\chi \colon M \to M'$ to the corresponding tangent bundles.

	\item[(c)] According to its mass dimension,%
		\footnote{E.g., assuming that both the terms summed in the Lagrangian density of the vector KG field field
		$m^2 g^{ab} A_a A_b \sqrt{g}$ and $g^{ab} (\nabla A)_a (\nabla A)_b
		\sqrt{g}$ are dimensionless once supposed $\hbar=c=1$.} %
	the physical scaling degree of the vector KG field is  $\mathbf{d}_A = (n-4)/2$ when $\mathbf{g}\mapsto \lambda^{-2} \mathbf{g}$, $m^2\mapsto \lambda^2 m^2$ and $\xi\mapsto \xi$ according to \eqref{ps}. We recall that the presence of covariant derivatives do not change this rescaling behaviour as the coordinates are dimensionless.
	
	\item[(d)] We stress that all background fields of this model are scalars of non-negative physical scaling weight and hence are admissible according to Definition~\ref{def_admissible}.
\end{itemize}

\begin{rem} {~}\label{rem_KG}

	{\bf (1)} The quantum vector KG field, in addition to the requirements in Definition~\ref{local_field}, it is also supposed to verify \eqref{proca} in a {\em distributional sense} for every background geometry
	\begin{equation}
	\label{qPe}
	A_{(M,\mathbf{b})}\left((-\square_{\mathbf{g}}+m^2+\xi R)f \right)=0 \quad \forall f \in \mathscr{D}(TM)\:.
	\end{equation}
	Though this fact does not play any role in our work, it implies several relevant facts which are mentioned in the some of subsequent remarks. Moreover, exactly as does the Klein-Gordon equation for the scalar field, this equation of motion plays a crucial role in the construction of an explicit algebra of Wick polynomials~\cite{HW1}.

	{\bf (2)} It is well-known~\cite[Sec.3.3.1]{AAQFT15Ch3} that the KG operator $P=-\square +m^2 + \xi R$ is {\em Green hyperbolic} for every choices of the involved given smooth functions ($m^2$ may attain non-positive values in particular)
	and thus the retarded and advanced Green operators of $P$ exist. In particular the function $\Delta_{(M,\mathbf{b})}$ discussed in~(4) Remark~\ref{rem_commutator} in this case is the {\bf causal propagator} of the KG equation~\cite[Sec.3.3.1]{AAQFT15Ch3}. 
	As a consequence of the standard properties of the causal propagator, we also have that $[A_{(M,\mathbf{b})}(f),A_{(M,\mathbf{b})}(g)]=0$ when the supports of $f$ and $g$ are causally disjoint.

	{\bf (3)} As is well-known, exactly as for the scalar field (\emph{e.g.},\ see~\cite[Sec.3.3.1]{AAQFT15Ch3}), the statement of the time-slice axiom for the locally covariant vector field $A$ can be sharpened, based on the properties of the causal propagator of equation of motion~\eqref{qPe}. Namely, if $O$ is an open neighborhood of any Cauchy surface of $(M,\mathbf{g})$ and $f \in \mathscr{D}(TM)$, then
	$A_{(M,\mathbf{b})}(f) = A_{(M,\mathbf{b})}(h)$ for a suitable $h \in \mathscr{D}(TM)$, depending on $f$, whose support is contained in $O$.

	{\bf (4)} When defining the Wick products $A^k_{(M,\mathbf{b})}(f)$, the class of states $\mathcal{S}_{(M,\mathbf{b})}$ appearing in the smoothness requirement in Definition~\ref{def_wick} should be naturally interpreted as consisting of the extensions of \emph{Hadamard states}~\cite{sahlmann01} from the unital $*$-subalgebra $\mathcal{W}_{0,(M,\mathbf{b})} \subset \mathcal{W}_{(M,\mathbf{b})}$ to the whole ambient algebra, where $\mathcal{W}_{0(M,\mathbf{b})}$ is generated by $1$ and products of elements $A_{(M,\mathbf{b})}(f)$.

	{\bf (5)} It is worth also stressing that the case $m=0$, even if the spacetime is Minkowski one, does {\em not} correspond to the quantization of the electromagnetic field (within Lorenz-gauge choice). 
	Indeed, we are dealing here with the algebraic approach and, in a given spacetime, the (Weyl) $*$-algebra of vector KG field is well defined for every choice of the function $m^2$ which may also attain negative values, because its definition only relies on the
	fact that the spacetime is globally hyperbolic and on the
	nature of the operator $P=-\square +m^2 + \xi R$ which is Green hyperbolic. The existence of Hadamard states playing a role in {\em requirement 5} can be proved with a standard deformation argument even in Minkowski spacetime for $m^2\leq 0$ constantly: it is enough to smoothly change the function $m^2$ in the past of  a Cauchy surface $\Sigma$ until it becomes a constant function with value $m_0^2>0$ in the past  a second Cauchy surface $\Sigma'$ in the past of $\Sigma$ in Minkowski spacetime. Next, in the past of $\Sigma'$ one may construct the standard Poincar\'e-invariant vacuum for (constant) squared mass $m_0^2>0$ and spin-$1$ particles. This state can be viewed as a state over the algebra in the future of $\Sigma$ when taking advantage of time slice axiom and it remains Hadamard in view of the known singularity propagation property of Hadamard states. Obviously, for the algebra of fields in the future of $\Sigma$, the constructed state is {\em not} the Poincar\'e-invariant vacuum which cannot be defined if $m^2<0$ (constantly) and the problem with negative-norm states would immediately arise for $m^2=0$ (usually removed by means of the Gupta-Bleuler treatment which also lower to $2$ the physical degrees of freedom of particles associated to the field from the $3$ degrees of freedom of massive spin-$1$ particles).
	This way also the $m^2\leq 0$ theory in Minkowski spacetime admits Hadamard states, 
	but none of them is a Poincar\'e-invariant vacuum. In other words, for $m^2=0$, our vector KG field field does {\em not} describe photons. In the algebraic approach, photons are described by including gauge invariance into the algebra of fields from scratch which is a more complicated procedure than the one we are discussing~\cite{hollands-ym,FR-BV}. Using some delicate adiabatic changes of mass procedures similar to the ones pointed out above it is however possible, at least for the scalar field, to transform vacua states into vacua states with different masses~\cite{DHP17,DD16,DG17}.

{\bf (6)} It is also worth commenting on the existence of prescriptions
of Wick polynomials that are smooth in $m^2$, including at $m^2=0$. For
that, it is important to recall the precise form of the smoothness axiom
(Definition~\ref{def_wick}, Axiom 5) and that the main candidate for
such a construction is point splitting regularized with a Hadamard
parametrix. That is, in the simplest $k=2$ case, what we must check is
the joint smoothness of the integral kernel $\omega_{ab}(s,x)$ in the
expression
\begin{equation} \label{A2-m-smooth}
	\lim_{y\to x} \omega \circ \tau_s^{-1}
		\left( A_{a\,(M,\mathbf{b}_s)}(x) A_{b\,(M,\mathbf{b}_s)}(y)
			- H_{ab\,(M,\mathbf{b}_s)}(x,y) 1 \right)
	= \omega_{ab}(s,x) ,
\end{equation}
where $H_{ab\,(M,\mathbf{b}_s)}(x,y)$ is the Hadamard parametrix and
$\omega$ is any Hadamard state on the algebra
$\mathcal{W}(M,\mathbf{b}_0)$, with $\mathbf{b}_0 = (\mathbf{g}_0,
m^2=m_0^2, \xi=\xi_0)$ and $\mathbf{b}_s$ a compactly supported
variation thereof. It is well-known that, already on (even dimensional)
Minkowski space with $m^2 = m_0^2$ constant, the Hadamard parametrix
$H_{ab}(x,y)$
contains terms proportional to $v_{m^2}(x,y) \log(\mu^2 \sigma(x,y))$, where
$\sigma(x,y)$ is the squared geodesic distance, $\mu^2$ is an
arbitrary dimensionful constant, and the dependence of $v_{m^2}$ on
$m^2$ is bilocal and smooth. On the other hand, the Wightman 2-point
function $\omega_{m^2} (A_{a}(x) A_{b}(y))$, where $\omega_{m^2}$ is the
Fock vacuum, also contains terms proportional to $\log(m^2\sigma(x,y))$.
Thus, we expect that the point split regularization
\begin{equation} \label{A2-m0}
	\lim_{y\to x} \omega_{m_0^2}(A_a(x) A_b(y) - H_{ab}(x,y) 1)
\end{equation}
gives rise to a smooth function of $x$ for $m=m_0$ fixed, because of the
cancellation of singular $\sigma(x,y)$-dependent terms. However, we also
expect that the arising result contain terms proportional to $\log
m_0^2/\mu^2$. Thus, at first glance, it might seem that the desired
smoothness property in~\eqref{A2-m-smooth} would not hold because of a
logarithmic singularity encountered as $m^2$ varies from $m_0^2$ to $0$
as a function of $s$. This is not the case because a careful comparison
of~\eqref{A2-m-smooth} and~\eqref{A2-m0} reveals that they are not
analogous expressions. In fact, one can never represent the family
$\omega_{m^2}$ of Fock vacua as $\omega_{m_0^2} \circ \tau_s^{-1}$ for
some fixed constant $m_0^2$ and an $s$-dependent compactly supported
variation thereof, because the difference $m^2-m_0^2$ would not be
compactly supported. To clarify some further information about
$m^2$-regularity at $m^2=0$, it is useful to observe  that, with
$\omega$ fixed in~\eqref{A2-m-smooth}, the difference between
$A_{a\,(M,\mathbf{b}_0)}(x) A_{b\,(M,\mathbf{b}_0)}(y)$ and
$\tau_s^{-1}(A_{a\,(M,\mathbf{b}_s)}(x) A_{b\,(M,\mathbf{b}_s)}(y))$ can
be expressed using advanced and retarded propagators for the vector KG
operators respectively on $(M,\mathbf{b}_0)$ and $(M,\mathbf{b}_s)$. We
expect and conjecture that the retarded propagator smoothly depend on
the difference $m^2(s,x) - m_0^2$, establishing the wanted smoothness
property. Though we do not have a demonstration of that, we outline a
possible way to construct a proof in the next paragraph. (The above
conclusions are implicit in the discussion of Section~5.2
of~\cite{HW2}.)

To argue that the retarded propagator with mass $m^2(s,x)$ has smooth
dependence on $m^2(s,x) - m_0^2$, when the difference has compact
support, we will refer to some results from~\cite{DHP17}. 

More precisely, we can express the retarded propagator $\Delta^R_{m^2}$
in terms of the retarded propagator $\Delta^R_{m^2_0}$ and an operator
$R_{m^2_0}$ (Lemma~3.10 in~\cite{DHP17}), where $R_{m^2_0} =
[1+\Delta^R_{m^2_0} (m^2-m^2_0)]^{-1}$ (Proposition~3.8
in~\cite{DHP17}). This comes down to the perturbative expression, cf.\
Equation~(43) in~\cite{DHP17},
\[
	\Delta^R_{m^2}
	= \sum_{n\geq 0} \left[-\Delta^R_{m^2_0} \left(m^2-m^2_0\right)\right]^n \Delta^R_{m^2_0} \:.
\]

Lemma~B.1 of~\cite{DHP17} uses the support properties of
$\Delta^R_{m^2_0}$ to show that the above series, together with all of
its functional derivatives with respect to the difference $m^2 - m_0^2$,
converges when $m_0^2=0$ and the background spacetime is Minkowski.
Though we do not have a proof and the issue should be investigated
elsewhere, it seems plausible that  the same proof generalizes to more
general globally hyperbolic spacetimes.

\end{rem}
\bigskip

With the concrete case of the vector KG field, Theorem~\ref{lemma_ipotesi} can be sharpened to give a more explicit expression for the renormalization coefficients $C_k$.
In terms of algebra valued distributions, equation \eqref{formula_generic} can be rewritten as
\begin{equation}
\widetilde{A_{b_1} \cdots A_{b_k}}(x)=A_{b_1} \cdots A_{b_k} (x)+\sum_{l=0}^{k-1}\binom{k}{l}C_{k-l}[M,\mathbf{b}]_{(b_1\cdots b_{k-l}}(x)A_{b_1} \cdots A_{b_l)} (x)
\end{equation}
with $C_k[M,\mathbf{b}]_{b_1\cdots b_{k}}(x) \in (T^*_x)^{\otimes k}M$ fully symmetric. Using Theorem~\ref{thm_uniqueness} we can immediately obtain a precise form of the symmetric covariant $k$-tensor fields $C_k[M,\mathbf{b}]$. For example, if we choose $n=4$ and $k=2$ we obtain for all $f\in\mathscr{D}(S^2TM)$
\begin{multline*}
\widetilde{A}^2_{(M,\mathbf{b})}(f)=A^2_{(M,\mathbf{b})}(f) \\
	+ 1\left(\left(y_1m^2\mathbf{g}+y_2R\mathbf{g}+y_3\mathbf{Ric}+y_4\square\xi \mathbf{g}+y_5\nabla^2\xi+y_6\mathbf{g}(\nabla\xi)^2+y_7(\nabla\xi)^{\odot 2}\right)\cdot_{2}f\right)
\end{multline*}
which can be written in terms of distributional fields, omitting explicit $x$-dependence for simplicity, as
\begin{multline*}
\widetilde{A_{a}A_{b}} =
A_{a} A_{b}+\left(y_1 g_{ab}m^2+y_2g_{ab}R +y_3R_{ab}
\right. \\ \left. {}
+y_4g_{ab}\square\xi +y_5\nabla_{(a}\nabla_{b)}\xi+y_6g_{ab}\nabla^c\xi\nabla_c\xi+y_7\nabla_{(a}\xi\nabla_{b)}\xi\right)\:,
\end{multline*}
where $y_j(x) := Y_j(\xi(x))$ and $Y_j$,  for $j=1,\ldots, 5$, are dimensionless smooth functions  which do not depend on the chosen spacetime. 
Obviously, in concrete physical theories the final values of some background fields like $m^2$ and $\xi$ are taken to be everywhere constant. In this case all derivatives of these fields disappear. In particular
\begin{equation*}
\widetilde{A_{a}A_{b}} =
A_{a} A_{b}+\left(y_1g_{ab}m^2+y_2g_{ab}R  +y_3R_{ab} \right)\:,
\end{equation*} 
where the $y_j := Y_j(\xi)$ turn out to be true \emph{renormalization constants} independent form the chosen spacetime. 

 \subsubsection{Vector Klein-Gordon field with tensor curvature coupling}
 \label{sec_vector_KG_xi_tens}
 It is possible to complicate a bit the previous example by adding a non-trivial background field. We consider a tensorial coupling to the scalar curvature in the vector KG equation, \emph{i.e.},
 \begin{equation} \label{vector_KG_xi_tens}
 -\square_{\mathbf{g}} A_a+m^2A_b+R\xi_a^bA_b=0.
 \end{equation}
 Lowering the upper index of the coupling tensor, $\xi_{ab}=g_{ac}\xi_b^c$, we have a fully covariant background $2$-tensor field. We will take $\xi_{ab}$ to be symmetric, both for simplicity and because only symmetric tensorial coefficients are compatible with the existence of a Lagrangian density for~\eqref{vector_KG_xi_tens}. Then, the bundle of background field is now completely defined by
 \begin{equation}
 BM= \mathring{S}^2T^*M\oplus \mathbb R\oplus S^2T^{*}M
 \end{equation}
 and the sections $M\rightarrow BM$ are triples $\mathbf{b}=(\mathbf{g},m^2,\boldsymbol{\xi})$. The background field $\boldsymbol{\xi}$ is marginal%
 	\footnote{In the Lagrangian density, the curvature coupling term becomes $Rg^{ad}g^{bc}\xi_{ac}A_bA_d\sqrt{-g}$.} %
 since the tensor index $l_M=2$ and the physical scaling weight $s_M=-2$, hence satisfying $l_M+s_M=0$. Clearly, $\xi_{ab}$ is the only marginal background field. All other hypotheses remain invariant with respect to the previous example.

 To apply our main Theorem~\ref{thm_uniqueness}, we first need to analyze the structure of the scalar polynomial invariants on the fibers of $S^2T^*M$ under the action of $O(1,n-1)$ (or $SO(1,n-1)$) and the separability of closed orbits by these invariants. As is well known~\cite[Sec.11.8]{procesi}, a generating set of the polynomial invariants is given by the contractions
\begin{equation} \label{xi_invar_scal}
\left(\tr \boldsymbol{\xi} = \xi_a^{a}, \tr\boldsymbol{\xi}^2 = \xi_a^b\xi_b^a, \ldots, \tr\boldsymbol{\xi}^n = \xi_{a_1}^{a_2}\xi_{a_2}^{a_3} \cdots \xi_{a_n}^{a_1}\right) \:,
\end{equation}
 which, as indicated, can be interpreted as traces of successive powers of $\xi_a^b$, interpreted as $n$-dimensional endomorphisms (or $n\times n$ matrices). All higher order contractions are algebraically dependent due to the Cayley-Hamilton identity.  
 The result obtained in Theorem~\ref{thm_uniqueness} applied to this case when, for example, we choose $n=4$ and $k=2$ gives, omitting the $x$-dependence for simplicity,
 \begin{equation*}
 \widetilde{A_{a}A_{b}} =
 A_{a} A_{b}+\left(y_1g_{ab}m^2+y_2g_{ab}R +y_3 R_{ab}+y_4 m^2\xi_{ab}+y_5\xi_{ab}R +B_\xi \right) \:,
 \end{equation*}
 with all terms that vanish when the background fields are constant collected in
 \begin{flalign*}
 B_\xi={}&y_6g_{ab}\square\xi_c^c+y_7\nabla_{(a}\nabla_{b)}\xi_c^c+y_8g_{ab}\nabla^c\xi_d^d\nabla_c\xi_d^d+y_9g_{cd}\nabla_{(a}\xi^{cd}\nabla_{b)}\xi_c^c\\
 &+y_{10} \left(\nabla_{(a}\nabla_{b)}\xi_{cd}\right)\xi^{cd}+y_{11}\nabla_{(a}\xi^{cd}\nabla_{b)}\xi_{cd}+y_{12}g_{ab}\left(\square\xi_{cd}\right)\xi^{cd}+y_{13}g_{ab}\nabla^c\xi_{de}\nabla_c\xi^{de}\\
 &+y_{14}\xi_{ab}\square\xi_c^c+y_{15}\xi_{ab}\nabla^c\xi_d^d\nabla_c\xi_d^d+y_{16} \square\xi_{ab}+y_{17}\xi_{ab}\left(\square\xi_{cd}\right)\xi^{cd}+y_{18}\xi_{ab}\nabla^c\xi_{de}\nabla_c\xi^{de}\\
 &+y_{19}\xi_{cd}\nabla_{(a}\xi^{cd}\nabla_{b)}\xi_c^c+y_{20}\xi_{cd}\xi_{ef}\nabla_{(a}\xi^{ef}\nabla_{b)}\xi^{cd}
 \:,
 \end{flalign*}
 where the $y_i$ are locally smooth functions our invariant
 scalars~\eqref{xi_invar_scal} in the sense of
 Definition~\ref{def_loc_poly} and Proposition~\ref{prop_luna_ext}.

 Now, we analyse in detail the structure of the coefficients $y_i$. In
 general, illustrating the phenomenon discussed in
 Appendix~\ref{section_equivar}, our invariant polynomials do not
 separate the closed orbits of  $O(1,n-1)$ (or $SO(1,n-1)$) acting on
 the fibers of $S^2T^*M$. For instance, given an orthonormal basis
 $v^0_a,\ldots,v^3_a$ with $\mathbf{v}^0$ timelike and the rest
 spacelike, the following symmetric tensors with distinct
 $\lambda_0,\ldots,\lambda_4$ cannot be distinguished by invariant
 polynomials
 \begin{equation*}
 	\boldsymbol{\xi} =
		-\lambda_0 v^0_a v^0_b
		+\lambda_1 v^1_a v^1_b
		+\lambda_2 v^2_a v^2_b
		+\lambda_3 v^3_a v^3_b
		\quad \text{and} \quad
		\boldsymbol{\xi}' =
		-\lambda_1 v^0_a v^0_b
		+\lambda_0 v^1_a v^1_b
		+\lambda_2 v^2_a v^2_b
		+\lambda_3 v^3_a v^3_b \:,
 \end{equation*}
 even though they belong to different orbits. The orbits are distinct
 because any linear transformation mapping $\boldsymbol{\xi}$ to
 $\boldsymbol{\xi}'$ must exchange the $\lambda_0$- and
 $\lambda_1$-eigenvectors, hence exchanging a spacelike vector with a
 timelike vector, which cannot be done by any element of $O(1,n-1)$.
 Other examples of this kind can be constructed by looking at the
 complete classification of the orbit types of symmetric
 2-tensors~\cite[Sec.5.1]{stephani-sols}. On the other hand, invariant
 polynomials do distinguish the orbit of $\boldsymbol{\xi}$ from the
 orbit of any other point in a sufficiently small neighborhood, because
 the case of distinct eigenvalues allows us to choose the eigenvectors
 smoothly under small variations, and small variations of timelike
 (spacelike) vectors remain timelike (spacelike). Thus, the subsets
 where invariant polynomials can locally distinguish orbits must be
 separated by a ``barrier'' (the $Z^0$ subset of
 Proposition~\ref{prop_luna_ext}). Since any continuous path from
 $\boldsymbol{\xi}$ to $\boldsymbol{\xi}'$ must pass through some tensor
 with \emph{degenerate} eigenvalues, we can take $Z^0$ to consist of all
 tensors with at least two equal eigenvalues. The open sets $Z_j$ of
 Proposition~\ref{prop_luna_ext} can then be identified with the
 connected components of $Z \setminus Z_0$, where $Z$ is a generic fiber
 of $S^2T^*M$.

 The reason why the set $Z^0 \subset Z$ and the partition $Z \setminus
 Z^0 = \bigcup_j Z_j$ is consistent with Proposition~\ref{prop_luna_ext}
 is that $Z^0$ is actually the zero-set of an invariant polynomial
 $p_0(\boldsymbol{\xi}) = \mathrm{disc}(\boldsymbol{\xi})$, known as the
 \textbf{matrix discriminant}. It is defined by requiring that, for
 diagonalizable tensors with eigenvalues $\lambda_i$, it takes the value
 \begin{equation*}
\mathrm{disc}\left(\boldsymbol{\xi}\right)=\prod_{i<j}\left(\lambda_i-\lambda_j\right)^2 \:,
 \end{equation*}
 which can be shown to coincide with the polynomial
 $\mathrm{disc}\left(\boldsymbol{\xi}\right)=\det\left(\tr\boldsymbol{\xi}^{i+j-
 2}\right)_{i,j=1}^n$~\cite[Lem.1]{PARLETT02}. In the $n=4$ case, it has
 the explicit form
 \begin{equation*}
 \mathrm{disc}\left(\boldsymbol{\xi}\right)=
 \det\begin{pmatrix}
 \tr I & \tr \boldsymbol{\xi} &  \tr \boldsymbol{\xi}^{2} &\tr \boldsymbol{\xi}^3  \\ 
 \tr \boldsymbol{\xi} & \tr \boldsymbol{\xi}^2 &  \tr \boldsymbol{\xi}^3 & \tr \boldsymbol{\xi}^4 \\ 
 \tr \boldsymbol{\xi}^{2}& \tr \boldsymbol{\xi}^3 & \tr \boldsymbol{\xi}^{4} & \tr \boldsymbol{\xi}^5  \\ 
 \tr \boldsymbol{\xi}^3 & \tr \boldsymbol{\xi}^4 & \tr \boldsymbol{\xi}^5 & \tr \boldsymbol{\xi}^6
 \end{pmatrix},
 \end{equation*}
 where we recall that $\tr \boldsymbol{\xi}^5,\tr \boldsymbol{\xi}^6$
 are algebraically dependent on lower order contractions due to the
 Cayley-Hamilton identity.  Thus, the coefficients $y_i$ are  locally
 smooth functions (Definition~\ref{def_loc_poly}) of the scalar
 polynomials invariants~\eqref{xi_invar_scal}, \emph{i.e.},
 \begin{equation*}
 y_i(x) = [Y_i]_{S^2T^*M}(\tr\boldsymbol{\xi}(x), \tr\boldsymbol{\xi}^2(x), \tr\boldsymbol{\xi}^3(x), \tr\boldsymbol{\xi}^4(x)) \:,
 \end{equation*}
 for $i=1,\ldots, 20$,  with respect to the partition $Z\setminus Z^0 =
 \bigcup_j Z_j$ indicated above, with $Z$ a generic fiber of $S^2T^*M$.

\subsection{Scalar field with derivative} \label{sec_grad_scalar}
We now consider the renormalization of Wick powers of a \textbf{scalar field with its first derivative}. 
The classical configurations of the scalar KG field over an oriented globally hyperbolic spacetime $(M, \mathbf{g})$ are smooth  real-valued functions, \emph{i.e.},\ sections of the bundle $T^{*\otimes 0}M=M\times\mathbb{R}$, namely
$\varphi \in \mathscr{E}(T^{*\otimes 0}M)=C^\infty (M)$. Similarly to the previous case we have the following equation of motion
\begin{equation}
\label{scalarKG}
-\square_{\mathbf{g}} \varphi+m^2\varphi+\xi R\varphi=0,
\end{equation}
where $m^2$ and $\xi$ are smooth real-valued functions on $M$ (they can be constant functions, but in general we admit that $m^2$ and $\xi$ can vary on the spacetime). 
Since we want to consider renormalization of a scalar field with its first derivative, we construct the field $\Phi$ as the pair of fields
\begin{equation*}
\Phi=\left(\varphi,\nabla_a\varphi\right).
\end{equation*}

When passing to the quantum formulation, the \emph{locally-covariant quantum field $\Phi$}, indicated by the same symbol $\Phi$, is defined as in Definition~\ref{local_field}, with the following details.
\begin{itemize}
	\item[(a)] As in the previous case, the net of local quantum observables $\mathcal{W}$ including the scalar field, as in Definition~\ref{def_netalgebra}, is fixed
	according to  equation \eqref{scalarKG}, which suggests that the natural bundle of background fields is the one completely defined by
	\begin{equation}\label{BMKGgrad}
	BM= \mathring{S}^2T^*M\oplus (M\times{\mathbb R}) \oplus (M\times\mathbb R),
	\end{equation}
	so that the 
	sections $M\rightarrow BM$ are triples $\mathbf{b}=(\mathbf{g},m^2,\xi)$.

	\item[(b)] The natural vector bundle is completely fixed by requiring 
	\begin{equation*}
	VM=(M\times \mathbb{R})\oplus TM
	\end{equation*} 
	and the morphism $V_{\chi}$, whose associated  pushforward on test sections is exploited to define the notion of local covariance as in Definition~\ref{local_field}, is nothing but the natural lift of the embeddings $\chi \colon M \to M'$ to the corresponding tangent bundles.

	\item[(c)] According to its mass dimension\footnote{E.g., assuming that both the terms summed in the Lagrangian density of the scalar field
		$m^2 \varphi^2 \sqrt{g}$ and $\nabla^a\varphi\nabla_a\varphi \sqrt{g}$ are dimensionless in natural $\hbar=c=1$ units.}, the physical scaling degree of the field $\Phi$ is  
		\begin{equation*}
		\mathbf{d}_\Phi = \left(\frac{n-2}{2},\frac{n-2}{2}\right),
		\end{equation*}
	when $\mathbf{g}\mapsto \lambda^{-2} \mathbf{g}$, $m^2\mapsto \lambda^2 m^2$ and $\xi\mapsto \xi$ according to \eqref{ps}. We recall that the presence of covariant derivatives do not change this rescaling behaviour as the coordinates are dimensionless.
	
	\item[(d)] We stress that all background fields of this model are scalars of non-negative physical scaling weight and hence are admissible according to Definition~\ref{def_admissible}.
\end{itemize}
For this specific model, using Theorem~\ref{lemma_ipotesi} and Theorem~\ref{thm_uniqueness}, we can immediately obtain a renormalization formula and a precise form of the renormalization counterterms. For example, if we choose $n=4$ and $k=2$ we obtain, in terms algebra valued distributions and for brevity omitting the all dependence on the spacetime point $x$,
\begin{equation*}
\begin{bmatrix}
\widetilde{\varphi}^2 \\[5pt]
\widetilde{\varphi\nabla_a\varphi} \\[5pt]
\widetilde{\nabla_{(a}\varphi\nabla_{b)}\varphi}
\end{bmatrix} =
\begin{bmatrix}
\varphi^2 \\[5pt]
\varphi\nabla_a\varphi\\[5pt]
\nabla_{(a}\varphi\nabla_{b)}\varphi
\end{bmatrix} +
\begin{bmatrix}
\alpha_1m^2+\alpha_2R + A_{\xi,m^2} \\[5pt]
\beta_1\nabla_a R+B_{\xi,m^2}\\[5pt]
g_{ab}\left(\gamma_1 (m^2)^2 + \gamma_2 m^2 R + \gamma_3 R^2\right)
+ \left(\gamma_4 m^2 + \gamma_5 \square\right) R_{ab}
+ C_{\xi,m^2}
\end{bmatrix}
\end{equation*}
where all $\alpha$-, $\beta$-, and $\gamma$-coefficients are smooth functions of $\xi$ and
\begin{align*}
A_{\xi,m^2}
	&= \alpha_3\nabla^a\xi\nabla_a\xi
		+\alpha_4\square\xi \:, \\
B_{\xi,m^2}
	&= \beta_2 \nabla_a m^2
		+\beta_3 m^2\nabla_a \xi
		+\beta_4 R \nabla_a \xi
		+\beta_5 R_{ab} \nabla^b \xi
 \\
	& \quad {}
		+\beta_6 (\nabla^b \xi \nabla_b \xi) \nabla_a \xi
		+\beta_7 \square\xi \nabla_a \xi
		+\beta_8 \nabla^b \xi \nabla_{(b} \nabla_{a)} \xi
		+\beta_9 \nabla_a \square \xi
		\:, \\
C_{\xi,m^2}
	&= \gamma_6 \nabla_{(a} \xi \nabla_{b)} m^2
		+\gamma_7 m^2 \nabla_{(a} \xi \nabla_{b)} \xi
		+\gamma_8 R \nabla_{a} \xi \nabla_{b} \xi
		+\gamma_9 R_{ab} \nabla^{c} \xi \nabla_{c} \xi
		+\gamma_{10} R_{c(a} \nabla_{b)} \xi \nabla^{c} \xi
		\\
	& \quad {}
		+\gamma_{11} g_{ab} \nabla^c \xi \nabla_c m^2
		+\gamma_{12}  g_{ab} m^2 \nabla^c \xi \nabla_c \xi
		+\gamma_{13} g_{ab} R \nabla^{c} \xi \nabla_{c} \xi
		+\gamma_{14} g_{ab} R^{bc} \nabla_{b} \xi \nabla_{c} \xi
		\\
	& \quad {}
		+\gamma_{15} \nabla_{(a} \nabla_{b)} m^2
		+\gamma_{16} m^2 \nabla_{(a} \nabla_{b)} \xi
		+\gamma_{17} \square \xi \nabla_{(a} \nabla_{b)} \xi
		+\gamma_{18} R \nabla_{(a} \nabla_{b)} \xi
		+\gamma_{19} R_{ab} \square \xi
		\\
	& \quad {}
		+\gamma_{20} g_{ab} \square m^2
		+\gamma_{21} g_{ab} m^2 \square\xi
		+\gamma_{22} g_{ab} (\square \xi)^2
		+\gamma_{23} g_{ab} R \square \xi
		\\
	& \quad {}
		+\gamma_{24} \nabla_{(a} \xi \nabla_{b)} \square \xi
		+\gamma_{25} \nabla_{(a} \nabla_{b)} \square \xi
		\\
	& \quad {}
		+\gamma_{26} g_{ab} \nabla^c \xi \nabla_c \square \xi
		+\gamma_{27} g_{ab} \square^2 \xi
\end{align*}
are terms which depend on covariant derivatives of $\xi$ and $m^2$. Thus, if we choose constant values for $m^2$ and $\xi$:
\begin{equation*}
\begin{bmatrix}
\widetilde{\varphi}^2 \\[5pt]
\widetilde{\varphi\nabla_a\varphi} \\[5pt]
\widetilde{\nabla_{(a}\varphi\nabla_{b)}\varphi}
\end{bmatrix} =
\begin{bmatrix}
\varphi^2 \\[5pt]
\varphi\nabla_a\varphi\\[5pt]
\nabla_{(a}\varphi\nabla_{b)}\varphi
\end{bmatrix} +
\begin{bmatrix}
\alpha_1m^2+\alpha_2R\\[5pt]
\beta_1\nabla_aR\\[5pt]
g_{ab}\left(\gamma_1 (m^2)^2 + \gamma_2 m^2 R + \gamma_3 R^2\right)
+ \left(\gamma_4 m^2 + \gamma_5 \square\right) R_{ab}
\end{bmatrix} \:.
\end{equation*}
Also, if we wanted to maintain the Leibniz rule $\nabla_a \varphi^2 =
2\varphi \nabla_a \varphi$ (cf.~\cite{HW5}), we would have to require
$2\beta_1 = \alpha_2$, with the further requirements $2\beta_2 =
\alpha_1$, $2\beta_3 = \alpha_1'$, $2\beta_4 = \alpha_2'$, $2\beta_6 =
\alpha_3'$, $2\beta_7 = \alpha_4'$, $2\beta_8 = 2\alpha_3$ and $2\beta_9
= \alpha_4$ (where ${}'$ denotes $\frac{d}{d\xi}$) for other
coefficients in the case of non-constant $m^2$ and $\xi$.

\begin{rem}
Following the same ideas of this section, it is possible to renormalize a scalar field with derivatives of arbitrary order. If we construct the $(m+1)$-tuple 
\begin{equation*}
\left(\varphi,\nabla_{a_1}\varphi,\nabla_{a_1}\nabla_{a_2}\varphi,\ldots,\nabla_{a_1}\cdots\nabla_{a_m}\varphi\right),
\end{equation*}
\emph{i.e.},\ if we choose as bundle of dynamical fields
\begin{equation*}
VM=\bigoplus_{i=0}^mT^{*\otimes i}M,
\end{equation*}
we can use Theorem~\ref{lemma_ipotesi} and Theorem~\ref{thm_uniqueness} as we did in this section to obtain a renormalization formula with all renormalization counterterms. With the same idea it is possible to renormalize any tensor fields with an arbitrary number of derivatives.
\end{rem}

\section{Conclusions}
This paper has focused on the general notion of Wick powers for general boson fields within the 
formulation of locally covariant algebraic quantum field theory on globally hyperbolic curved spacetimes. For us, a general boson field is a section of an arbitrary natural vector bundle of the spacetime (where \emph{naturality} implies a well defined transformation law under diffeomorphisms). Besides the metric, the spacetime is also allowed to carry arbitrary classical background fields (also sections of natural vector bundles). In 
particular we have viewed the mass and other parameters as such background 
fields.

We define Wick powers axiomatically (Definition~\ref{def_wick}). Our list of axioms simply generalizes the axioms that were used for the scalar field in~\cite{KM16}, which in turn descend from those given in~\cite{HW1} (with the crucial difference that their ``analytic dependence'' axiom was replaced by our ``smooth dependence'' axiom). Our main results consist of a classification of all possible finite renormalizations of Wick powers, which refer to the ambiguities in their axiomatic definition. Our work provides the first rigorous and complete such classification for non-scalar fields. The are analogous to those given in~\cite{KM16}, but become more complicated in the details, due to the higher degree of generality.

The first half of our main result (Theorem~\ref{lemma_ipotesi}), by an
application of the Peetre-Slov\'ak theorem, reduces finite
renormalizations of a $k$-th Wick power to a linear combination of Wick
powers of lower order with coefficients $C_k$ that are differential
operators locally depending on the background fields, of fixed physical
scaling weight and transforming covariantly under diffeomorphisms. The
second half of our main result (Theorem~\ref{thm_uniqueness}) is
specialized to the case when both the dynamical and background fields
are restricted to be tensors (\emph{e.g.},\ the case of connection
fields is not covered), by an application of a general version of the
Thomas Replacement theorem (Appendix~\ref{section_coord_scal}) and some
fundamental results from smooth classical invariant theory of the
orthogonal group $O(1,n-1)$ or $SO(1,n-1)$ (Appendix~\ref{section_equivar}), reduces
the $C_k$ to linear combinations of finitely many tensor polynomials
covariantly constructed from the curvature, the background tensor field,
and all of their covariant derivatives. This finiteness result crucially
depends on an \emph{admissibility} criterion for all the background
fields (Definition~\ref{def_admissible}), which relates the physical
scaling weight of a background field with its tensor rank by an
inequality. The structure of these tensor polynomials is controlled by
their physical scaling weights. It is possible that for a given tensor
type and scaling weight the list of such polynomials is empty, meaning
that the corresponding component of $C_k$ vanishes. The strongest
departure from the results of~\cite{KM16} is in the structure of the
scalar coefficients in front of these polynomial terms. These
coefficients are actually allowed to depend smoothly (not just
polynomially) on the background fields, but in a very restricted way.
Namely, they are allowed to \emph{locally} be smooth functions only of a finite
number of scalar polynomial invariants constructed covariantly from the
subset of \emph{marginal} background fields (those that saturate the
admissibility inequality). The notion of local smooth dependence on
these scalar invariants (cf.~Definition~\ref{def_loc_poly} and
Proposition~\ref{prop_luna_ext}) can be made precise only by looking at
the structure of the orbits of the action of $O(1,n-1)$ or $SO(1,n-1)$
on the marginal background tensor fields. In the scalar Klein-Gordon
case considered~\cite{KM16}, the only marginal background field was the
scalar curvature coupling $\xi$.

We illustrate our results in detail with two physically relevant
examples, checking in particular that they satisfy all the admissibility
hypotheses: the vector Klein-Gordon field $A_a$
(Section~\ref{proca_section}), possibly coupled to the curvature through
a tensor background field $\xi_{ab}$
(Section~\ref{sec_vector_KG_xi_tens}), and the case of Klein Gordon
scalar field $\varphi$ accompanied by its spacetime derivative $\nabla_a
\varphi$ (Section~\ref{sec_grad_scalar}).

Several open issues remain and certainly deserve investigation. First of
all, a theorem of existence for Wick polynomials should be established.
This should be possible with existing tools, since the standard Hadamard
parametrix regularization method~\cite{HW2} should be suitable for
vector fields too, as discussed in Section~\ref{proca_section}. The main
problem is to check that our ``smooth dependence'' axiom is actually
satisfied by this method. We have already made more detailed comments on
this in Section~4 of~\cite{KM16}. As remarked at the end of
Section~\ref{section_peetre}, it might be practically easier to verify
the ``smooth dependence'' axiom when expressed in terms of
\emph{Bastiani differentiability}~\cite{BDLR17}, rather than our
\emph{weak regularity} (Definition~\ref{def_regular}).

Second,  the constructed formalism should be so enlarged, possibly
adding or changing some axioms, to cover the more delicate case of the
Proca field. Here the main problem is that the zero mass limit $m^2\to
0$ is known not to be smooth (see~\cite{2017arXiv170901911S} for a
careful recent discussion), whereas one of our axioms for Wick powers
requires regularity exactly at the zero value of the mass. Some related
remarks about subtleties with regular mass dependence appear in
Remark~\ref{rem_KG}.

Third, our results should be generalized to more general kinds
of bosonic fields (for instance non-tensorial fields like connections)
and also to fermionic fields (for instance Dirac spinor fields).
We believe that such extensions should be fairly straightforward by
building on the ground work that we have already laid. Such extensions
will be discussed in forthcoming work.

A different and much more difficult extension would regard the
renormalization of time ordered products of Wick powers, extending the
existing results~\cite{HW1,HW2,HW5}, which are again currently available
only in the scalar case. All these issues will be investigated
elsewhere.

\paragraph{Acknowledgments.}
The authors are grateful to Charles Torre for sharing with them the
unpublished report~\cite{at-report}, also to Klaus Fredenhagen and
Nicola Pinamonti for raising and clarifying some issues in
Remark~\ref{rem_KG}(6), and also to Jan Slov\'ak for discussions that
were helpful for Appendix~\ref{section_coord_scal}. IK was partially
supported by the ERC Advanced Grant 669240 QUEST ``Quantum Algebraic
Structures and Models'' at the University of Rome 2 (Tor Vergata). AM is grateful to the Math Dept. of University of Rome 2 (Tor Vergata) and of University of Milan for kind hospitality during the development of this work.

\appendix
\section{Technical result on physical scaling} \label{section_scaling}
In this Appendix we report some results from~\cite[Sec.2.4]{KM16} with some generalization in order to consider the case of a tensor valued function.
We recall that in Section \ref{section_tensor_field} we have defined the bundle of background fields as	
\begin{equation}
BM= \mathring{S}^2T^*M\oplus\left(\bigoplus_{j=1}^{K}  T^{*\otimes l_j}M\right)
\end{equation}
and that the physical scaling transformation on the sections of $\Gamma(BM)$ is given by
\begin{equation*}
BM\ni\left(p,\mathbf{g}(p),\mathbf{t}_j(p)\right)\longmapsto \left(p,\lambda^{-2}\mathbf{g}(p),\lambda^{s_j}\mathbf{t}_j(p)\right)\in BM,
\end{equation*}
where $\lambda\in\mathbb{R}^+$ defines the scaling transformation. 

This (globally defined) representation of the multiplicative group $\mathbb{R}^+$ can be written in local coordinates
\begin{equation*}
x^a\mapsto x^a,\quad g_{ab}\mapsto \lambda^{-2}g_{ab},\quad (t_j)_{a_1\ldots a_{l_j}}\mapsto \lambda^{s_j} (t_j)_{a_1\ldots a_{l_j}}.
\end{equation*}
This transformation lifts to a transformation of the jet bundle $J^rBM$. In local coordinates
\begin{equation*}
x^a\mapsto x^a,\quad g_{ab,A}\mapsto \lambda^{-2}g_{ab,A},\quad (t_j)_{a_1\ldots a_{l_j},A}\mapsto \lambda^{s_j} (t_j)_{a_1\ldots a_{l_j},A}.
\end{equation*}

With respect to the Definition~\ref{def_homogeneous}, we are interested in the case $W=C^\infty\left(J^rBM, S^kV^*M\right)$.
Moreover, since we have to consider also smaller domains $Z^r\subseteq J^r BM$ (with $Z^r$ not invariant under physical scaling), it is more convenient to consider the infinitesimal version of these transformations, which are effected by the following vector field\footnote{We use the following notation \begin{equation*}
\partial^{ab,A}:=\frac{\partial}{\partial g_{ab,A}}\quad \partial^{a_1\ldots a_{l_j},A}:=\frac{\partial}{\partial (t_j)_{a_1\ldots a_{l_j},A}}
	\end{equation*} and an analogous one for contravariant coordinates.}
\begin{equation*}
e=-2g_{ab,A}\partial^{ab,A}+s_j(t_j)_{a_1\ldots a_{l_j},A}\partial^{a_1\ldots a_{l_j},A},
\end{equation*}
in the sense that the induced action on tensor functions on $J^rBM$ satisfies
\begin{equation}
\label{lie_derivative}
\left.\frac{\de}{\de\lambda}\right|_{\lambda=1}F_\lambda=\mathcal{L}_eF,
\end{equation}
where $\mathcal{L}_e$ is the Lie derivative, $F\in W$ and $F_\lambda\in W$ is the transformed under physical scaling of $F$. We stress that, since the physical scaling transformation is globally defined, the vector field $e$ is globally defined on $J^rBM$.
\begin{lem}
	\label{lem_inf_scaling}
	A smooth function $F\colon J^rBM\rightarrow S^kV^*M$ that has almost homogeneous degree $k$ and order $l$ when the action $F\rightarrow F_\lambda$ is the one induced by physical scaling transformations, satisfies
	\begin{equation*}
	\left(\mathcal{L}_e-k\right)^{l+1}F=0.
	\end{equation*}
\end{lem}
\begin{proof}
If $F$ is an almost homogeneous function of degree $k$ and order $l$, using equation \eqref{lie_derivative}, we obtain
\begin{equation*}
	\left(\mathcal{L}_e-k\right)F=G^{(l-1)},
\end{equation*}
where $G^{(l-1)}$ is an almost homogeneous function of degree $k$ and order $l-1$. If we repeat this operation $l$ we obtain an homogeneous function $G^{(0)}$ of degree $k$:
\begin{equation*}
	\left(\mathcal{L}_e-k\right)^lF=G^{(0)}.
\end{equation*}
Since, for all homogeneous function $B$, we have  $\left(\mathcal{L}_e-k\right)B=0$ the proof is concluded.
\end{proof}
Thanks to this result, we can give an infinitesimal definition of homogeneous and almost homogeneous function. This definition is very useful since we have to consider function defined on a subset $Z^r\subseteq J^rBM$ which is not invariant under physical scaling.
\begin{defn}
	A smooth function $F\colon Z^r\subseteq J^rBM\rightarrow S^kV^*M$, where $Z^r$ is an open subset which may coincide with all of $J^rBM$, is said to have \textbf{almost homogeneous degree $k\in\mathbb{R}$ and order $l\in\mathbb{N}$} (with $l\geq 0$) \textbf{under physical scalings} if it satisfies the identity
	\begin{equation*}
	\left(\mathcal{L}_e-k\right)^{l+1}F=0.
	\end{equation*}
	If $l=0$, $F$ is said to have \textbf{homogeneous degree} $k\in\mathbb{R}$.
\end{defn}
In the contravariant coordinates $(x^a,g^{ab,A},(t_j)^{a_1\ldots a_{l_j},A})$, defined in Section~\ref{section_coordinates_bundle}, finite and infinitesimal physical scalings take the form
\begin{equation*}
x^a\mapsto x^a,\quad  g\mapsto \lambda^{-2n}g,\quad g^{ab,A}\mapsto \lambda^{2+2|A|}g^{ab,A},\quad (t_j)^{a_1\ldots a_{l_j},A}\mapsto\lambda^{s_j+2|A|}(t_j)^{a_1\ldots a_{l_j},A}
\end{equation*}
\begin{equation}
\label{phys_scaling_field}
e=\left(2+2|A|\right)g^{ab,A}\partial_{ab,A}+\left(s_j+2|A|\right)(t_j)^{a_1\ldots a_{l_j},A}\partial_{a_1\ldots a_{l_j},A}
\end{equation}
where, as remarked previously, we use $g$ as coordinate in place of one of the $g^{ab}$. Since $e$ is everywhere non zero its integral curves form a foliation of $J^rBM$ and hence of $Z^r$. Moreover, since $\mathcal{L}_eg^{-\frac{1}{2n}}=g^{-\frac{1}{2n}}$, $g$ restricts to a global coordinate on each orbit of $e$ and then the level sets of $g$ form another foliation of $J^rBM$, transverse to the integral curves of $e$. For this reason it is convenient to study the structure of almost homogeneous function in \emph{rescaled} coordinates:
  \begin{figure}
	\begin{center}
		\def\svgwidth{8cm}
		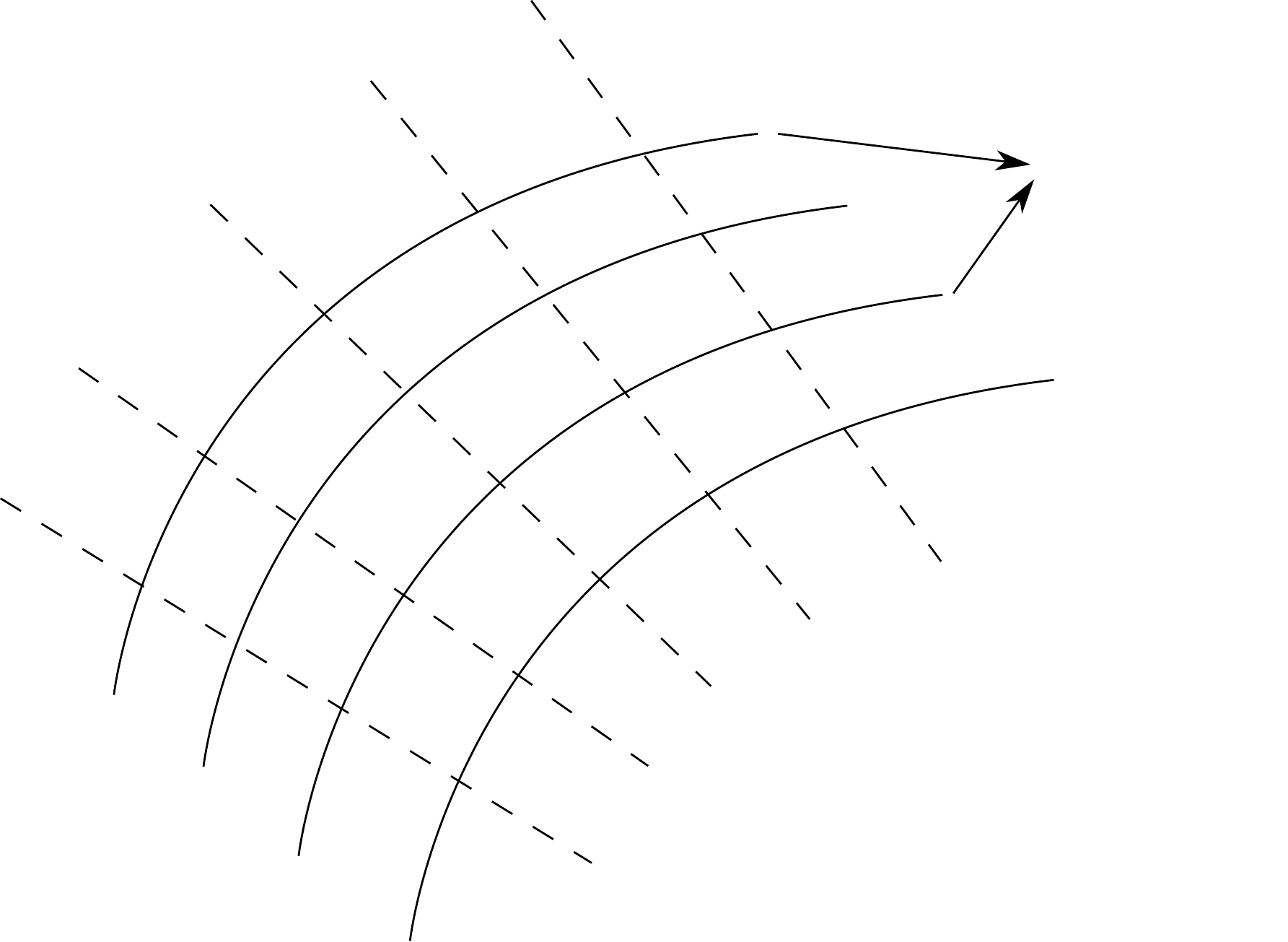
	\end{center}
\end{figure}
\begin{equation*}
\left(x^a,g,g^{-\frac{1}{n}}g_{ab},g^{\frac{1}{n}+\frac{1}{n}|A|}g^{ab,A},g^{\frac{l_j}{n}+\frac{s_j}{2n}+\frac{1}{n}|A|}(t_j)^{a_1\ldots a_{l_j},A}\right).
\end{equation*}
Note that each of these functions but $g$ is invariant under physical scaling. In our notation, we mean that the coordinates $g$ and $g^{-\frac{1}{n}} g_{ab}$ are functionally independent only up to the identity $g = \left| \det g_{ab} \right|$.
\begin{lem}
	\label{lem_functional_form_1}
	Suppose that $Z^r\subseteq J^rBM$ is an open set equipped with either coordinates $(x^a,g^{ab,A},(t_j)^{a_1\ldots a_{l_j},A})$ or some other coordinate system introduced in Section~\ref{section_coordinates_bundle}, and $F\colon Z^r\rightarrow S^kV^*M$ is a smooth function that has almost homogeneous degree $k$ and order $l$ with respect to physical scaling. Then there exist uniquely defined \emph{homogeneous of degree $0$} functions $B_j\colon Z^r\rightarrow S^kV^*M$, for $j=0,1,\ldots,l$, such that
	\begin{equation*}
	F=g^{-\frac{k}{2n}}\sum_{j=0}^l\log^j\left(g^{-\frac{1}{2n}}\right)B_j.
	\end{equation*}
	In particular, using rescaled contravariant coordinates, each $B_j$ can be taken independent of $g$ and written in the form
	\begin{equation*}
	B_j=B_j\left(x^a,g^{-\frac{1}{n}}g_{ab},g^{\frac{1}{n}+\frac{1}{n}|A|}g^{ab,A},g^{\frac{l_j}{n}+\frac{s_j}{2n}+\frac{1}{n}|A|}(t_j)^{a_1\ldots a_{l_j},A}\right).
	\end{equation*}
\end{lem}
\begin{proof}
	The proof is the same as~\cite[Lem.2.4]{KM16} (since it is based on the notion of Lie derivative).
\end{proof}

\section{Thomas replacement theorem}
\label{section_coord_scal} 

In this section we state and give a mostly self-contained proof of a
version of the Thomas Replacement Theorem~\ref{thomas_rep_thm}, which
basically states that any non-linear differential operator that depends
on a Lorentzian (or pseudo-Riemannian) metric and a finite number of any
kind of other tensor fields while itself transforming as a tensor field
under diffeomorphisms must be expressible as a function of the covariant
derivatives of the Riemann curvature and the other tensor arguments.
This is a rather old result, with versions of it going back to the work
of Thomas~\cite{thomas} and in some form even to earlier works of
Christoffel~\cite{christoffel}. However, it has since then taken on a
folk nature, making it difficult to find precise references that state
the result in a form most convenient for our applications, give a
complete proof, with modern notation and terminology, that is concise
and without an overabundance of formalism. If one omits at least some of
the above conditions, the result of Theorem~\ref{thomas_rep_thm} can be
found in~\cite[\textsection III.7]{schouten},
\cite[\textsection\textsection28.14,33.10]{kms},
and~ \cite[Thm.3]{slovak-inv}. Thus, this section aims to be of convenience
to the reader and to those who will need prove related but slightly
different results, which could be useful when tensors are replaced by
more general natural geometric objects (like connections or possibly
higher order jets) or even spinors. Such results could be useful in
investigating finite renormalizations of Wick polynomials of fields with
these more general transformation properties. Our attempt at providing
such a useful reference is not the first and similar material, motivated
by the heat kernel approach to the Index Theorem, can be found
in~\cite{gilkey73, abp73}. These references concentrate on differential
forms covariantly constructed from the metric, so the final results we
state here are somewhat more general.

Let $BM \to M$ be a natural bundle of the form
\begin{equation}
	BM = (S^2T^* \oplus T^{k_1}\otimes T^{*l_1} \oplus \cdots
		\oplus T^{k_N} \otimes T^{*l_N}) M ,
\end{equation}
where $T^k\otimes T^{*l}M$ is the bundle of $(k,l)$-tensors.
Consider the \emph{curvature coordinates} introduced in
Section~\ref{section_tensor_field} (which is a version of the
system~(25) from~\cite{KM16}) on $J^rBM$, which we consider with a
slight change of notation,
\begin{equation}\label{eq:jet-coords}
(x^a, g_{ab}, \Gamma^a_{bC}, Q_{ab,C},
t^{a_1\cdots a_{k_1}}_{b_1\cdots b_{l_1},C}, \ldots
t^{a_1\cdots a_{k_N}}_{b_1\cdots b_{l_N},C}) ,
\end{equation}
where the multi-indices $C=c_1\cdots c_{|C|}$ range through the sizes
$|C|=1,\ldots, r$. They have the symmetry properties $\Gamma^a_{bC} =
\Gamma^a_{(bC)}$, $Q_{ab,C} = Q_{(ab),C} = Q_{ab,(C)}$ and $Q_{a(b,C)} =
0$. The notational change is that we use the notation $Q_{ab,cdE}$
instead of $\bar{S}_{ab(cd,E)}$ that was introduced in
Section~\ref{section_tensor_field}. The reason for the change is that
the $Q$-notation is better adapted to some index manipulation of which
we will make use below. The above coordinates are defined by the
relations
\begin{align}
\Gamma^a_{bcD} \circ j^r g
&= \del_{(D} \Gamma^a_{bc)}[g] , \\
Q_{ab,cdE} \circ j^r g
&= \nabla_{(E} \bar{S}_{cd)ab}[g] , \\
t^{a_1\cdots a_{k_i}}_{b_1\cdots b_{k_i},C}
&= \nabla_{(C)} t^{a_1\cdots a_{k_i}}_{b_1\cdots b_{k_i}} ,
\end{align}
with $|D| \ge 0$ and $|E| \ge 0$, where
\begin{equation}
\Gamma^a_{bc}[g]
= \frac{1}{2} g^{ae} (\del_b g_{ce} + \del_c g_{be} - \del_e g_{bc})
\end{equation}
are the usual Christoffel symbols and we recall that $\bar{S}_{abcd}[g] =
\bar{R}_{a(c|b|d)}[g]$, which of satisfies
\begin{equation}
\bar{S}_{abcd}[g]
= \bar{S}_{cdab}[g]
= \bar{S}_{(ab)cd}[g]
= \bar{S}_{ab(cd)}[g]
\quad \text{and} \quad
\bar{S}_{a(bcd)} = 0 .
\end{equation}

All coordinates, other than $(x^a,\Gamma^a_{(bc,A)})$, correspond to
components of tensor densities (Definition~\ref{tensor_rep})
transforming under $GL(n)$, where $GL(n)$ is interpreted as the quotient
of $\mathrm{Diff}_x(M)$, the subgroup of diffeomorphisms fixing the
point $x\in M$, by the subgroup of diffeomorphisms with vanishing
Jacobian at $x$.

\begin{rem}
	In the context of the use of Young diagrams to describe irreducible
	representations of $GL(n)$~\cite{fulton}, we can say the following.
	Given a point $x\in M$, we can choose a section $\beta \colon M \to BM$
	such that $(g_{ab},Q_{ab,C}) \circ j^r\beta(x)$ take on arbitrary values
	consistent with the symmetry type of the tensors ($g_{ab}$ has covariant
	$(2)$ Young type, while $Q_{ab,C}$ has contravariant $(|C|,2)$ Young
	type written in the row-symmetric convention~\cite[p.193]{penrose}).
\end{rem}

Given the particular symmetrizations that we have applied in defining
the coordinates $\Gamma^a_{bC}$ and $Q_{ab,cdE}$, it is not immediately
obvious that the system~\eqref{eq:jet-coords} really is a local coordinate
system on $J^rBM$. This result is stated in
Lemma~\eqref{lem:curv-coord-struct}. Our main reference for this result
is~\cite{at-report},%
	\footnote{Although the technical report~\cite{at-report} is
	unpublished, its authors have kindly shared it with us.} %
some of whose results are also reported in~\cite{at}.
A version of the coordinates~\eqref{eq:jet-coords} was introduced in
Equation~(2.18) of~\cite{at-report}. The structural results presented below
can also be found in the more recent and detailed~\cite{jentsch}. There
are two non-trivial facts that need to be noted.
\begin{lem} \label{lem:curv-coord-struct}
	(a) The coordinates~\eqref{eq:jet-coords} actually constitute a complete
	coordinate system on $J^rBM$, as can be seen from the inversion formulas
	\begin{align}
	\del_c g_{ab}
	&= 2 g_{d(a} \Gamma^d_{b)c} , \\
	\del_C g_{ab}
	&= 2 g_{d(a}\Gamma^d_{b)C}
	- \frac{2(|C|-1)}{|C|+1} Q_{ab,C}
	+ (\text{l.o.t}_{|C|-1})_{abC}
	\quad (|C|\ge 2) .
	\end{align}
	
	(b) The total coordinate and covariant derivatives act as follows
	(superscripted symmetrizations are performed later):
	\begin{align}
	\del_c g_{ab}
	&= 2 g_{d(a} \Gamma^d_{b)c} , \\
	\del_b \Gamma^a_{C}
	&= \Gamma^a_{bC}
	+ \frac{2}{|C|+1} Q^{a}{}_{b,C}
	+ \text{l.o.t}_{|C|-1}(g,\Gamma,Q)^a_{bC}
	\quad (|C|\ge 2) , \\
	\notag
	\nabla_c Q_{ab,C}
	&= Q_{ab,cC}
	+ \frac{2}{|C|+2} Q_{c(a,b)C}
	+ \frac{|C|}{|C|+2} Q_{(c_1c_2,c_3\cdots c_{|C|})abc} \\
	& \quad{}
	+ \text{l.o.t}_{|C|-1}(g,Q)_{c,ab,C}
	\quad (|C|\ge 2) .
	\end{align}
	In each case, $\text{l.o.t}_r$ depends on $g_{ab}$ and its derivatives
	only up to order $r$, and only via given coordinates when indicated.
\end{lem}
\begin{proof}
	The inversion formulas follow from equations~(2.7) and~(2.17)
	of~\cite{at-report}.
	
	The total derivative formulas follow directly%
		\footnote{Unfortunately, the last term on the first line of~(2.22)
		in~\cite{at-report} has the wrong index structure. It can be
		corrected by re-deriving the result from the proof of Theorem~2.6
		in~\cite{at-report}.} %
	from the \emph{structure equations}~(2.19), (2.19) and~(2.22)
	of~\cite{at-report}. 
\end{proof}

\begin{lem} \label{lem:xi-taylor}
	For any point $x\in M$, there exists a vector field $\xi \in
	\mathfrak{X}(M)$ such that $\xi^a(x)$, $\nabla_a\xi^b(x)$ and $\del_C
	\xi^a(x)$, with $|C|\ge 2$, can be selected arbitrarily.
\end{lem}

\begin{thm} \label{thomas_rep_thm}
	Let $F \colon V_x^{\prime r} \sse J^rBM \to T^k\otimes T^{*l}M$ be a smooth bundle
	map that is defined on a $\mathrm{Diff}(M)$-invariant domain
	$V_x^{\prime r}$ and is $\mathrm{Diff}(M)$-equivariant, given by
	\begin{equation}
	F^{a_1\cdots a_k}_{b_1\cdots b_l}
	= F^{a_1\cdots a_k}_{b_1\cdots b_l}
	(x, g, \Gamma, Q, t)
	\end{equation}
	in adapted coordinates~\eqref{eq:jet-coords} on a chart $V^{r}_x
	\sse V^{\prime r}_x$. That is, given a diffeomorphism $\chi\colon M\to M$, we
	have $\chi^* \circ F = F \circ p^r\chi^*$, where on the left $\chi^*$ is
	the pullback along $\chi$ acting on the tensor bundle $T^k\otimes T^{*l}M$, while on
	the right $p^r\chi^*$ is the $r$-jet prolongation of the pullback along
	$\chi$ acting on the bundle $BM$ of background fields. Then, when
	restricted to a chart $V_x^r \sse V_x^{\prime r}$ covered by adapted
	coordinates~\eqref{eq:jet-coords}, $F$ must be expressible as
	\begin{equation}
	F^{a_1\cdots a_k}_{b_1\cdots b_l}
	= G^{a_1\cdots a_k}_{b_1\cdots b_l}
	(g, Q, t) ,
	\end{equation}
	where the function $G$ is equivariant with respect to the action of
	$GL(n)$ on its arguments and the action of $GL(n)$ on the fibers of
	$T^k\otimes T^{*l}M$.
\end{thm}

\begin{proof}
	The pullback $\chi^*F$ of a bundle map $F\colon J^rBM \to T^k\otimes T^{*l}M$ by a
	diffeomorphism $\chi\colon M \to M$ is defined by the identity $\chi^*
	\circ F = (\chi^*F) \circ p^r\chi^*$, which can be illustrated by the
	following commutative diagram:
	\begin{equation}
	\begin{tikzcd}
	T^k\otimes T^{*l}M \ar{d} \& T^k\otimes T^{*l}M \ar{l}{\chi^*} \ar{d}
	\&
	J^r BM \ar{d} \ar{l}{F} \ar{r}{p^r\chi^*} \& J^r BM \ar{d} \ar{r}{\chi^*F}
	\&
	T^k\otimes T^{*l}M \ar{d}
	\\
	M \ar{r}{\chi} \& M \ar{r}{=}
	\&
	M \& M \ar{l}{\chi} \ar{r}{=}
	\&
	M
	\end{tikzcd} ,
	\end{equation}
	where the first and last columns should be identified. The equivariance
	condition then just says that $\chi^*F = F$.
	
	Given a vector field $\xi \in \mathfrak{X}(M)$, we denote the
	corresponding $1$-parameter family of diffeomorphisms $s\mapsto
	\chi^\xi_s \in \mathrm{Diff}(M)$, meaning $\xi = \frac{\d}{\d s}
	\chi^\xi_s |_{s=0}$. The Lie derivative of a tensor field $t\colon M \to
	T^k\otimes T^{*l}M$ or a jet $\beta^r\colon J^rBM$ is defined in the usual way:
	\begin{align*}
	\Lie_\xi \beta^r
	&= \left. \frac{\d}{\d s} p^r (\chi^\xi_s)^* \circ \beta^r
	\circ \chi^\xi_s \right|_{s=0} \\
	\Lie_\xi t
	&= \left. \frac{\d}{\d s} (\chi^\xi_s)^* \circ t
	\circ \chi^\xi_s \right|_{s=0} .
	\end{align*}
	The Lie derivative with respect to a vector field $\xi$ of a section
	$\alpha \colon M \to AM$ of a diffeomorphism-natural bundle $AM\to M$
	defines a section that we denote $(\alpha, \Lie_\xi\alpha)\colon M\to
	VAM$, where $VAM\to AM \to M$ is the compound vertical tangent bundle of
	$AM\to M$. When $AM\to M$ is a vector or affine bundle, we can identify
	$VAM \cong (A\oplus A)M$ and consider the Lie derivative as a section
	$\Lie_\xi\alpha \colon M\to AM$.
	
	The infinitesimal version of the equivariance condition $(\chi^\xi_s)^*
	\circ F = F\circ p^r (\chi^\xi_s)^*$ for $F$ is then
	\begin{multline} \label{eq:inf-eqvar}
	\left. \frac{\d}{\d s} \left[
	(\chi^\xi_s)^* \circ F - F\circ p^r(\chi^\xi_s)^* \right]
	\circ \beta^r \circ \chi^\xi_s \right|_{s=0} \\
	= (F\circ \beta^r, \Lie_\xi (F\circ \beta^r))
	- VF\circ (\beta^r, \Lie_\xi\beta^r)
	= 0 ,
	\end{multline}
	for any section $\beta^r\colon M\to J^rBM$, and with $VF \colon VJ^rBM
	\to V(T^k\otimes T^{*l})$ is a bundle map given by the restriction of the tangent map
	of $F$, $TF \colon T(J^rBM) \to T(T^k\otimes T^{*l}M)$ to the vertical tangent
	bundle $VJ^rBM \subset T(J^rBM)$.
	
	For tensors, the Lie derivative has the following well-known form:
	\begin{align}
	(\Lie_\xi t)^{a_1\cdots a_k}_{b_1\cdots b_k}
	\notag
	&= \xi^c \del_c t^{a_1\cdots a_k}_{b_1\cdots b_k} \\
	\notag
	& \quad {}
	- (\del_{a'_1}\xi^{a_1}) t^{a'_1\cdots a_k}_{b_1\cdots b_l} - \cdots
	- (\del_{a'_k}\xi^{a_k}) t^{a_1\cdots a'_k}_{b_1\cdots b_l} \\
	& \quad {}
	+ (\del_{b_1}\xi^{b'_1}) t^{a_1\cdots a_k}_{b'_1\cdots b_l} + \cdots
	+ (\del_{b_l}\xi^{b'_l}) t^{a_1\cdots a_k}_{b_1\cdots b'_l}
	\label{eq:tens-rule-coord}
	\\
	\notag
	&= \xi^c \nabla_c t^{a_1\cdots a_k}_{b_1\cdots b_k} \\
	\notag
	& \quad {}
	- (\nabla_{a'_1}\xi^{a_1}) t^{a'_1\cdots a_k}_{b_1\cdots b_l} - \cdots
	- (\nabla_{a'_k}\xi^{a_k}) t^{a_1\cdots a'_k}_{b_1\cdots b_l} \\
	& \quad {}
	+ (\nabla_{b_1}\xi^{b'_1}) t^{a_1\cdots a_k}_{b'_1\cdots b_l} + \cdots
	+ (\nabla_{b_l}\xi^{b'_l}) t^{a_1\cdots a_k}_{b_1\cdots b'_l}
	.
	\label{eq:tens-rule}
	\end{align}
	where $\nabla_c$ is the Levi-Civita connection defined by the metric
	$g_{ab}$ (though the same formula also holds with any symmetric
	connection). We have specifically chosen the curvature coordinates
	\begin{equation}
		(x^a, \Gamma^a_{bC}, Q_{ab,C},
		t^{a_1\cdots a_{k_1}}_{b_1\cdots b_{l_1},C}, \ldots
		t^{a_1\cdots a_{k_N}}_{b_1\cdots b_{l_N},C})
	\end{equation}
	on $J^rBM$ so that the coordinate components of the Lie derivative of
	a holonomic $r$-jet, say the extension $j^r\beta$ of a section
	$\beta\colon M \to BM$, are given by
	(Lemma~\ref{lem:curv-coord-struct}(b))
	\begin{align}
	\label{eq:lie-Gamma}
	\Gamma^a_{C} \circ \Lie_\xi j^r\beta
	&= \del_C \xi^a
	+ \text{l.o.t}_{|C|-1;r+1}(\xi;g,\Gamma,Q)^a_C \circ j^{r+1}\beta , \\
	\label{eq:lie-g}
	g_{ab} \circ \Lie_\xi j^r \beta
	&= \left(g_{ca} \nabla_b \xi^{c} + g_{cb} \nabla_a \xi^{c}\right)
	\circ j^{1}\beta , \\
	\label{eq:lie-Q}
	Q_{ab,C} \circ \Lie_\xi j^r\beta
	\notag
	&= 2\nabla_{(a} \xi^{b'} Q_{b)b',C} \circ j^r\beta
	+ |C|\nabla_{(c_1|} \xi^{c'} Q_{ab,c'|c_2\ldots c_{|C|})} \circ j^r\beta \\
	\notag
	& \quad {}
	+ \xi^{c} \nabla_{c} (Q_{ab,C} \circ j^r\beta) \\
	\notag
	&= \left(
	2\nabla_{(a} \xi^{b'} Q_{b)b',C}
	+ |C|\nabla_{(c_1|} \xi^{c'} Q_{ab,c'|c_2\ldots c_{|C|})}
	\phantom{\frac{2}{|C|+1}} \right. \\
	\notag
	& \quad {}
	+ \xi^{c} Q_{ab,cC}
	+ \frac{2}{|C|+2} \xi^c Q_{c(a,b)C}
	+ \frac{|C|}{|C|+2} \xi^c Q_{(c_1c_2,c_2\cdots c_{|C|})abc} \\
	& \quad {} \left. \rlap{$\displaystyle\phantom{\frac{2}{|C|+1}}$}
	+ \xi^c \text{l.o.t}_r(g,Q)_{c,ab,C}
	\right) \circ j^{r+1}\beta , \\
	t^{a_1\cdots a_{k_i}}_{b_1\cdots b_{l_i}} \circ \Lie_\xi j^r\beta
	\notag
	&= \xi^c \nabla_c t^{a_1\cdots a_{k_i}}_{b_1\cdots b_{l_i}}
	+ \nabla_{b_1} \xi^{b_1'} t^{a_1\cdots a_{k_i}}_{b_1'\cdots b_{l_i}}
	+ \cdots
	+ \nabla_{b_{l_i}} \xi^{b_{l_i}'} t^{a_1\cdots a_{k_i}}_{b_1\cdots b_{l_i}'}
	\\
	& \quad {}
	- \nabla_{a_1'}\xi^{a_1} t^{a_1'\cdots a_{k_i}}_{b_1\cdots b_{l_i}}
	- \cdots
	- \nabla_{a_{k_i}'}\xi^{a_{k_i}} t^{a_1\cdots a_{k_i}'}_{b_1\cdots b_{l_i}},
	\end{align}
	where $\text{l.o.t}_{|C|-1;r+1}(\xi;g,\Gamma,Q)$ stands for terms that
	may only involve coordinates on $J^{r+1}BM$ and derivatives of $\xi^a$
	up to order $|C|-1$, while $\text{l.o.t}_r(g,Q)$ does not depend on
	$\Gamma$ or $\xi$ and depends equivariantly on the $g$ and $Q$
	coordinates on $J^rBM$. Let us also introduce the following notation for
	the components of the vertical tangent map $VF$, with $(\beta^r,
	\dot{\beta}^r) \colon M \to VJ^rBM$:
	\begin{multline} \label{eq:inf-eqvar-rhs}
	VF \circ(\beta^r, \dot{\beta}^r)
	= \left(F\circ\beta^r, (g_{ab}\circ \dot{\beta}^r)
	[(\del_g)^{ab} F] \circ \beta^r \right.\\
	+ (\Gamma^{a}_{bC} \circ \dot{\beta}^r)
	[(\del_\Gamma)_{a}^{bC} F] \circ \beta^r
	+ (Q_{ab,C} \circ \dot{\beta}^r)
	[(\del_{Q})^{abC} F] \circ \beta^r \\ \left. {}
	+ (t^{a_1\cdots a_{k_i}}_{b_1\cdots b_{l_i},C} \circ \dot{\beta}^r)
	[(\del_t)_{a_1\cdots a_{k_i}}^{b_1\cdots b_{l_i}C} F]
	\right) .
	\end{multline}
	When $\dot{\beta}^r = \Lie_\xi j^r\beta$, the compositions with
	$\dot{\beta}^r$ should be expanded using Equations~\eqref{eq:lie-Gamma},
	\eqref{eq:lie-g} and~\eqref{eq:lie-Q}. What is important to note is that
	all the resulting terms, with the exception of those proportional to
	$(\del_\Gamma)_{a}^{bC}({\cdots})$, will be proportional to either
	$\xi^a$ or $\nabla_a \xi^b$.
	
	On the other hand, translating the tensor Lie derivative to covariant
	derivatives, we have
	\begin{multline} \label{eq:inf-eqvar-lhs}
	\Lie_\xi (F\circ j^r\beta)^{a_1\cdots a_k}_{b_1\cdots b_l}
	= \xi^{a'} \nabla_{a'} (F^{a_1\cdots a_k}_{b_1\cdots b_l} \circ j^r\beta) \\
	- (\nabla_{a'_1}\xi^{a_1}) F^{a'_1\cdots a_k}_{b_1\cdots b_l} \circ j^r\beta - \cdots
	- (\nabla_{a'_k}\xi^{a_k}) F^{a_1\cdots a'_k}_{b_1\cdots b_l} \circ j^r\beta \\
	+ (\nabla_{b_1}\xi^{b'_1}) F^{a_1\cdots a_k}_{b'_1\cdots b_l} \circ j^r\beta + \cdots
	+ (\nabla_{b_l}\xi^{b'_l}) F^{a_1\cdots a_k}_{b_1\cdots b'_l} \circ j^r\beta
	.
	\end{multline}
	What is important to note is that each term is proportional to either
	$\xi^a$ or $\nabla_a \xi^b$. In principle, we could expand the $\xi^{a'}
	\nabla_{a'} (F^{a_1\cdots a_k}_{b_1\cdots b_l} \circ j^r\beta)$ term
	further, by using the chain rule. However, the chain rule here cannot be
	written solely in terms of covariant derivatives and the explicit
	expression in coordinate derivatives leads to rather complicated
	formulas that will not be immediately necessary.
	
	Then, choosing $\xi^a(x) = 0$ and $\nabla_a\xi^b(x) = 0$
	(Lemma~\ref{lem:xi-taylor}), the infinitesimal equivariance
	condition~\eqref{eq:inf-eqvar}, expanded using
	Equations~\eqref{eq:inf-eqvar-lhs} and~\eqref{eq:inf-eqvar-rhs}, by
	eliminating all terms proportional to $\xi^a$ or $\nabla_a\xi^b$ reduces
	to
	\begin{equation}
	\left(\del_{C} \xi^a + \text{l.o.t}_{|C|-1;r+1}(\xi;g,\Gamma,Q)^a_C\right)
	(\del_\Gamma)_a^{C} F^{a_1\cdots a_k}_{b_1\cdots b_l} \circ
	j^{r+1}\beta(x) = 0 ,
	\end{equation}
	for arbitrary $j^{r+1}\beta$ at $x\in M$. Since $\del_C \xi^a(x)$ can
	still be chosen arbitrarily, for $|C|\ge 2$, even with $\xi^a(0) = 0$
	and $\nabla_a \xi^b(x) = 0$ (Lemma~\ref{lem:xi-taylor}), we find that
	$(\del_\Gamma)_a^{C} F^{a_1\cdots a_k}_{b_1\cdots b_l}(x,g,\Gamma,Q) = 0$
	at any value of its arguments. In other words, we have part of the
	desired conclusion:
	\begin{equation}
	F^{a_1\cdots a_k}_{b_1\cdots b_k}(x,g,\Gamma,Q,t)
	= H^{a_1\cdots a_k}_{b_1\cdots b_k}(x,g,Q,t) ,
	\end{equation}
	with $H\colon J^rBM \to T^k\otimes T^{*l}M$ still $\mathrm{Diff}(M)$-equivariant. It
	remains to show that $H$ is a $GL(n)$-equivariant function of its
	tensorial arguments at any $x\in M$ and that the dependence on $x^a$ is
	trivial.
	
	Since the dependence on the $\Gamma$-coordinates is trivial, choosing
	$\xi^a(x) = 0$ and $\nabla_a \xi^b(x) = J_a^b \in GL(n)$ at some point
	$x\in M$, the infinitesimal equivariance condition~\eqref{eq:inf-eqvar},
	again expanded using Equations~\eqref{eq:inf-eqvar-lhs}
	and~\eqref{eq:inf-eqvar-rhs}, simplifies to
	\begin{multline}
	- J_{a'_1}^{a_1} H^{a'_1\cdots a_k}_{b_1\cdots b_l} - \cdots
	- J_{a'_k}^{a_k} H^{a_1\cdots a'_k}_{b_1\cdots b_l}
	+ J_{b_1}^{b'_1} H^{a_1\cdots a_k}_{b'_1\cdots b_l} + \cdots
	+ J_{b_l}^{b'_l} H^{a_1\cdots a_k}_{b_1\cdots b'_l}
	\\
	= \left(2 J^{b'}_{(a} Q_{b)b',C}
	+ |C| J^{c'}_{(c_1|} Q_{ab,c'|c_2\cdots c_{|C|})} \right)
	(\del_Q)^{abC} H^{a_1\cdots a_k}_{b_1\cdots b_l}
	\\
	+ \left(g_{ca} J_b^c + g_{cb} J_a^c\right)
	(\del_g)^{ab} H^{a_1\cdots a_k}_{b_1\cdots b_l}
	\\
	+ \Big(
	J_{(c_1|}^{c'} t^{a_1\cdots a_{k_i}}_{b_1\cdots b_{l_i},c'|c_2\cdots c_{|C|})}
	- J_{a_1'}^{a_1} t^{a_1'\cdots a_{k_i}}_{b_1\cdots b_{l_i},C}
	-\cdots -J_{a_{k_i}'}^{a_{k_i}} t^{a_1\cdots a_{k_i}'}_{b_1\cdots b_{l_i},C}
	\\
	+ J_{b_1}^{b_1'} t^{a_1\cdots a_{k_i}}_{b_1'\cdots b_{l_i},C}
	+\cdots + J_{b_{l_i}}^{b_{l_i}'} t^{a_1\cdots a_{k_i}}_{b_1\cdots b_{l_i}',C}
	\Big)
	(\del_t)_{a_1\cdots a_{k_i}}^{b_1\cdots b_{l_i}C}
	H^{a_1\cdots a_k}_{b_1\cdots b_l}
	,
	\end{multline}
	where $H$ is seen as a function on the product of the tensor bundles of
	appropriate ranks and index structures. We were justified by cancelling
	the composition with $(x, g, Q, t) \circ j^r\beta$ because $\beta$
	can be chosen so that these coordinate components have arbitrary values
	with respect to the corresponding tensor type. The resulting identity is
	precisely the infinitesimal version of the $GL(n)$-equivariance
	condition~\cite{kms,slovak} of $H$ at $x\in M$, with $x$ arbitrary.
	
	Finally, we need to show that $H^{a_1\cdots a_k}_{b_1\cdots b_l}(x, g,
	Q, t)$ is actually independent of $x^a$ in any adapted coordinate system
	on $J^rBM$ induced from coordinates $(x^a)$ on $M$. We apply the
	equivariance condition~\eqref{eq:inf-eqvar} to $H$ with an arbitrary
	choice of $\xi^a(x) \ne 0$ and $\del_a \xi^b(x) = 0$ at some $x\in M$.
	In this case, according to~\eqref{eq:tens-rule-coord}, we can replace
	the action of $\Lie_\xi$ on tensors at $x\in M$ by the derivative
	operator $\xi^{a'} \del_{a'}$ in adapted coordinates induced by the
	coordinates $(x^a)$:
	\begin{multline}
	(\xi^{a'} \del_{a'} g_{ab}\circ j^r\beta)
	(\del_g)^{ab} H^{a_1\cdots a_k}_{b_1\cdots b_l} \circ j^r\beta
	\\
	+ (\xi^{a'} \del_{a'} Q_{ab,C} \circ j^r\beta)
	(\del_Q)^{abC} H^{a_1\cdots a_k}_{b_1\cdots b_l} \circ j^r\beta
	\\
	+ \left(\xi^{a'}\del_{a'} t^{a_1\cdots a_{k_i}}_{b_1\cdots b_{l_i},C}\right)
	(\del_t)_{a_1\cdots a_{k_i}}^{b_1\cdots b_{l_i}C}
	H^{a_1\cdots a_k}_{b_1\cdots b_l} \circ j^r\beta
	\\
	= \xi^{a'} \del_{a'}|_{g,Q,t=\text{const.}}
	H^{a_1\cdots a_k}_{b_1\cdots b_l} \circ j^r\beta
	+ (\xi^{a'} \del_{a'} g_{ab}\circ j^r\beta)
	(\del_g)^{ab} H^{a_1\cdots a_k}_{b_1\cdots b_l} \circ j^r\beta
	\\
	+ (\xi^{a'} \del_{a'} Q_{ab,C} \circ j^r\beta)
	(\del_Q)^{abC} H^{a_1\cdots a_k}_{b_1\cdots b_l} \circ j^r\beta
	\\
	+ \left(\xi^{a'}\del_{a'} t^{a_1\cdots a_{k_i}}_{b_1\cdots b_{l_i},C}\right)
	(\del_t)_{a_1\cdots a_{k_i}}^{b_1\cdots b_{l_i}C}
	H^{a_1\cdots a_k}_{b_1\cdots b_l} \circ j^r\beta
	.
	\end{multline}
	Cancelling the common terms from both sides of the above identity, we
	obtain the condition
	\begin{equation}
	\xi^{a'} \del_{a'}|_{g,Q=\text{const.}}
	H^{a_1\cdots a_k}_{b_1\cdots b_l} \circ j^r\beta = 0.
	\end{equation}
	Because the choices of $x\in M$, $\xi^a(x)$ and $j^r\beta(x)$ were all
	arbitrary, we can then conclude that $H^{a_1\cdots a_k}_{b_1\cdots
		b_l}(x, g, Q, t) = G^{a_1\cdots a_k}_{b_1\cdots b_l}(g, Q, t)$, in an
	arbitrary adapted coordinate system on $V^r_x\subset J^rBM$, for some
	function $G\colon V^r_x\subset J^rBM \to T^k\otimes T^{*l}M$.
\end{proof}

\section{Invariant theory} \label{section_equivar}
The goal of this section is to state and prove the Equivariance
Lemma~\ref{lem_gen_equiv}, which generalizes some results proven
in~\cite[Sec.2.6]{KM16}. This Lemma is used in the proof of our main
Theorem~\ref{thm_uniqueness} to characterize all smooth
$GL(n)$-equivariant (resp.~$GL^+(n)$-equivariant, if we restrict
ourselves to transformations that preserve spacetime orientation)
tensor-valued maps that depend on a Lorentzian metric and any number of
tensorial arguments.

The main difference with the previous weaker~\cite[Lem.2.8]{KM16} is the
allowed dependence on other tensors besides the metric. As a result of
this generalization, the final characterization is a bit more
complicated. In particular, while any such equivariant map is still
polynomial in the metric $g$, its inverse $g^{-1}$ and possibly the
Levi-Civita tensor $\varepsilon(g)$, it may depend on the additional
tensor arguments $z$ in two different ways. First, being tensor-valued,
any such equivariant may will be \emph{polynomially} and
\emph{covariantly} constructed from $g$, $g^{-1}$, $\varepsilon(g)$ and
the tensor components of $z$, but the coefficients in these polynomial
will be allowed to depend in an essentially arbitrary \emph{smooth} way
on invariant scalar polynomials built out of $g$, $g^{-1}$,
$\varepsilon(g)$ and the tensor components of $z$.

The precise statements and proofs of these results depend on some
fundamental notions and facts from \emph{classical invariant theory} of
the $GL(n)$ and $O(1,n-1)$ (resp.~$GL^+(n)$ and $SO(1,n-1)$ in the
oriented case) groups. Invariant theory, which studies invariants of
linear representations of groups and other related topics) is a highly
developed subject (we will only mention~\cite{procesi}
and~\cite{goodman-wallach} as an introduction to the literature), but
the majority of the literature, especially at the introductory level,
focuses on polynomial invariants on representations of complex algebraic
groups. Thus, it is not always easy to locate some (even classical)
results in the context of real Lie groups and smooth (rather than
polynomial) invariants. For the convenience of the reader, we summarize
the relevant notions and results below and, when possible, try to
provide reasonably concise and elementary proofs that are not easy to
extract from the literature.

In the following we will use the one point space $\ast\cong\mathbb{R}^0$
with the trivial action of $GL(n)$ of any of its subgroup thereon.
\begin{defn}
	Let $X$ and $Y$ be spaces carrying actions of the group $G$, respectively $\rho^{(X)}_u \colon X\to X$ and $\rho^{(Y)}_u \colon Y \to Y$ for $u\in G$, in terms of bijective maps resp. $X\rightarrow X$ and $Y\rightarrow Y$. A map $f\colon X\rightarrow Y$ is said to be \textbf{equivariant} if it commutes with the action of $G$:
	\begin{equation*}
	f\circ\rho_u^{(X)}=\rho_u^{(Y)}\circ f,\quad \text{for every $u\in G$.}
	\end{equation*}
	In the spacial case $Y = \mathbb{R}$ carrying the trivial
	representation, an equivariant $f\colon X \to \mathbb{R}$ is called a
	\textbf{(scalar) invariant}. We denote the space of all scalar
	invariants by $\mathcal{S}_X$. When $X$ is a vector space, we denote
	the subspace of \textbf{(scalar) polynomial invariants} by
	$\mathcal{P}_X \subseteq \mathcal{S}_X$. The subspace $\mathcal{P}^k_X
	\subset \mathcal{P}_X$ consists of all homogeneous polynomials of
	degree $k$.
\end{defn}

With the above definitions, it is easy to establish a relation between
scalar invariants and equivariant maps for linear group representations
by the following obvious
\begin{prop} \label{prop_equiv_inv}
	Let $X$ and $Y$ be finite dimensional vector spaces with linear representations of the
	group $G$, and denote by $Y^*$ the dual linear of $Y$ equipped with the contragredient representation of
	$G$. If $f\colon X \to Y$ is an equivariant
	map, then $f^*(x,y^*) := y^*\cdot f(x)$ is a scalar invariant
	$f^*\colon X\times Y^* \to \mathbb{R}$. If $h\colon X\times Y^* \to
	\mathbb{R}$ is a scalar invariant, then $\frac{\del h}{\del
	y^*}|_{y^*=0} \colon X \to Y$ is an equivariant map. Moreover, for any
	equivariant map $f\colon X\to Y$, $\frac{\del f^*}{\del y^*}|_{y^*=0}
	= f$.
\end{prop}

\begin{defn}
	\label{tensor_rep}
	Let $M_n^p$ be the space of $p$-multilinear forms on $\mathbb{R}^n$ and consider the natural linear action of $GL(n)$ thereon. Denote by ${M^p_n}^*$ the dual of $M_n^p$, with the contragredient $GL(n)$ representation on it. Let $T$ be a finite-dimensional real vector space carrying a representation of $GL(n)$.
	\begin{enumerate}
		\item If $T$, with respect to some linear embedding $T\hookrightarrow M_n^p \otimes {M_n^q}^*$, is invariant under the action of $GL(n)$, and if (the representation carried by) $T$ is the restriction of the action of $GL(n)$ on $M_n^p \otimes {M_n^q}^*$, then  $T$ is called \textbf{tensor representation} of $GL(n)$. We call $(p,q)$ the \textbf{(covariant, contravariant) tensor rank} of $T$ and $p+q$ the \textbf{total tensor rank} of $T$.
		\item If $T$ is as in 1.,\ but the action of $GL(n)\ni u\mapsto\rho(u)$ on $T$ is given by a tensor representation up to a multiplication by $\left|\det u\right|^s$, then $T$ is called \textbf{tensor density representation} of $GL(n)$. We call $s$ the tensor \textbf{weight} of $T$.
		\item Denote by $\eta \in M_n^2$ is the standard Minkowski metric with signature $({-}{+}\cdots{+})$, and by $\epsilon\in M_n^n$ the standard antisymmetric Levi-Civita tensor. The orthogonal subgroup $O(1,n-1) \subset GL(n)$ (resp.~$SO(1,n-1) \subset GL^+(n)$) is the stabilizer subgroup of $\eta$ under the action on $M_n^2$. A \emph{tensor (density) representation} of the orthogonal group is a restriction of a tensor density representation of the general linear group.
	\end{enumerate}
\end{defn}

\begin{rem} \label{rem_tens_orthog}
Clearly, since for any $u\in O(1,n-1)$, $\left|\det u\right|=1$ and
$(u^{-1})^T = \eta u \eta^{-1}$ in the fundamental representation, the
restriction of any two tensor density representations of $GL(n)$ to
$O(1,n-1)$ or $SO(1,n-1)$ are linearly equivalent as long as their total
tensor rank is the same. So it is sufficient to talk only about
\emph{tensor} (rather than \emph{tensor density}) representations of
these subgroups. 
\end{rem}

\begin{rem} \label{rem_reductive}
Below, some results about a group $G$ and its representations require as
a hypothesis that $G$ be \emph{reductive}. There are several different
flavors of reductive groups (cf.~\cite[Sec.7.3]{procesi}), not all of
them being equivalent, with different ones serving as natural hypotheses
for different results. The general property that they share is that each
representation from a certain class is completely reducible
(\emph{i.e.},\ no reducible but indecomposable representations may
occur). For the sake of uniformity, we specialize all results stated
below to \emph{linearly reductive} groups, even if the original result
could be stated under looser hypotheses. First, note that a \emph{real
(complex) algebraic} group is a subgroup of $GL(n;\mathbb{R})$
($GL(n;\mathbb{C})$), for some $n$, that is also a real (complex)
algebraic subvariety (it is defined by polynomial equations). A
\emph{real (complex) linearly reductive} group $G$ is a real (complex)
algebraic group such that each real (complex) finite dimensional
rational representation of $G$ is completely reducible. Here
\emph{polynomial} and \emph{rational} mean with respect to the matrix
elements of the embedding of $G$ into $GL(n;\mathbb{R})$
($GL(n;\mathbb{C})$). Obviously, any real algebraic group gives rise to
a complex algebraic group, its \emph{complexification}, simply by
extending the defining polynomial equations from $GL(n;\mathbb{R})$ to
$GL(n;\mathbb{C})$. \emph{A priori}, the property of being reductive is
different for a real algebraic group and its complexification.
Fortunately, we only need to appeal to such hypotheses for the real
orthogonal groups $O(1,n-1) := O(1,n-1;\mathbb{R})$ and $SO(1,n-1) :=
SO(1,n-1;\mathbb{R})$, both of which are known to be linearly reductive,
and so are their complexifications $O(1,n-1;\mathbb{C})$ and
$SO(1,n-1;\mathbb{C})$ (see~\cite[Sec.7.3.2]{procesi},
\cite[Sec.5.2]{rich-slo}). Unless explicitly mentioned, below we always
refer to real groups and their representations on real vector spaces.
\end{rem}

\begin{defn} \label{def_equivar_tens}
Let $L_n \subset M^2_n$ denote the space of Lorentzian bilinear forms
(non-degenerate, with signature $({-}{+}\cdots{+})$), and let it inherit
the natural action of $GL(n)$ (resp.~$GL^+(n)$). Let $Z = \bigoplus_j
Z_i$ and $T = \bigoplus_j T_j$ be finite sums of tensor density
representations of $GL(n)$ (resp.~$GL^+(n)$). We will refer to a smooth
equivariant map
\begin{equation}
	\tau\colon L_n \times Z \to T
\end{equation}
as a \textbf{$GL(n)$-equivariant tensor density}
(resp.~\textbf{$GL^+(n)$-equivariant tensor densities}).
The space of $GL(n)$-equivariant tensor densities will be denoted by
$\mathcal{E}_{Z,T}$. The space of $GL^+(n)$-equivariant tensor densities
will be denoted by $\tilde{\mathcal{E}}_{Z,T}$. In the special case when
$T = \mathbb{R}$ carries the trivial representation, we call
$\mathcal{S}_{Z} := \mathcal{E}_{Z,\mathbb{R}}$
(resp.~$\tilde{\mathcal{S}}_{Z} := \tilde{\mathcal{E}}_{Z,\mathbb{R}}$)
the space of \textbf{scalar invariants}.
\end{defn}

\begin{defn} \label{def_isotropic}
Let $Z = \bigoplus_i Z_i$ and $T = \bigoplus_j T_j$ be finite sums of
tensor representations of $O(1,n-1)$ (resp.~$SO(1,n-1)$). We will refer
to a smooth equivariant map
\begin{equation}
	\tau \colon Z \to T
\end{equation}
as the space of \textbf{$O(1,n-1)$-isotropic tensors}
(resp.~\textbf{$SO(1,n-1)$-isotropic tensors}).
The space of $O(1,n-1)$-isotropic tensors will be denoted by
$\mathcal{E}_{Z,T}$. The space of $SO(1,n-1)$-isotropic tensors will be
denoted by $\tilde{\mathcal{I}}_{Z,T}$.
\end{defn}

The above definitions can be contrasted with the Definitions~2.6--7
of~\cite{KM16}. There, the simpler notion equivariant and isotropic
tensors did not allow for dependence on the extra parameter space $Z$
and use the simpler notations $\mathcal{E}_{T} \cong
\mathcal{E}_{\ast,T}$ and $\mathcal{I}_{T} \cong \mathcal{I}_{\ast,T}$,
where $\ast = \mathbb{R}^0$ is the 1-point space or equivalently the
trivial vector space (with complete analogy in the oriented case).

\begin{prop} \label{prop_equiv_ext}
With the notation of Definitions~\ref{def_equivar_tens}
and~\ref{def_isotropic}, the space of equivariant tensor densities
$\mathcal{E}_{Z,T}$ (resp.~$\tilde{\mathcal{E}}_{Z,T}$) is isomorphic to
the space of isotropic tensors $\mathcal{I}_{Z,T}$
(resp.~$\tilde{\mathcal{I}}_{Z,T}$).
\end{prop}

\begin{proof}
Here we use the same logic as in~\cite[Lem.2.8]{KM16}, where it is
spelled out a bit less tersely. Let $\eta \in L_n \cong GL(n) / O(1,n-1)
\cong GL^+(n) / SO(1,n-1)$, where the orthogonal group is interpreted as
the stabilizer subgroup of $\eta$. The equivariance of $\tau\colon L_n
\times Z \to T$ implies that $\tau(\eta,uz) = u\tau(\eta,z)$, whenever
$u \in GL(n)$ and $u\cdot \eta = \eta$, meaning that $\tau(\eta,-)\colon
Z\to T$ is $O(1,n-1)$ (resp.~$SO(1,n-1)$) equivariant. On the other
hand, since any $L_n \ni g = u_g \cdot \eta$ for some $u_g \in GL(n)$,
the knowledge of $\tau(\eta,-)$ uniquely determines the equivariant
extension $\tau(g,z) := u_g \tau(\eta, u_g^{-1} z)$. Clearly, this
correspondence is bijective.
\end{proof}

For the fundamental representations of $O(1,n-1)$ and $SO(1,n-1)$,
homogeneous polynomials, invariant linear functionals and isotropic
tensors all have a very explicit description. We give this description
below in several different versions, related as follows. Any
\emph{polynomial} on a vector space that is invariant under the action
of a linear representation can be written as a sum of invariant
\emph{homogeneous} polynomials. Any invariant homogeneous polynomial of
degree $p$ is also naturally a linear functional on a $p$-fold symmetric
tensor product of the original representation and vice versa. By
duality, the adjoint of a linear functional on a $p$-fold tensor product
representation defines an equivariant map from $\ast$ to the dual of the
$p$-fold tensor product representation.

\begin{prop} \label{prop_fft}
Let $M^p_n$ and $\eta$ be as in Definition~\ref{def_isotropic}, let
$V^p=(\mathbb{R}^n)^p$ be the space of $(v_1,\ldots, v_p)$ of $p$-copies
of vectors in the fundamental representation of $O(1,n-1)$ (or
$SO(1,n-1)$), and let $T = \bigoplus_j T_j$ be a finite sum of tensor
representations of ranks $p_j$ of $O(1,n-1)$ (or $SO(1,n-1)$).
\begin{enumerate}
\item
	Polynomials $p(v_1,\cdots,v_p) \in \mathcal{P}_{V^p}$
	invariant under the simultaneous action of $O(1,n-1)$ on its arguments
	are generated by the contractions $\eta_{ab} v_i^a v_j^b$, with
	$i,j=1,\ldots,p$.\\
	Polynomials $p(v_1,\cdots,v_p) \in \mathcal{P}_{V^p}$
	invariant under the simultaneous action of $SO(1,n-1)$ on its
	arguments are generated by the contractions $\eta_{ab} v_i^a v_j^b$
	and $\epsilon_{a_1\cdots a_p} v_{i_1}^{a_1} \cdots v_{i_p}^{a_p}$,
	with $i,j,i_k=1,\ldots,p$.
\item
	The isotropic tensors $\mathcal{I}^p_n$ are linear combinations of
	tensor products of copies of $\eta_{ab}$ with arbitrarily permuted
	indices.\\
	The isotropic tensors $\tilde{\mathcal{I}}^p_n$ are spanned by tensor
	products of $\eta_{ab}$ and $\epsilon_{a_1\cdots a_n}$ with
	arbitrarily permuted indices.
\item
	All $O(1,n-1)$-invariant linear functionals on $M^p_n$ are spanned by
	arbitrary complete contractions of a tensor $t_{a_1\cdots a_p} \in
	M^p_n$ with copies of $\eta^{ab}$, in an arbitrary order of indices.\\
	All $SO(1,n-1)$-invariant linear functionals on $M^p_n$ are spanned by
	arbitrary complete contractions of a tensor $t_{a_1\cdots a_p} \in
	M^p_n$ with copies of $\eta^{ab}$ and $\epsilon^{a_1\cdots a_n}$, in
	an arbitrary order of indices.
\item
	All degree $k$ homogeneous polynomial scalar $O(1,n-1)$-invariants
	$p(t) \in \mathcal{P}^k_T$ on $T$ are spanned by complete contractions
	of tensor products
	\begin{equation} \label{k_tens_prod}
		(t_{j_1})_{a^1_{1}\cdots a^1_{p_{j_1}}} \cdots
		(t_{j_k})_{a^k_{1}\cdots a^k_{p_{j_k}}}
	\end{equation}
	with copies of $\eta^{ab}$, when $t = \bigoplus_j t_j$.\\ All degree
	$k$ homogeneous polynomial scalar $SO(1,n-1)$-invariants $p(t) \in
	\mathcal{P}^k_T$ on $T$ are spanned by complete contractions of tensor
	products~\eqref{k_tens_prod} with copies of $\eta^{ab}$ and
	$\epsilon^{a_1\cdots a_n}$.
\end{enumerate}
\end{prop}

This proposition sometimes goes under the name of the \emph{joint tensor
version of the First Fundamental Theorem (FFT)} of invariant theory of
$O(1,n-1)$ (respectively $SO(1,n-1)$). While this specific version is
well-known folklore, it is difficult to find with a concise statement and
proof, especially for part~3. Thus, we briefly sketch a proof below,
summarizing the arguments from~\cite[Sec.11.2.1]{procesi}
and~\cite[Sec.11.6.8]{procesi}.

\begin{proof}
Parts~2--4 basically follow from part 1, so we first discuss these
implications and then discuss a proof of part~1.

\emph{2.} Note that the scaling transformations $(v_1, \ldots, v_i,
\ldots, v_p) \mapsto (v_1, \ldots, \lambda v_i, \ldots, v_p)$ for
$\lambda \in \mathbb{R}$, with $i=1,\ldots, p$, commute with the action
of the group on $V^p$. Hence, any invariant polynomial on $V^p$
decomposes into a sum of invariant polynomials with fixed homogeneous
degrees in each of the vector arguments $v_1,\ldots,v_p$. Since we can
put $M^p_n$ in bijection with polynomials on $V^p$ which are homogeneous
of degree $1$ in each vector argument, the desired claims about
$\mathcal{I}^p_n$ and $\tilde{\mathcal{I}}^p_n$ immediately follow.

\emph{3.} Part~3 follows directly from part~2 by duality, since $\eta$
defines a non-degenerate inner product on $M^p_n$ by pairwise
contraction of indices.

\emph{4.} Since the decomposition $T = \bigoplus_j T_j$ is into tensor
representations, we have the equivariant embeddings $T_j \to M_n^{p_j}$.
As in the proof of part~2, the group representation on $T$ commutes with
separately multiplying each $T_j$ by a scalar. Hence, degree $k$
homogeneous polynomials, which are in bijection with the symmetric
tensor power $S^k T^*$, decompose into
\begin{equation*}
	\bigoplus_{\sum_j q_j =k} \bigotimes_j S^{q_j} T_j^*
\end{equation*}
and invariant polynomials respect this decomposition. Thus, to
characterize all invariant polynomials in $T$, it is sufficient to
characterize invariant linear functionals on spaces of the form
$\bigotimes_j S^{q_j} T_j$, each of which come with equivariant
embeddings into $M_n^{P}$ with $P=\sum_j q_j p_j$. The pullback along
this embedding is a surjective equivariant map $(M_n^P)^* \to
\bigotimes_j S^{q_j} T_j^*$. Now, invoking the fact that both $O(1,n-1)$
and $SO(1,n-1)$ are \emph{linearly reductive} groups
(Remark~\ref{rem_reductive}), both $(M_n^P)^*$ and $\bigotimes_j S^{q_j}
T_j^*$ decompose into direct sums of irreducible representations (for
either group). The equivariance of the pullback map means that it
diagonalizes with respect to the decomposition of the two spaces into
isotypic components (maximal subspaces consisting of copies of a single
irreducible representation) and its surjectivity means that it remains
surjective on each isotypic component. The subspace invariant under the
action of the group is simply one of the isotypic components
(corresponding to the trivial representation) and hence every invariant
linear functional on $\bigotimes_j S^{q_j} T_j$ comes from pulling back
an invariant linear functional from $M_n^P$. Finally, the result of
part~3 implies the desired structure of invariant polynomials on $T$.

\emph{1.}
To prove part~1, we first reduce $p$ to $p=n$. Then, we proceed by
induction on $n$. We will use several times the following elementary
fact: if both the variables $y^i$ and a polynomial $p(x;y) = \sum_I
p_I(x) y^I$ are invariant under a group action, then the individual
coefficients $p_I(x)$ are also individually invariant. Another useful
elementary fact is that two polynomials that agree on a non-empty open
set agree everywhere.

First, assume that the desired conclusion holds for $V^n$. If $p<n$,
then the desired conclusion follows from identifying invariant
polynomials $p(v_1,\ldots,v_p)$ with invariant polynomials
$p_1(v_1,\ldots,v_p,\ldots, v_n)$ that are constant with respect to the
$v_{p+1},\ldots,v_n$ arguments. Considering an invariant polynomial
$p(v_1,\ldots, v_p)$ with $p>n$, we can restrict it to the \emph{open} subset
of $V^p$ where the first $v_1,\ldots, v_n$ vectors are linearly
independent. Then, for $i>n$, we can write $v_i = \sum_{j=1}^n w^j_i
v_j$, where the $w^i_j$ are invariant scalars. In fact, the $w^j_i$ can
be explicitly written as polynomials in the contractions $\eta_{ab}
v_j^a v_k^b$ and $(\det_{1\le j,k \le n} \eta_{ab} v_j^a
v_k^b)^{-1}$. Hence,
\begin{align*}
	p(v_1,\ldots, v_n, v_{n+1}, \ldots , v_p)
	&= p\left(v_1,\ldots, v_n, \sum_j w^j_{n+1} v_j, \ldots , \sum_j w^j_p v_j\right) \\
	&= p_1(v_1,\ldots,v_n; w^j_i) \\
	&= p_2\left(\eta_{ab} v_j^a v_k^b, \epsilon_{a_1\cdots a_n}
		v_1^{a_1} \cdots v_n^{a_n}; w^j_i\right) \:,
\end{align*}
where $p_1$ is also polynomial in its arguments and $p_2$ is another
polynomial that exists by applying to the $w$-coefficients of $p_1$ our
earlier hypothesis that the desired conclusion holds for $V^n$. The
contractions in the arguments of $p_2$ up to the semicolon involve only
the vectors $v_1,\ldots, v_n$. Plugging in the explicit rational
expressions for the $w^i_j$ into the arguments of $p_2$, since the
result equals the polynomial $p$, all the denominators must cancel
and we end up with an identity
\begin{equation*}
	p(v_1, \ldots, v_p)
		= p_3(\eta_{ab} v_j^a v_k^b,
			\epsilon_{a_1\cdots a_n} v_{i_1}^{a_1} \cdots v_{i_n}^{a_n}) \:,
\end{equation*}
where now the contractions may involve any of the $v_1,\ldots, v_p$
vectors, which holds for some polynomial $p$ on an open subset of $V^p$
and hence everywhere. Of course, the contractions with $\epsilon$ appear
only in the case of $SO(1,n-1)$.

Next, assume the inductive hypothesis that the desired conclusion holds
for $O(1,n'-1)$ and $SO(1,n'-1)$ for all $0<n'<n$, with the $n'=1$ case
being trivial. Consider an invariant polynomial $p(v_1,\ldots, v_n)$,
which we can restrict to the \emph{open} subset of $V^n$ where arguments are
linearly independent and the first $v_1,\ldots, v_{n-1}$ vectors span a
hyperplane with a \emph{$\eta$-spacelike} oriented unit normal vector
$\mathbf{v}$. Then, we can always write
\begin{equation*}
	v_n = \lambda \mathbf{v} + w^1 v_1 + \cdots + w^{n-1} v_{n-1} \;,
\end{equation*}
where the $w^i$ are invariant scalars. In fact, the $w^i$ can be
explicitly written as polynomials in $\eta_{ab} (v_i)^a (v_n)^b$ and
$(\eta_{ab} (v_i)^a (v_i)^b)^{-1}$. Hence,
\begin{align*}
	p(v_1,\ldots, v_n)
	&= p(v_1,\ldots, v_{n-1},
		\lambda \mathbf{v} + w^1 v_1 + \cdots + w^{n-1} v_{n-1}) \\
	&= \sum_k p_{1,k}(v_1,\ldots,v_{n-1}; w^i) \lambda^k \;,
\end{align*}
for some polynomials $p_{1,k}$. Consider for now only the invariance of
$p$ under the subgroup of $SO(1,n-1)$ that fixes the spacelike vector
$\mathbf{v}$ up to a sign, which corresponds to $\lambda \mapsto \pm
\lambda$ (the negative sign only accompanies those transformations that
change the orientation of the hyperplane orthogonal to $\mathbf{v}$).
Let us identify this orthogonal hyperplane with $\mathbb{R}^{n'}$, where
$n'=n-1$ and let the $\eta'$ and $\epsilon'$ denote the restrictions of
$\eta$ and $\epsilon$ to the hyperplane. The above mentioned subgroup
can hence be identified with $O(1,n'-1)$ acting on $\mathbb{R}^{n'}$ and
preserving $\eta'$. The invariance of $p$, together with the
identification of $v_1,\ldots, v_{n-1}$ with vectors in $\mathbb{R}^{n'}$,
implies that the $w$-coefficients of the $p_{1,k}$ are invariant under
$O(1,n'-1)$ for even $k$ and invariant under $SO(1,n'-1)$ for odd $k$,
and more specifically the odd $k$ coefficients are also odd under the
change of orientation of $\mathbb{R}^{n'}$. It is now helpful to note
that $\lambda^2$ can be written as a polynomial in $\eta_{ab} (v_i)^a
(v_j)^b$ (including $i,j=n$) and $w^i$, while
\begin{equation*}
	\epsilon'_{a_1\cdots a_{n-1}} v_1^{a_1} \cdots v_{n-1}^{a_{n-1}} \lambda
	= \epsilon_{a_1\cdots a_n} v_1^{a_1} \cdots v_{n-1}^{a_{n-1}} v_n^{a_n} .
\end{equation*}
Thus, for even $k$, we have
\begin{align*}
	p_{1,k}(v_1,\ldots, v_{n-1}; w^i) \lambda^k
	&= p_{2,k}(\eta'_{ab} v_i^a v_j^b; w^i) (\lambda^2)^{k/2} \\
	&= p_{3,k}(\eta_{ab} v_i^a v_j^b; w^i) \:,
\end{align*}
for some polynomials $p_{2,k}$, with $i,j\le n-1$ in its arguments
before the semicolon, and $p_{3,k}$, with $i,j\le n$ in its arguments
before the semicolon. For odd $k$, we have
\begin{align*}
	p_{1,k}(v_1,\ldots, v_{n-1}; w^i) \lambda^k
	&= p_{2,k}(\eta'_{ab} v_i^a v_j^b; w^i)
		\epsilon'_{a_1\cdots a_{n-1}} v_1^{a_1} \cdots v_{n-1}^{a_{n-1}}
		\lambda (\lambda^2)^{(k-1)/2} \\
	&= p_{3,k}(\eta_{ab} v_i^a v_j^b; w^i)
		\epsilon_{a_1\cdots a_n} v_1^{a_1} \cdots v_n^{a_n} \:,
\end{align*}
for some polynomials $p_{2,k}$, with $i,j\le n-1$ in its arguments
before the semicolon, and $p_{3,k}$, with $i,j\le n$ in its arguments
before the semicolon. Note that higher powers of $\epsilon'$ never
needed to be considered because of the usual identity relating tensor
powers of $\epsilon'$ with permutations of products of $\eta'$. Plugging
in the explicit rational expressions for the $w^i$ into the arguments of
$p_{3,k}$, since the result equals the polynomial $p$, all the
denominators must cancel and we end up with an identity
\begin{equation*}
	p(v_1, \ldots, v_n)
		= p_4(\eta_{ab} v_j^a v_k^b,
			\epsilon_{a_1\cdots a_n} v_{i_1}^{a_1} \cdots v_{i_n}^{a_n}) \:,
\end{equation*}
for some polynomial $p_4$, where now the contractions may involve any of
the $v_1,\ldots, v_p$ vectors. It is important to note that, up until
now, the above identity has only been established for
$v_1,\ldots,v_{n-1}$ orthogonal to a given spacelike unit vector
$\mathbf{v}$. Fortunately, once we note that there always exists a
transformation in $SO(1,n-1)$ that will transform any set of $n-1$
vectors orthogonal to another unit spacelike vector $\mathbf{v}'$ into a
set of $n-1$ vectors orthogonal to a given $\mathbf{v}$, the invariance
of both the original polynomial $p$ and the individual invariance of
each argument of the polynomial $p_4$ implies that the above identity
between $p$ and $p_4$ holds for any $\mathbf{v}$. So the above identity
between $p$ and $p_4$ holds on an open subset of $V^n$ and hence
everywhere. Of course, the contractions with $\epsilon$ appear only in
the case of $SO(1,n-1)$.

This concludes the proof.
\end{proof}

\begin{rem}
In the inductive step of the above proof, we reduced the problem from
$O(1,n-1)$ to $O(1,n-2)$, by restricting to a subspace orthogonal to a
\emph{spacelike} vector, relying crucially also on the \emph{transitive}
action of $O(1,n-1)$ on the \emph{open} subset of spacelike vectors.
Clearly, the inductive step could have also used a timelike vector
instead, without interfering with these crucial properties. It should
also be clear that the same argument would work directly in the case of
any $O(p,q)$, with reductions to either $O(p,q-1)$ or $O(p-1,q)$ both
being possible inductive steps.
\end{rem}

Before stating and proving our Equivariance Lemma~\ref{lem_gen_equiv},
we need the following fundamental results from invariant theory.

\begin{prop}[Hilbert~{\cite[Sec.14.1]{procesi}}, {\cite[\textsection7.2]{michor-dg}}] \label{prop_hilbert}
Let $G$ be a linearly reductive group with a rational representation on
a finite dimensional vector space $Z$. Then the algebra of polynomial
scalar $G$-invariants on $Z$ is finitely generated.
\end{prop}

\begin{defn} \label{def_stab_inv}
Let $G$ be a linearly reductive group with a rational representation on
a finite dimensional vector space $Z$ and let $p_i \in C^\oo(Z)$,
$i=1,\ldots, N_Z$, be a generating set for the algebra of polynomial
scalar $G$-invariants (Proposition~\ref{prop_hilbert}). A smooth
function $\sigma \in C^\oo(Z)$ is said to be \textbf{stably
$G$-invariant} if it is constant along each joint level set of the
invariant polynomials $p_i$, $i=1,\ldots N_Z$.
\end{defn}

Clearly, any function that is stably $G$-invariant is also
$G$-invariant, but the converse is not always true. Also, it is easy to
see that the definition is independent of the choice of the generating
polynomials $p_i$. The \emph{stability} in this definition is meant with
respect to complexification, since upon replacing $G$ with its
complexification the orbits become larger, while the invariant
polynomials remain the same, in a way that invariant polynomials do
completely separate all closed orbits, which erases the difference
between $G$-invariant and stably $G$-invariant functions. In
Section~\ref{sec_vector_KG_xi_tens}, we discuss the action of $O(1,n-1)$
on the subspace of symmetric forms in $M_n^2$ (in this case, the action
coincides with that of $SO(1,n-1)$). There, we give an explicit list of
a generating set of scalar invariant polynomials and also discuss the
structure of the orbits. That case also gives an explicit example of the
difference between $G$-invariant and stably $G$-invariant functions,
because invariant polynomials do not separate closed orbits on symmetric
bilinear forms.

The following results seem to be close to the state of the art in
characterizing the \emph{smooth} scalar invariants that apply to our
cases of interest. Unfortunately, we actually require a somewhat
strengthened version of these results (though see
also~\cite{stoetzel-phd} for more recent work), which we state below in
Proposition~\ref{prop_luna_ext}, but whose proof we do not discuss
(Remark~\ref{rem_luna_ext}).

\begin{prop}[Luna~\cite{luna}, {\cite[\textsection7.14]{michor-dg}}
{\cite[\textsection26.3]{kms}}] \label{prop_luna}
Let $G$ be a linearly reductive group with a rational representation on
a finite dimensional vector space $Z$ and let $p_i \in \mathcal{P}_Z$,
$i=1,\ldots, N_Z$, be a generating set for the algebra of polynomial
scalar $G$-invariants (Proposition~\ref{prop_hilbert}). Then a smooth
stably $G$-invariant function $\sigma \in C^\oo(Z)$ can always be
written as $\sigma = \Sigma(p_1,\ldots, p_{N_Z})$ where $\Sigma$ is a
smooth function of its arguments.
\end{prop}

Though, as indicated above, the statement of Luna's theorem can be found
in several references, as far as we know, a proof is available only in
the original reference~\cite{luna}, written in French. However, the more
recent result on the structure of invariants of finite $C^k$
differentiability~\cite{rumberger} does use a proof that is logically
similar to Luna's.

\begin{prop}[Richardson~{\cite[Thms.2.3,4.1]{richardson}}] \label{prop_richardson}
Let $G$ be a linearly reductive group with a linearly reductive
complexification and a rational representation of on a finite
dimensional vector space $Z$. Let $p_i \in C^\oo(Z)$, $i=1,\ldots,N_Z$,
be homogeneous polynomials generating the algebra of polynomial scalar
$G$-invariants on $Z$ (Proposition~\ref{prop_hilbert}). Then, there is a
$p_0 = P(p_1,\ldots, p_{N_Z})$ polynomial in its arguments and, with
$Z^0 = p_0^{-1}(0)$, a partition $Z \setminus Z^0 = \bigcup_j Z_j$ into
finite union of disjoint connected open subsets $(Z_j)$ where each $Z_j$
is stable under the action of $G$ and, for each $j$ and for any two
points $z_1,z_2\in Z_j$ the stabilizer subgroups $G_{z_1}, G_{z_2}
\subseteq G$ are conjugate in $G$.
\end{prop}

The following definition is rather technical, but is necessary to
precisely capture the difference between the behavior of \emph{smooth}
invariants and \emph{polynomial} invariants (or \emph{analytic}, or even
\emph{stable smooth} invariants).

\begin{defn} \label{def_loc_poly}
Let $Z$ be a finite dimensional vector space, $p_i \in C^\oo(Z)$,
$i=1,\ldots, N_Z$, be a set of homogeneous polynomials on $Z$, and $p_0
= P(p_1,\ldots, p_{N_Z})$ a polynomial in its arguments. With $Z^0 =
p_0^{-1}(0)$, consider a partition $Z \setminus Z^0 =
\bigcup_{j=1}^{r_Z} Z_j$ into pairwise disjoint open sets $Z_j$, for
some $r_Z < \oo$. We say that a function $\sigma \in C^\oo(Z)$ is
\textbf{locally a smooth function of the polynomials $p_i$ with respect
to the partition $(Z_j)$} if there exist $\Sigma_j \in
C^\oo(\mathbb{R}^{N_Z})$, $j=1,\ldots, r_Z$, such that $\sigma =
\Sigma_j(p_1,\ldots, p_{N_Z})$ on $Z_j$. We say that $\sigma$ is a
\textbf{function of the $p_i$ (globally)} if we can choose $\Sigma_j =
\Sigma_i$, for $i,j=1,\ldots,r_Z$. We write $\sigma =
[\Sigma]_Z(p_1,\ldots, p_{N_Z})$.
\end{defn}

\begin{prop}[extended Luna-Richardson] \label{prop_luna_ext}
Let $G$ be a linearly reductive group with a linearly reductive
complexification and a rational representation on a finite dimensional
vector space $Z$. Also, let $p_i \in \mathcal{P}_Z$, $i=1,\ldots, N_Z$,
be homogeneous polynomials generating the algebra of polynomial scalar
$G$-invariants on $Z$ (Proposition~\ref{prop_hilbert}). Then, there
exists a $p_0 = P(p_1,\ldots, p_{N_Z})$ polynomial in its arguments and,
with $Z^0 = p_0^{-1}(0)$, a partition $Z \setminus Z^0 =
\bigcup_{j=1}^{r_Z} Z_j$ into pairwise disjoint open $G$-invariant sets,
such that any $G$-invariant function $\sigma \in C^\oo(Z)$ is locally a
smooth function $\sigma = [\Sigma]_Z(p_1,\ldots, p_{N_Z})$ of the
polynomials $p_i$ with respect to the partition $(Z_j)$
(Definition~\ref{def_loc_poly}).
\end{prop}

\begin{rem} \label{rem_luna_ext}
The proof of Proposition~\ref{prop_luna_ext} follows from combining the
details of the proofs of Propositions~\ref{prop_luna}
and~\ref{prop_richardson}, which can be found in the original
references~\cite{luna} and~\cite{richardson} respectively. Discussing a
complete proof goes beyond the scope of the current work and will be
discussed elsewhere.
\end{rem}

Combining the results presented so far allows us to finally formulate
the main Equivariance Lemma that is needed in the proof of our main
Theorem~\ref{thm_uniqueness}.

\begin{lem}[Equivariance] \label{lem_gen_equiv}
Consider finite sums of tensor density representations $Z = \bigoplus_j
Z_j$ and $T = \bigoplus_j T_j$ of $GL(n)$ (resp.~$GL^+(n)$), and its
natural action on $L_n$. Recall also
(Definitions~\ref{def_equivar_tens}, \ref{def_isotropic}) the notion of
invariant scalars ($\mathcal{S}_{Z}, \tilde{\mathcal{S}}_{Z} \subset
C^\oo(L_n\times Z)$), equivariant tensors ($\mathcal{E}_{Z,T},
\tilde{\mathcal{E}}_{Z,T} \subset C^\oo(L_n\times Z;T)$) and isotropic
tensors ($\mathcal{I}_{Z,T}, \tilde{\mathcal{I}}_{Z,T} \subset
C^\oo(Z;T)$), as well as their characterizations
(Propositions~\ref{prop_equiv_inv}, \ref{prop_equiv_ext}
and~\ref{prop_fft})
\begin{enumerate}
\item
	There exist diagonalizable intertwiners $s_Z\colon Z \to Z$ and $s_T
	\colon T \to T$ such that $(u,z) \mapsto \left|\det u\right|^{-s_Z}
	(u\cdot z)$ and $(u,t) \mapsto \left|\det u\right|^{-s_T} (u\cdot t)$,
	for $u\in \mathrm{GL}(n)$, $z \in Z$ and $t\in T$, define
	\emph{tensor} representations (\emph{i.e.},\ with density weight zero)
	on $Z$ and $T$. Denoting these tensor representations by $Z'$ and
	$T'$, we have $\mathcal{E}_{Z,T} \cong \mathcal{E}_{Z',T'}$
	(resp.~$\mathcal{E}_{Z,T} \cong \mathcal{E}_{Z',T'}$).
\item
	When $Z$ carries a tensor representation and $p \in \mathcal{S}_Z$
	(resp.~$\tilde{\mathcal{S}}_Z$) such that $p(g,z)$ is polynomial in
	$z$, then $p$ is a covariantly constructed scalar that is polynomial
	in the tensor components of $g$, $g^{-1}$ and $z$ (resp.~of $g$,
	$g^{-1}$, $\varepsilon(g)$ and $z$).
\item
	There is a finite number of invariant $p_i \in \mathcal{S}_Z$
	(resp.~$\tilde{\mathcal{S}}_Z$), $i=1,\ldots, N_Z$, such that each
	$p_i(g,z)$ is a homogeneous polynomial in $z$ and each $\sigma \in
	\mathcal{S}_Z$ (resp.~$\tilde{\mathcal{S}}_Z$) is locally a smooth
	function $\sigma = [\Sigma]_Z(p_1,\ldots, p_{N_Z})$ of the invariant
	polynomials $p_i$, $i=1,\ldots, N_Z$, as in
	Proposition~\ref{prop_luna_ext}.
\item
	There is a finite number of equivariant tensors $q_j \in
	\mathcal{E}_{Z,T}$ (resp.~$\tilde{\mathcal{I}}_{Z,T}$), $j=1, \ldots,
	N_{Z,T}$, whose components are homogeneous polynomials on $Z$, such
	that each $\tau \in \mathcal{I}_{Z,T}$
	(resp.~$\tilde{\mathcal{E}}_{Z,T}$) is of the form $\tau =
	\sum_{j=1}^{N_{Z,T}} \sigma^j q_j$ with $\sigma^j \in \mathcal{S}_{Z}$
	(resp.~$\tilde{\mathcal{S}}_{Z}$).
\end{enumerate}
\end{lem}

\begin{proof}
After we establish point~1, we can without loss of generality assume that $Z$ and $T$ consist of direct sums of only \emph{tensor} representations. 

\emph{1.} By hypotheses, both $Z$ and $T$ reduce to a sum of tensor density
representations. This means that there exist diagonalizable intertwiners
$s_Z\colon Z\to Z$ and $s_T\colon T\to T$ such that $(u,z) \mapsto \left|\det
u\right|^{-s_Z} (u\cdot z)$ and $(u,t) \mapsto \left|\det
u\right|^{-s_T} (u\cdot t)$, for $u\in \mathrm{GL}(n)$, $z \in Z$ and
$t\in T$, define \emph{tensor} representations on $Z$ and $T$. Let us
refer to the corresponding representations as $Z'$ and $T'$. If
$\tau\colon L_n \times Z \to T$ is an equivariant map with respect to
the tensor density representations on $Z$ and $T$, then
\begin{equation}
	\tau'(g,z') = \left|\det g\right|^{-s_T} \tau(g, \left|\det g\right|^{s_Z} z')
\end{equation}
defines an equivariant map $\tau'\colon L_n \times Z' \to T'$ with
respect to the corresponding tensor representations. Clearly, this
operation can be reversed.

\emph{2.} Recall that, in our notation, $\mathcal{S}_{Z} \cong
\mathcal{E}_{Z,\mathbb{R}}$ (resp.~$\tilde{\mathcal{S}}_{Z} \cong
\tilde{\mathcal{E}}_{Z,\mathbb{R}}$) where $\mathbb{R}$ carries the
trivial representation. Then, by Proposition~\ref{prop_equiv_ext}, we have the
isomorphism $\mathcal{S}_{Z} = \mathcal{I}_{Z,\mathbb{R}}$
(resp.~$\tilde{\mathcal{S}}_{Z} = \tilde{\mathcal{I}}_{Z,\mathbb{R}}$).
Under this isomorphism, an invariant $p(g,z)$ is polynomial in $z$ iff
the corresponding $p_\eta(z) = p(g=\eta,z)$ is polynomial. Moreover, by
the classification Proposition~\ref{prop_fft}, any such polynomial
$p_\eta(z)$ consists of a complete contraction of products of the tensor
components of $z$ with copies of $\eta$ (and also $\epsilon$ in the
oriented case). Recalling the details of the restriction of tensor
representations to the orthogonal subgroup (Remark~\ref{rem_tens_orthog}),
the invariant extension $p(g,z)$ of $p_\eta(z)$ clearly constitutes the
same complete contraction of products of the tensor components of $z$,
but with every occurrence of $\eta$ replaced by either $g_{ab}$ (when
contracting two contravariant indices), $g^{ab}$ (when contracting two
covariant indices) or $\delta_a^b$ (when contracting a covariant and a
contravariant index). Respectively, a contraction with $\epsilon$ is
replaced by a contraction with $\varepsilon(g)$ with its indices
appropriately raised or lowered by $g$. Thus, we arrive at the desired
conclusion about the polynomiality of $p(g,z)$ in $g$, $g^{-1}$ (and
resp.~$\varepsilon(g)$).

\emph{3.} Recall the isomorphism $\mathcal{S}_{Z} \cong
\mathcal{I}_{Z,\mathbb{R}}$ (resp.~$\tilde{\mathcal{S}}_{Z} \cong
\tilde{\mathcal{I}}_{Z,\mathbb{R}}$) from point~2. Then, the desired
conclusion follows from Proposition~\ref{prop_luna_ext}, noting that
$O(1,n-1)$ (resp.~$SO(1,n-1)$) is a \emph{linearly reductive} Lie group
(and so is its complexification, cf.~Remark~\ref{rem_reductive}) and any
tensor representation (Definition~\ref{tensor_rep}) is obviously
rational. The finiteness of the number of generating invariant
polynomials $p_i$ ultimately follows from Hilbert's theorem
(Proposition~\ref{prop_hilbert}), which can obviously be chosen to be
homogeneous.

\emph{4.} It follows from Proposition~\ref{prop_equiv_inv} that any
equivariant $\tau \in \mathcal{E}_{Z,T}$
(resp.~$\tilde{\mathcal{E}}_{Z,T}$), can be written as a gradient
$\tau(g,z) = \left. \frac{\del}{\del t^*} \sigma(g,z,t^*)
\right|_{t^*=0}$, for some invariant $\sigma \in \mathcal{S}_{Z\times
T^*}$ (resp.~$\tilde{\mathcal{S}}_{Z\times T^*}$) that is linear in the
$t^*$ arguments. On the other hand, point~3 implies that
\begin{equation*}
	\sigma = [\Sigma]_{Z\times T^*}(p_1,\ldots,p_{N_Z}, Q_1, \ldots, Q_{N_{Z,T}})
\end{equation*}
is locally a smooth function of the invariants polynomial on $Z\times
T^*$, split into the $p_i$ that do not depend on the $T^*$, and the
$Q_i$ that depend on the $T^*$ only linearly. By combining the chain
rule with the notion of local dependence on polynomials
(Definition~\ref{def_loc_poly}), we get
\begin{align*}
	\tau(g,z)
	&= \left. \frac{\del}{\del t^*}
		[\Sigma]_Z(p_1,\ldots, p_{N_Z}, Q_1,\ldots Q_{N_Z}) \right|_{t^*=0} \\
	&= \sum_{j=1}^{N_{Z,T}} \left.
		\left[\frac{\del}{\del Q_j}\Sigma\right]
			(p_1,\ldots, p_{N_Z}, Q_1,\ldots, Q_{N_Z}) \right|_{Q_j=0}
		\left. \frac{\del Q_j}{\del t^*} \right|_{t^*=0} \\
	&= \sum_{j=1}^{N_{Z,T}} \sigma^j(g,z) q_j(g,z) ,
\end{align*}
with the obvious definitions for $\sigma^j$ and $q_j$.
This concludes the proof.
\end{proof}


\bibliographystyle{utphys-alpha-auth}
\bibliography{bibliography}

\end{document}